\newcommand {\rf} {\mathit{rank}}

\newcommand {\ent} {\mathrel{{\mathrel|\joinrel\sim}}}

\newcommand {\nott} {\lnot}

\newcommand {\perogni} {\forall}

\newcommand {\sx} {\langle}
\newcommand {\dx} {\rangle}

\newcommand {\appartiene} {\in}
\newcommand {\emme} {\mathcal{M}}
\newcommand {\enne} {\mathcal{N}}

\newcommand {\elle} {\mathcal{L}}
\newcommand {\modello} {\models}

\newcommand {\tc} {\mid}

\newcommand {\WW} {\mathcal{W}}

\newcommand{\tip}{{\bf T}}

\newcommand{\alc}{\mathcal{ALC}}

\newcommand{\alctmin}{\mathcal{ALC}+\tip_{min}}

\newcommand {\defincl} {\mathrel{{\scriptstyle \sim^{\!\!\!\!\! \sqsubset}}}}

\newcommand{\alctr}{\mathcal{ALC}+\tip_{\textsf{\tiny R}}}

\newcommand{\alctm}{\mathcal{ALC}^{\Pe}_{min}\tip}

\newcommand{\be}{\begin{enumerate}}
\newcommand{\ee}{\end{enumerate}}

\newcommand{\hide}[1]{}

\def \cases{\left \{\begin{array}{l}}
\def \endcases{\end{array}\right .}

\newcommand {\Pe} {{\bf P}}
\newcommand {\Ra} {{\bf R}}

\newcommand {\ri} {\rightarrow}
\newcommand {\Ri} {\Rightarrow}

\newcommand {\bes} {\begin{description}}
\newcommand{\ens} {\end{description}}
\newcommand {\la} {\langle}
\newcommand {\ra} {\rangle}

\newcommand {\beq} {\begin{quote}}
\newcommand {\enq} {\end{quote}}
\newcommand {\bit} {\begin{itemize}}
\newcommand {\enit} {\end{itemize}}

\newenvironment{pozz}{\color{black}}{\color{black}}

\documentclass{elsarticle}
\usepackage{times}
\usepackage{helvet}
\usepackage{courier}
\usepackage{graphicx}
\usepackage{latexsym}
\usepackage{graphicx}
\usepackage{color}
\usepackage{amssymb}
\usepackage{colortbl}
\usepackage{amsmath}
\usepackage{microtype}

\def \ri{\rightarrow}
\def \Ri{\Rightarrow}

\begin{document}

\begin{frontmatter}

\title{A reconstruction of  multipreference closure } 


 \author[unipmn]{Laura Giordano\corref{ca}}
\ead{laura.giordano@uniupo.it}

\author[unito]{Valentina Gliozzi\corref{ca}}
\ead{valentina.gliozzi@unito.it}

\cortext[ca]{Corresponding author}

\address[unipmn]{Universit\`a del Piemonte Orientale ``A. Avogadro'' \\
       DISIT, viale Teresa Michel, 11, 15121 Alessandria, Italy}
\address[unito]{Universit\`a degli Studi di Torino \\
       Center for Logic, Language and Cognition and Dipartimento di Informatica, \\
       C.So Svizzera, 185 - 10149 Torino, Italy}

 \begin{abstract} 
 The paper describes a preferential approach for dealing with exceptions in KLM preferential logics, based on the rational closure.  It is well known that rational closure  does not allow an independent handling of inheritance of different defeasible properties of concepts.  
 In this work, we consider an alternative closure construction, called  Multi Preference closure (MP-closure), which has been first considered  for reasoning with exceptions in DLs. We reconstruct the notion of MP-closure in the propositional case and show that it is a natural (weaker) variant of Lehmann's lexicographic closure,
which appears to be too bold in some cases.
The MP-closure defines a preferential consequence relation that, although weaker than  lexicographic closure, is stronger than Relevant Closure. 

\end{abstract}

\begin{keyword}
Nonmonotonic Reasoning \sep Preferential semantics \sep Rational Closure   \sep Knowledge Representation
\end{keyword}

\end{frontmatter}

\newtheorem{theorem}{Theorem}
\newdefinition{definition}{Definition}
\newdefinition{example}{Example}
\newtheorem{proposition}{Proposition}
\newtheorem{lemma}{Lemma}
\newtheorem{fact}{Fact}
\newtheorem{corollary}{Corollary}
\newproof{proof}{Proof}

\section{Introduction}
Kraus, Lehmann and Magidor  \cite{KrausLehmannMagidor:90} and Lehmann and Magidor  \cite{whatdoes}  investigate 
 the properties that a notion of plausible inference from a conditional knowledge base should satisfy (these are the KLM properties, for short). 
These properties led to the definition of  the notions of preferential and rational consequence relation,
as well as to the definition of the rational closure of a conditional knowledge base  \cite{whatdoes}.
Although not all non-monotonic formalisms in the literature  satisfy KLM properties, and although the adequacy of these properties 
has been  (and still is)  subject of debate, starting from the work by Pearl and Geffner \cite{GeffnerAIJ1992}, by Benferhat et al. \cite{Benferhat2000}, and more recently by 
Bonatti and Sauro \cite{BonattiSauro17},   Koutras et al. \cite{KoutrasKR18},  and Casini et al. \cite{CasiniJelia2019}, 
 the rational closure construction 
has been considered for defeasible reasoning in description logics  \cite{casinistraccia2010,CasiniJAIR2013,CasiniDL2013,AIJ15,CasiniISWC15}, which are the formalisms at the basis of OWL ontologies \cite{OWL}.

While the rational closure  provides a simple and efficient approach for reasoning with exceptions 
(as its construction requires a polynomial number of inference steps in the underlying logical formalism),
it is well known that ``it does not provide for inheritance of generic properties to exceptional subclasses" \cite{Lehmann95}.
This problem was called by Pearl \cite{PearlTARK90} the ``blockage of property inheritance" problem, 
and it is an instance of the ``drowning problem"  \cite{BenferhatIJCAI93}.
\color{black} Roughly speaking, the problem is that property inheritance from classes to subclasses is not guaranteed: if the subclass is exceptional with respect to the superclass for a given property, it does not inherit from that superclass any other property. For instance, in the well-known penguin example, penguins are exceptional birds with respect to the property of flying, hence in rational closure they do not inherit any other property of birds (whereas we would like to conclude that penguins inherit from birds the property of having wings) \cite{PearlTARK90} .
\normalcolor

To overcome this weakness of rational closure, Lehmann  \cite{Lehmann95} introduced the notion of  lexicographic closure, as a ``uniform way of constructing a rational superset of the rational closure" 
thus strengthening the rational closure but still defining a rational consequence relation.

In this paper  we propose an alternative construction to rational closure which does not incur into the blockage of inheritance problem: MP-closure.
The MP-closure is a natural (weaker) variant of Lehmann's lexicographic closure, 
which simply uses a different lexicographic ordering.
{\color{black} Our investigation of the MP-closure construction is motivated by the fact that in some cases the lexicographic  closure appears to be too bold and the weaker MP-closure better suited (see for instance Example \ref{example-differenza_MC_LC} in Section  \ref{sez:Lex-closure}). 
A similar motivation is at the basis of the development of other refinements of rational closure, such as the Relevant Closure,  
a closure recently proposed by Casini et al.  \cite{Casini2014}  in the context of description logics, as a weaker alternative to  lexicographic closure.}


The MP-closure construction departs from  lexicographic closure in the choice
that, in case of contradictory defaults with the same rank, 
one tries to satisfy as many defaults as possible 
(where the {\em number} of defaults matters, rather than the defaults themselves). 
Abandoning this choice, 
and following 
an alternative route that was also considered but not explored by Lehmann \cite{Lehmann95},
  leads to a construction  
which defines a preferential consequence relation rather than a rational one,
and to a more cautious notion of entailment. 

We believe that MP-closure defines  {\color{black} a simple and  interesting} notion of entailment, 
\color{black} which is strongly related with other proposals in the literature,  such as: {\em system ARS} \cite{IsbernerRitterskamp2010}, that refines and improves {\em system Z} \cite{PearlTARK90}, Brewka's Basic Preference Descriptions \cite{Brewka04}, the Relevant Closure \cite{Casini2014}, and Conditional Entailment \cite{Geffner1992}.
MP-closure may be reasonable in specific contexts, for instance,   
when reasoning about multiple inheritance in ontologies,
where the lexicographic closure appears to be too bold. 
%
In particular, as we will see below, in a situation in which typical $A$s are $P$s and typical $B$s are $\neg P$s, what do we want to conclude about typical $A \wedge B$s?
MP-closure does not conclude $P$ nor $\neg P$, while the lexicographic closure might conclude either $P$ or $\neg P$ depending on the syntactic 
formulation of the knowledge base 
giving rise to
some hazardous inferences. 
%
We will come back on this point with some examples in Sections  \ref{sez:Lex-closure} and \ref{sez:BPsem}.

Following the pattern by Lehmann \cite{Lehmann95}, in this paper we present a characterization of MP-closure both in terms of maxiconsistent sets and 
in terms of preferential models. 
MP-closure  defines a preferential consequence relation but not a rational consequence relation.
It is stronger than  rational closure but weaker than lexicographic closure.
Considering the natural (semantic) transformation, proposed by Lehmann and Magidor \cite{whatdoes},  to map a preferential consequence relation into a rational one, we apply 
it to  MP-closure and define a consequence relation which is a rational superset of  MP-closure. 
{\color{black} This is another uniform way of constructing a rational superset of the rational closure, with respect to the lexicographic closure.} We investigate its relationships with the lexicographic closure  and show that they are incomparable (none of them is weaker than the other).
%
%
Interestingly enough, 
we consider the operator 
used for defining 
a rational extension of  MP-closure,
and show that the minimal fixed-points of such an operator lead
to the multipreference  semantics proposed by Gliozzi \cite{GliozziAIIA2016}.

We conclude the paper 
by comparing the MP-closure 
 with
{\em system ARS} 
\cite{IsbernerRitterskamp2010}, 
with Brewka's Basic Preference Descriptions \cite{Brewka04}, with Geffner's conditional entailment \cite{Geffner1992}, 
with  Relevant Closure  \cite{Casini2014}, and with other preferential formalisms in the literature.
 In particular, we show that both the Basic and the Minimal Relevant Closure are weaker than MP-closure.

%

\normalcolor

The paper is organized as follows.
In Section \ref{sez:RC} we recall the definition of  rational closure and its semantics and,
in Section 
\ref{sez:Lex-closure}, the definition of the lexicographic closure, and discuss some examples which motivate the interest in investigating a weaker notion of closure. 
In Section \ref{sez:BPsem} we reformulate the MP-closure construction in the propositional setting in terms of maxiconsistent sets and, then, we study its model-theoretic semantics, its properties  and its relations with the lexicographic closure. 
In Section \ref{sec:rational_relation} a rational consequence  relation is defined, which is a superset of MP-closure 
and it is compared with  lexicographic closure. The relationships with 
 multipreference semantics and other strongly related constructions are investigated in Section \ref{sec:further_issues}. Section \ref{sec:conclu} concludes the paper.




\vspace{-0.1cm}
\section{The rational closure}\label{sez:RC}
\vspace{-0.2 cm}
In this section we recall the definition of rational closure by Lehmann and Magidor \cite{whatdoes} and its semantics,
which we will exploit it to define the semantics of MP-closure. 

Let language $\elle$ be the propositional language
defined from a set of propositional variables $\mathit{ATM}$ and 
Boolean connectives.
A {\em conditional assertion over $\elle$} has the form $A \ent B$, where $A, B \in \elle$.
\color{black}
A {\em knowledge base} $K$ is a set of conditional assertions $A \ent B$ over $\elle$. \normalcolor
In the following, we will restrict our attention to  {\em finite knowledge bases}
over a finite language.

{\color{black}  Kraus, Lehmann and Magidor \cite{KrausLehmannMagidor:90} define the family of preferential models and study related entailment relations.
Their work led to a classification of nonmonotonic consequence relations, 
and to the definition of the related KLM properties,
which have been widely accepted as the conservative core of default reasoning.
Lehmann and Magidor \cite{whatdoes} then introduce a family of stronger models, namely ranked models, and a stronger notion of entailment, called rational entailment.
As both preferential and ranked models are central to our work, in the following we will recall their definitions as well as the representation results relating preferential (ranked) models with KLM properties.}

The  semantics of conditional KBs in a preferential structure is defined by considering a set of worlds $\WW$ equipped
with a preference relation $<$.  Intuitively, the
meaning of $x < y$ is that $x$ is more typical/more normal/less
exceptional than $y$.  As we will see in Definition \ref{semantica_rational} below, a conditional $A\ent B$ is  satisfied
 in a preferential model if $B$ holds in all the most normal worlds satisfying $A$, i.e., in all $<$-minimal worlds satisfying $A$.

\begin{definition}[Preferential models and ranked models]\label{semantica_rational}
A $\mathit{preferential}$ model  is a triple $\emme= \sx \WW, <, v \dx$ where:
\begin{itemize} 
\item $\WW$ is a non-empty set of worlds;
\item $<$ is 
an irreflexive, transitive  relation 
on $\WW$ satisfying the Smoothness condition defined below;
\item  $v$ is a function $v: \WW \longrightarrow$ $2^{\mathit{ATM}}$, which
    assigns to every world $w$ the set of atoms holding in
    that world. If $F$ is a Boolean combination of formulas, its truth
conditions ($\mathcal{M}$, $w \models F$) are defined as for
propositional logic. Let $A$ be a propositional formula; we define
$Min_{<}^{\emme}(A)=\{w\appartiene \WW
     \tc \emme$$, w \modello A $ and $\perogni w'$, $w' < w$ implies $\emme,$$
w' \not\modello A\}$. Moreover: 
\begin{center}
{\color{black} $\emme \modello A \ent B$} iff for all $w'$, if $w' \appartiene
  Min_{<}^{\emme}(A)$ then $\emme,$$ w' \modello B$.
\end{center}
 \end{itemize}

\noindent At this point we can define the {\em Smoothness condition}: if $\emme,$$ w
\modello A$, then either $w \in Min_<^{\emme}(A)$ or there is $w' \in Min_<^{\emme}(A)$
such that $w' < w$ \footnote{\color{black} As observed by Lehmann and Magidor \cite{whatdoes}, the smoothness condition is a technical condition and 
when the language ${\cal L}$ is logically finite, we could limit ourselves to finite models and forget the smoothness condition.}.

A $\mathit{ranked}$ model is a preferential model  $\emme= \sx \WW, <, v \dx$ for which the relation $<$ is modular: for all $x, y, z$, if
$x < y$ then $x < z$ or $z < y$.

\end{definition}
%
{\color{black} The satisfiability of a conditional formula $A \ent B$ in a model is not defined with respect to a specific world, and we write
 $\emme \modello A \ent B$ to mean: $\emme$ {\em satisfies} $A \ent B$.}
 $A \ent B$ is {\em satisfiable} in the preferential (rational) semantics, if there is some preferential (ranked) model $\emme$ such that  $\emme \modello A \ent B$.
 Given a set  $K$ of conditional assertions over $\elle$ and a model $\emme$$=\sx \WW$$,
<, v \dx$, we say that $\emme$ is a model of $K$, written $\emme$
$\models K$, if for every $A \ent B \in K$, 
$\emme \models A \ent B$.
$K$ {\em preferentially entails} a
conditional assertion $A \ent B$, written $K \models_{\Pe} A \ent B$ if $A \ent B$ is satisfied in all 
preferential models of $K$.
 $K$ {\em rationally entails} a
conditional assertion $A \ent B$, written $K \models_{\Ra} A \ent B$ if $A \ent B$ is satisfied in all
ranked models of $K$.

As a consequence of Theorems 6.8 and 6.9 in \cite{toclKLM}, 
if a set of formulas $K$ is satisfiable in a ranked model, then it is satisfiable in a {\em finite} ranked model. 
In the following, we will restrict our consideration to ranked models with a finite set of worlds.
\color{black} This is in agreement with \cite{BoothParis98} where, for knowledge bases over a finite set of propositional variables, finite ranked 
models are considered, which are represented as finite sequences $t_1,t_2, \ldots, t_n$ of sets of atoms, each atom representing a propositional interpretation (a world).


Kraus, Lehmann and Magidor \cite{KrausLehmannMagidor:90} studied the properties of several families of nonmonotonic {\em consequence relations},
namely, well-behaved sets of conditional assertions (i.e., binary relations $\ent$ on $\cal L$).
In particular, they introduced a notion of {\em preferential consequence relation} as a consequence relation satisfying the following properties, expressed in the form of inference rules (also called KLM postulates):
\vspace{-0.3cm}
\begin{tabbing}
$\mathit{(Left \; Logical\; Equivalence)}$ \= xxxxxxxxxx \= \kill \\
$\mathit{(Left \; Logical\; Equivalence)}  ~ \mbox{ If } \models \mathit{A \equiv B}  \mbox{ and }   \mathit{A \ent C} \mbox{ then }    \mathit{B \ent  C} $ \\
$\mathit{(Right\; Weakening)}  ~ \mbox{ If } \models  \mathit{B \rightarrow  C}  \mbox{ and }       \mathit{A \ent B}    \mbox{ then }   \mathit{ A \ent C} $ \\
$\mathit{(Reflexivity)} ~  \mathit{A \ent A } $ \\
$\mathit{(And)}  ~ \mbox{ If }   \mathit{A \ent B}  \mbox{ and }   \mathit{A \ent C} \mbox{ then }  \mathit{A \ent  B \wedge C} $ \\
$\mathit{(Or)}  ~ \mbox{ If }  \mathit{A \ent  C}  \mbox{ and }   \mathit{B \ent C} \mbox{ then }    \mathit{A \vee B \ent C} $ \\
$\mathit{(Cautious \; Monotonicity)}  ~ \mbox{ If } \mathit{A \ent B}  \mbox{ and }   \mathit{A \ent C} \mbox{ then }  \mathit{A \wedge B \ent C} $ 
\end{tabbing}
Kraus, Lehmann and Magidor also proved a characterization result for preferential consequence relations, stating that 
a binary relation $\ent$ on $\cal L$ is a preferential consequence relation if and only if it is the consequence relation defined by some preferential model, 
where the consequence relation $\ent_\emme$ defined by a preferential model $\emme= \sx \WW, <, v \dx$ is the set of pairs $(A,B)$ such that 
for all $w'$, if $w' \appartiene
  Min_{<}^{\emme}(A)$, then $\emme,$$ w' \modello B$ (i.e., $A \ent B$  is satisfied in $\emme$).

The nonmonotonic system defined by the properties above is also known as system $\Pe$, and corresponds to a  flat (not nested) fragment of some conditional logic \cite{Nute80}. 
System $\Pe$ was proposed by Pearl \cite{Pearl:88}  as the {\em conservative core} of a nonmonotonic reasoning system, 
and the notions of {\em p-entailment} by Adams \cite{Adams:75}, {\em $\epsilon$-entailment} by Pearl \cite{Pearl:88} and {\em 0-entailment} in the {\em System Z} \cite{PearlTARK90} 
are all equivalent to preferential entailment. In particular, Lehmann and Magidor \cite{whatdoes} 
established a link between preferential entailment and Adams' probabilistic entailment.



Kraus, Lehmann and Magidor \cite{KrausLehmannMagidor:90}  studied representation results for systems weaker than $\Pe$.
Such systems, same as for system  $\Pe$ itself,  lack a property that seems to be desirable, namely the property of rational monotonicity:
\begin{tabbing}
$\mathit{(Rational \; Monotonicity)}  ~ \mbox{ If } \mathit{A \ent C}  \mbox{ and }     \mathit{A \not \ent \neg B} \mbox{ then }    \mathit{A \wedge B \ent C} $ 
\end{tabbing}
i.e., if $ A \ent C$ belongs to the consequence relation and $A \ent \neg B$ does not, then 
$A \wedge B \ent  C$ must belong as well to the consequence relation.
Lehmann and Magidor \cite{whatdoes}  considered 
{\em rational consequence relations}, i.e., preferential consequence relations that satisfy Rational Monotonicity,  and  proved  a semantic characterization result for rational consequence relations based on ranked models: each rational consequence relation is the consequence relation defined by some ranked model, and vice-versa.

In ranked models, each world can be associated with a rank in a natural way. \normalcolor

\begin{definition}[Rank $k_{\emme}(w)$ of a world in $\emme$]\label{definition_rank_prop} 
Given  a (finite) ranked model $\emme$ $=\sx\WW$$, <, v\dx$, the rank $k_{\emme}$  of a world $w \in \WW$, written $k_{\emme}(w)$,  is the length of the maximal chain $w_0 < \dots < w$ from $w$
to a minimal $w_0$ (i.e., there is no ${w'}$ such that  ${w'} < w_0$).\footnote{A chain $w_0 < w_1 < \ldots < w_n$ is maximal if there is no element 
$w'$ such that for some $i=0, \ldots, n-1$ it  holds $w_i < w' < w_{i+1}$.}
\end{definition}
Hence, the preference relation $<$ of a ranked model $\emme$ defines a ranking function $k_{\emme}: \WW \longrightarrow \mathbb{N}$
(this is just a special case of the general result in \cite{whatdoes} where there is no restriction to finite models).
%
%
Notice also that Definition \ref{definition_rank_prop}  makes sense even if the relation $<$ is not modular
and that, for a modular relation on a finite set, all maximal chains 
from an element $w$ to a minimal $w_0$ have the same length.

The previous definition  defines from $<$ a   ranking function $k_{\emme}: \WW \longrightarrow \mathbb{N}$. The opposite is also possible 
and $<$ can be defined from a ranking function $k_{\emme}$  by letting $x < y$ if and only if $k_{\emme}(x) < k_{\emme}(y)$ (this is similarly stated in \cite{whatdoes}, where a ranking function $r$ over a possibly infinite set is considered, since there is no restriction to finite models). 

The rank of a propositional formula $F$ in a model $\emme$ depends on the rank of the worlds satisfying the formula.
\begin{definition}[Rank of a formula in a model]\label{definition_rank_formula_prop}
The rank $k_{\emme}(F)$ of a formula $F \in {\cal L}$ in a model $\emme$ is $i = min\{k_{\emme}(w):
\emme,$ $w \models F \}$. If there is no $w$ such that $\emme,$ $w \models F$, then
$F$ has no rank in $\emme$.
\end{definition}

Lehmann and Magidor  \cite{whatdoes} proved that, for a  knowledge base $K$ which  is a
set of positive conditional assertions,  i.e., assertions with form $A \ent B$,
rational entailment is equivalent to preferential entailment.
Indeed, the set of conditional assertions satisfied in all ranked models of $K$
exactly coincides with the  set of conditional assertions satisfied in all preferential models of $K$.
Also, rational entailment does not define a rational consequence relation, i.e., a consequence relation satisfying the property of Rational Monotonicity.

In order to strengthen  rational entailment and to define a rational consequence relation, Lehmann and Magidor \cite{whatdoes} 
introduce the notion of {\em rational closure}, which 
can be seen as the ``minimal", in some sense (see below)
rational consequence completing a set of conditionals.
Let us recall the definition of the rational closure.

\begin{definition}[Exceptionality of formulas] 
Let $K$ be a knowledge base (i.e., a finite set of positive
conditional assertions) and $A$ a propositional formula. $A$ is
said to be {\em exceptional} for $K$ if and only if $K \models_{\Pe} \top \ent
\neg A$.  
A conditional formula $A \ent B$ is exceptional for $K$ if its antecedent $A$ is exceptional for $K$.  The set of conditional formulas of $K$ which are exceptional for $K$ will be denoted
as $E(K)$.
\end{definition}

\noindent   It is possible to define a non-increasing sequence of subsets of
$K$, $C_0 \supseteq C_1, C_1 \supseteq C_2, \dots$ by letting $C_0 = K$ and, for
$i>0$, $C_i$ the set of conditionals of $C_{i-1}$ exceptional for $C_{i-1}$, i.e.,  $C_i = E(C_{i-1})$. 
Observe that, being $K$ finite, there is
an $n\geq 0$ such that $C_n = \emptyset$ or for all $m> n, C_m = C_n$. 
\color{black} The sets $C_i$  are used to define the rank of a formula, as  in the next definition. Notice that if there is an $m$ such that $C_m = C_{m+1}$, then for all $k > m$, it will hold that $C_m = C_{k}$ (indeed $E(C_{m}) = E(C_{m+1}) = \dots =E(C_{k}$)). 
\normalcolor

\begin{definition}[Rank of a formula] \label{Def:Rank of a
formula}A propositional formula $A$ has {\em rank} $i$ (for $K$)\color{black}, written $rank(A)=i$, \normalcolor
if and only if $i$ is the least natural number for which $A$ is
not exceptional for $C_{i}$. {\color{black} If $A$ is exceptional for all
$C_{i}$, then 
we let $rank(A)=\infty$.}
\end{definition}
A conditional $A \ent B$ has rank equal to $\rf(A)$, and $C_i \setminus C_{i-1}$ is the set of conditionals (defaults) in $K$ having rank $i$.
 When $\rf(A)=\infty$, we say that the conditional $A \ent B$ has no rank.

\begin{example} \label{example-Student}
Let $K$ be the knowledge base containing the conditionals:
\begin{quote}
 1. $\mathit{ Student \ent \neg Pay\_Taxes}$\\
 2. $\mathit{Student \ent  Young}$\\
 3. $\mathit{ Employee \wedge Student \ent Pay\_Taxes}$
\end{quote}
stating that normally students do not pay taxes and are young, while employed students  normally pay taxes.
From the construction above, $C_0 = K$.
Furthermore, the set $C_1=E(K)$ of conditional formulas which are exceptional for $K$ only contains the default $ \mathit{ Employee \wedge Student  \ent Pay\_Taxes}$.
In fact, $\mathit{Student}$ is not exceptional for $K$, i.e.,  $K \not \models_{\Pe} \top \ent \neg \mathit{Student}$ (as there is some preferential model of $K$ falsifying $ \top \ent \neg \mathit{Student}$).
Hence, $\rf(\mathit{Student})=0$.
On the contrary, $\mathit{Employee \wedge Student }$ is exceptional for $K$, i.e., $K \models_{\Pe} \top \ent \neg (\mathit{Employee \wedge Student })$ 
(in all the preferential models of $K$, there cannot be a world with rank $0$ that satisfies $\mathit{Employee \wedge Student }$).
Therefore,
\begin{quote}
$C_0 = K$;

{ $C_1 = \{  \mathit{ Employee \wedge Student  \ent Pay\_Taxes} \}$}.
\end{quote}
In turn, $C_2=E(C_1)= \emptyset$ , as $\mathit{Employee \wedge Student }$ is not exceptional for $C_1$, i.e.,  $C_1 \not \models_{\Pe} \top \ent \neg (\mathit{Employee \wedge Student })$ (there is a preferential model satisfying $C_1$,  in which $\mathit{Employee \wedge Student}$ is satisfied in some world with rank $0$). {\color{black} As $\mathit{Employee \wedge}$ \linebreak $\mathit{Student }$ is exceptional for $C_0$, but not for $C_1$, in the rational closure  $\rf(\mathit{Employee \wedge}$ \linebreak $\mathit{ Student})=1$.}
%
Thus, the third conditional describing the properties of employed students has rank $1$ and is more specific than the conditionals describing the properties of  students, which have rank $0$.
%

\end{example}
Rational closure builds on the notion of exceptionality. Roughly speaking a conditional $A \ent B$ is in the rational closure of $K$ if 
$A \wedge B$ is less exceptional than $A \wedge \neg B$. 
 In next definition, we recall the rational closure construction by Lehmann and Magidor  in  \cite{whatdoes}, but limiting our consideration  to finite knowledge bases.
\begin{definition}[Rational closure \cite{whatdoes}] \label{def:rational closureDL}
Let $K$ be a (finite) knowledge base. The
rational closure  of $K$ is defined as: 
\begin{align*}
    \overline{K}= & \{A\ent B  \tc \mbox{ either } \ \rf(A) < \rf(A \wedge \nott B) \mbox{ or }  \rf(A)=\infty\} 
\end{align*}
where $A$ and $B$ are propositions in the language of $K$.
\end{definition}
Referring to Example \ref{example-Student}, $\mathit{Student \wedge Italian \ent  \neg Pay\_Taxes}$ is in the rational closure of 
$K$,
as $\rf(\mathit{Student \wedge Italian})=0 < \rf( \mathit{Student \wedge Italian \wedge Pay\_Taxes})$ $=1$. 
Similarly,
 $\mathit{Employee \wedge Student \wedge Italian}$ $\mathit{ \ent   Pay\_Taxes}$ is in $\overline{K}$.

It is worth noticing that a strongly related construction has been proposed by Pearl \cite{PearlTARK90} with his notion of 1-entailment, originating from a probabilistic interpretation of conditionals within the well-established System Z.


For what concerns the semantics of the rational closure,
Lehmann and Magidor   
have developed a model-theoretic semantics for the rational closure \cite{whatdoes}, 
showing that, given a well-founded preferential consequence relation $P$\footnote{\color{black} A preferential consequence relation is well-founded \cite{whatdoes} 
if the strict ordering relation $<$ on formulas it defines is  well-founded. Kraus, Lehmann and Magidor \cite{KrausLehmannMagidor:90} have shown that any preferential consequence relation defines a strict ordering on formulas by: $\alpha < \beta$ iff $\alpha \vee \beta \ent \alpha$ and $\alpha \vee \beta \not \ent \beta$.} and any well-founded preferential model $W$ defining it, a ranked model can be constructed 
by ``letting all the states of $W$ sink as low as they can respecting the order of $W$". The resulting ranked model provides a characterization of the rational closure of $P$.
%
 For 
knowledge bases over a finite language, a characterization of the rational closure has been developed by Booth and Paris 
by a simple model-theoretic construction  \cite{BoothParis98}, 
that has been further extended by the same authors to provide a semantic characterization of the rational closure of a knowledge base also including  negative assertions (i.e., assertions of the form $A \not \ent B$).
Furthermore, Pearl  has introduced a notion of minimal ranking function,  
showing the existence of a unique minimal Z ranking of a consistent knowledge base \cite{PearlTARK90}. 
 In the following, we will report the semantic characterization of  rational closure
as formulated  by Giordano et al. \cite{AIJ15}, given in terms of {\em minimal canonical} ranked models. 
In such models,  the rank of worlds is minimized to make each world as normal as possible and, 
roughly speaking, the needed propositional valuations (or worlds) are all taken into consideration (by the canonical world requirement).
While two alternative minimization criteria are considered therein,
here we will just recall  the notion of  {\em fixed interpretations minimal semantics}
($\mathit{FIMS}$), 
 %
%
%
where only models with the same set of worlds $\WW$ and valuation function $v$ are compared. 
\normalcolor
\begin{definition}[Minimal ranked models]\label{Preference between models in case of fixed valuation}
Let $\emme = $$\langle \WW$$, <, v \rangle$ and $\emme' =
\langle \WW'$$, <', v' \rangle$ be two ranked models.
 $\emme$ is preferred to
$\emme'$ with respect to the fixed interpretations minimal
semantics (and we write $\emme <_{\mathit{FIMS}} \emme'$) if: \  $\WW = \WW'$, $v = v'$ and
\begin{quote}
 for all $x \in \WW$, $ k_{\emme}(x) \leq k_{\emme'}(x)$ and \\
there exists $x'\in \WW$ such that $ k_{\emme}(x') < k_{\emme'}(x')$.\footnote{The second condition is needed as we are defining a strict partial order among models, so to give preference to those models assigning lower ranks to worlds.} 
\end{quote}
 Given a knowledge base $K$, we say that
$\emme$ is a  minimal model of $K$ with respect to $<_{\mathit{FIMS}}$ if $\emme$ is a model of $K$ and there is no
$\emme'$ such that $\emme'$ is a model of $K$ and $\emme' <_{\mathit{FIMS}} \emme$. 
\end{definition}
In \cite{AIJ15} it was also shown that a notion of canonical model is needed
when reasoning about the (relative) rank of the propositions in a model of $K$: it is important to have them true in some world of the model, whenever they are consistent with the knowledge base. 

Given a knowledge base $K$ and a query $Q$,
let ${\mathit{ATM}_{K,Q}}$ be the set of all the propositional variables of $\mathit{ATM}$ occurring in $K$ or in the query $Q$,
and let $\elle_{K,Q}$ be the restriction of the language  ${\cal L}$ to the propositional variables in ${\mathit{ATM}_{K,Q}}$.

A truth assignment 
\normalcolor
$v_0:  \mathit{ATM_{K,Q}} \longrightarrow
\{true, false\}$ 
\normalcolor
is $\mathit{compatible}$ with $K$, if there is no
propositional formula $A \in \elle_{K,Q}$ such that $v_0(A) = true$ and
$K\models A \ent \bot$
(where $v_0$ is extended as usual to arbitrary propositional formulas over the language  $\elle_{K,Q}$). 
The following definitions of canonical model and minimal canonical model are reported from \cite{AIJ15}.


\begin{definition}[Canonical models]\label{canonical_model} 
A model $\emme=$$\sx \WW$$, <, v \dx$ satisfying a knowledge base $K$
is  {\em canonical} if 
for each $v_0$ compatible with $K$,  there
exists a world $w$ in $\WW$ such that, for all propositional
formulas  $B \in \elle_{K,Q}$, $\emme,$$w \models B$ if and only if $v_0(B) = true$.
\end{definition}

\begin{definition}[Minimal canonical ranked models]\label{def-minimal-canonical-model}
$\emme$ is a minimal canonical ranked model of $K$, 
if it is a canonical ranked model of $K$ and it is minimal with respect to $<_{\mathit{FIMS}}$ (see Definition \ref{Preference between models in case of fixed
valuation}) among the canonical ranked models of $K$.
\end{definition}
We define a notion of minimal entailment w.r.t. minimal canonical ranked models of $K$.
\begin{definition}[Minimal entailment]\label{def-minimal-entailment}
 $K$ {\em minimally entails} a formula $F$, and we
write $K \models_{min} F$, if $F$ is true in all the minimal canonical ranked models
of $K$.
\end{definition}
It has been shown that, for any satisfiable knowledge base, a finite minimal canonical ranked model exists (\cite{AIJ15}, Theorem 1),
and that minimal canonical ranked models are an adequate semantic counterpart of rational closure.
 Let $Min_{RC}(K)$ be the set of all the minimal canonical ranked models of $K$.
The correspondence between minimal canonical ranked models and rational closure is established by the following  theorem. 
\begin{theorem}[\cite{AIJ15}]\label{rat_closure_modelli_minimali}
Let $K$ be a knowledge base and $\emme \in Min_{RC}(K)$ be a minimal canonical ranked model of
$K$. For all conditionals $A \ent B$ over ${\cal L}$:
\begin{center}
 $\emme$ $\models A \ent B$ if and only if $A \ent B \in \overline{K}$,
\end{center}
where $ \overline{K}$ is the rational closure of $K$.
\end{theorem}
Furthermore, when $\rf(A)$ is finite, the rank $k_\emme(A)$ of a proposition $A$ in any minimal canonical ranked model of $K$ is equal to the rank $\rf(A)$ assigned by the rational closure construction. Otherwise, $\rf(A)= \infty$ and  proposition $A$ is not satisfiable in any ranked model of $K$ (in any ranked model of $K$, $A$ has no rank).

Observe that, by Theorem \ref{rat_closure_modelli_minimali}, the set of conditionals minimally entailed from $K$ coincide with the set of conditionals true in any (arbitrarily chosen) minimal canonical ranked model $\emme$ of $K$. 
In the following, we will restrict our consideration to the {\em finite} minimal canonical models of the knowledge base $K$ (which, as said above, always exist when $K$ is consistent), and we denote their set by $Min_{RC}(K)$. 
As mentioned above, Booth and Paris \cite{BoothParis98} developed canonical model constructions for the rational closure and  Pearl \cite{PearlTARK90}  proved the existence of a unique minimal Z ranking for a consistent (finite) knowledge base.
 As a difference, in the representation theorem above, we 
consider the set $ Min_{RC}(K)$ of all the (finite) minimal canonical models of a knowledge base $K$.
Although the models in $Min_{RC}(K)$ may differ as concerns the worlds in $\WW$ and their propositional valuations, 
as a consequence of Theorem 1 above, 
all of them must satisfy exactly the same positive conditionals $A \ent B$ over ${\cal L}$, i.e., the conditionals belonging to the rational closure of $K$, and falsify all the others. 
This means that 
the ranked models $\emme \in Min_{RC}(K)$ all define the same rational consequence relation,
corresponding to  rational closure, as well as to the set of conditional assertions $A \ent B$ which are minimally entailed by $K$, according to Definition \ref{def-minimal-entailment}.\footnote{ The situation is not dissimilar from the one considered by  Lehmann and Magidor in  \cite{whatdoes}, where, for a well-founded preferential relation $P$, a ranked model defining the rational closure of $P$ is built starting from {\em any} preferential model defining $P$. }


Given the above mentioned canonical model constructions, in the following when needed we will feel free to refer to any (arbitrarily chosen)  model in $Min_{RC}(K)$ as {\em the} minimal canonical model of  rational closure. 
%


\begin{example} \label{es:modello_canonico}
Let us consider again the knowledge base $K$ in Example \ref{example-Student}.
We have seen that in the rational closure of $K$ conditionals  $\mathit{Student \ent \neg Pay\_Taxes}$, and  $\mathit{Student \ent}$ $\mathit{ Young}$ have rank $0$,
while  $\mathit{Employee \wedge Student \ent  Pay\_Taxes}$ has rank $1$.

Let us represent a world as the set of the propositional variables it satisfies (e.g.,  $\{ s, y\}$ represents the world in which $\mathit{Student}$ ($s$) and $\mathit{Young}$ ($y$) are true, 
while  $\mathit{Em}$- $\mathit{ployee}$ ($e$) and $\mathit{Pay\_Taxes}$ ($p$) are false).
We can describe a ranked model partitioning the worlds according to their rank. 
The following is a minimal canonical model for $K$ (let us call it $\emme$):

Rank $0$: \ $\emptyset$, $\{ p\}$, $\{ y \}$, $\{ p, y \}$,  $\{ e\}$, $\{ e,y\}$, $\{ e,p\}$, $\{ e,y,p\}$, $\{ s,y\}$

Rank $1$: \   $\{ s,e,p\}$, $\{ s,e,p,y\}$,  $\{ s\}$, $\{ s,p\}$, $\{ s,p,y\}$

Rank $2$: \   $\{ s,e\}$, $\{ s,e,y\}$

\noindent
The worlds with rank $0$ do not violate any conditional in $K$;  those with rank $1$ violate conditionals with rank $0$; and
those with rank $2$ violate conditionals with rank $1$.

In a minimal model each world has the least rank it may have.
For instance, if  $\{ s,p,y\}$ were assigned rank $2$, rather than $1$, the resulting interpretation would still be a ranked model of $K$,
but not a minimal one.
To get a minimal model from a given ranked model of $K$, one lets  ``worlds  sink as low as they can" \cite{whatdoes}, provided the resulting structure still satisfies $K$.

To establish whether conditional $\mathit{Employee \wedge Student } \ent \mathit{ Young}$ is satisfied in $\emme$,
one has to consider all the minimal worlds satisfying the antecedent $\mathit{Employee \wedge}$ $\mathit{ Student }$ and verify that $\mathit{ Young}$ holds in all of them.
As the minimal worlds satisfying $\mathit{Employee \wedge Student }$ are $\{ s,e,p\}$, $\{ s,e,p,y\}$, and in the first one $\mathit{ Young}$ is false,
the conditional $\mathit{Employee \wedge Student \ent Young}$ is not satisfied in $\emme$. 
This is enough to conclude that the conditional is not minimally entailed from $K$. 
Similarly, conditional $\mathit{Employee \wedge Student } \ent \mathit{ \neg Young}$ is not satisfied in $\emme$ and is not minimally entailed from $K$.
This is in agreement with the fact that both these conditionals are not in the rational closure of $K$. 

Let us observe that, in the general case, a ranked model may contain more than one world with a given propositional valuation, as the valuation function $v$ in the model may assign the same propositional interpretation to more than one world. However, worlds with the same valuation will have the same rank in all minimal canonical models.
\end{example}
\normalcolor
In the last example, nothing can be concluded about the typical employed students being young or not.
Employed students do not ``inherit" any of the more general defeasible properties of students, not even the property that students are normally young. 
Indeed, the rational closure ``does not provide for inheritance of generic properties to exceptional subclasses" \cite{Lehmann95}.

{\color{black} In particular, the rational closure does not satisfy the desirable property called by Lehmann  {\em the presumption of independence}:
``even if typicality is lost with respect to one consequent, we may still presume typicality with respect to another, unless there is reason to the contrary" \cite{Lehmann95}.}
This and other desirable properties led to the definition of the  {\em lexicographic closure}  as a ``uniform way of constructing a rational superset of the rational closure" \cite{Lehmann95},
thus strengthening the rational closure but still providing a rational consequence relation.

 \section{From  Lexicographic closure to  MP-closure} \label{sez:Lex-closure}
 
To overcome the weakness of rational closure, Lehmann introduced the notion of  {\em lexicographic closure} \cite{Lehmann95}, which strengthens the rational closure by allowing, roughly speaking, a class  to inherit as many as possible of the defeasible properties of more general classes,  
giving preference to the more specific properties.
In the example above, the property of students being young should be inherited by employed students, as it is consistent with all other default properties of  employed students (i.e., with default 3) and, by ``presumption of independence" \cite{Lehmann95}, 
even if employed students do not inherit from students the typical property of not paying taxes, they can still inherit other typical properties of students such as being young.

%

Let us recall the definition of the lexicographic closure from  \cite{Lehmann95}.

Let $K$ be a finite knowledge base (a set of defaults) and let $k$ be the {\em order} of $K$, i.e., the least finite $i$ such that $C_i-C_{i+1}=\emptyset$
(notice that $C_i-C_{i+1}$ is the set of defaults with rank $i$).
The order $k$ is such that there is no default in $K$ with a finite rank equal to $k$ or higher than $k$ but, if $k>0$, there is at least one default with rank $k-1$. 

In order to compare alternative subsets of defaults in $K$, Lehmann introduces a notion of {\em seriousness ordering} $\prec$ among subsets of defaults in $K$.

\begin{definition}[Seriousness ordering \cite{Lehmann95}] \label{def:seriousness_ord_LC}
Let $K$ be a set of defaults and $k$ its order.
 To every subset $D \subseteq K$ may be associated
a $k+1$-tuple of numbers $\langle n_0, n_1, \ldots, n_k \rangle_D$,  
%
where:  
\begin{itemize}
\item[-]
 $n_0$ is the number of defaults in $D$ with no rank 
 in the rational closure of $K$;
\item[-]
$n_i$ (for $1 \leq i \leq  k$) is the number of defaults  in $D$ with rank $k-i$ in the rational closure of $K$. 
 \end{itemize}
We shall order the subsets of $K$ by the natural lexicographic ordering on their associated tuples.
This is a strict modular partial ordering, that will be denoted by $\prec$ (the seriousness ordering).
\end{definition}
\normalcolor
%
%
Considering again the knowledge base $K$ of Example  \ref{example-Student} containing the conditionals (defaults):
\begin{quote}
 1. $\mathit{ Student \ent \neg Pay\_Taxes}$\\
 2. $\mathit{Student \ent  Young}$\\
 3. $\mathit{ Employee \wedge Student \ent Pay\_Taxes}$
\end{quote}
the subset of defaults $D=\{ \mathit{Student \ent }$ $\mathit{   Young, \;Employee \wedge Student \ent Pay\_Taxes} \}$, that  we will denote synthetically as $D=\{2,3\}$,
for instance,  is associated with the tuple $\langle n_0, n_1,  n_2 \rangle_D$ =$\langle 0, 1, 1  \rangle_D$ meaning that $D$ contains: no default with no rank ($n_0=0$), 
one default (default $3$)  with rank $1$ (i.e., $n_1=1$ default with rank $k-1=1$) and one default  (default $2$)  with rank $0$ (i.e., $n_2=1$ default with rank $k-2=0$).
The set $B=\{3\}$ is instead associated with the tuple $\langle 0, 1, 0  \rangle_B$, meaning that $B$ contains: no default with no rank, 
one default (default $3$)  with rank $1$ and no default  with rank $0$.
The set $D=\{2,3\}$ is more serious than $B=\{3\}$ as  $\langle 0, 1, 0  \rangle_B$ $\prec \langle 0, 1, 1  \rangle_D$.

The modular order $\prec$ among sets of defaults gives preference to those sets of defaults containing more specific defaults.
Notice that, for $i\neq 0$, the numbers $n_i$ in the tuple $\langle n_0, n_1, \ldots, n_k \rangle_D$ are in decreasing order w.r.t. the rank of the corresponding defaults, and the higher is the rank of a default, the more specific is the default. \normalcolor


Lehmann defines a notion of  {\em basis for a formula $A$ in a knowledge base $K$}. Let the material counterpart of $D$, denoted by $\tilde{D} $, be the set containing a material implication $A \ri B$, for each conditional $A \ent B$ in $ D$.

\begin{definition}[Basis \cite{Lehmann95}]
A  {\em basis} for $A$ is a set $D$ of defaults in $K$ such that $A$ is consistent with $\tilde{D} $,
the material counterpart of $D$,
and $D$ is maximal w.r.t. the seriousness ordering $\prec$ for this property.
\end{definition}
In the example above, the set of  defaults $D=\{2,3\}$ forms a  basis for  $\mathit{Employee \wedge}$ $\mathit{ Student}$,  as its materialization 
$\tilde{D} =\{ \mathit{ Student \rightarrow Young, \;  Employee  \wedge Student \rightarrow}$ $\mathit{ Pay\_Taxes} \}$ is consistent (in the propositional calculus) with 
$\mathit{Employee \wedge Student}$, and $D$ is maximal w.r.t. the seriousness ordering  among the sets having this property.
 $D$ is actually  the unique basis for $\mathit{Employee \wedge Student}$.

\begin{definition}[Lexicographic closure \cite{Lehmann95}]
A conditional $A \ent B$ {\em is in  the lexicographic closure of $K$} 
if $\tilde{D} \cup {A} \models B$, for any basis $D$ for $A$.
\end{definition}
In the example, $\mathit{Employee \wedge Student \ent Young}$ belongs to the lexicographic closure of $K$, 
as $\tilde{D} \cup \{ \mathit{(Employee \wedge Student)} \} \models Young$, for the unique basis $D$ for 
 $\mathit{Employee}$ \linebreak $\mathit{ \wedge Student}$. This is what is expected, since the property of typical students of being young is inherited by employed students by presumption of independence.

In the following we will consider two variants of the knowledge base in Example  \ref{example-Student} 
to illustrate the lexicographic closure and, later to describe its common points and differences with respect to MP-closure, that will be introduced in  Section \ref{sez:BPsem}.
\begin{example} \label{example-Student-new}
Let $K'$ be the knowledge base containing the conditionals:
\begin{quote}
1. $\mathit{ Student \ent \neg Pay\_Taxes}$\\
2. $\mathit{Student \ent  Bright}$\\
3. $\mathit{Employee \ent   Pay\_Taxes}$\\
4. $\mathit{Employee \wedge Student  \ent  Busy}$
\end{quote}
Here, Students and Employees have a conflicting property: students normally do not pay taxes,  while employees normally do pay taxes.
Furthermore, students are normally bright and  employed students are normally busy. 

According to  rational closure,  the formulas  $\mathit{Student}$ and $\mathit{Employee}$ have both rank $0$,
while the formula  $\mathit{Employee \wedge Student}$ has rank $1$.
Therefore, conditionals $1, 2, 3$ have rank $0$, while conditional $4$ has rank $1$.
 Furthermore, the propositions $\mathit{Employee \wedge Student \wedge  Pay\_}$ $\mathit{Taxes}$  and  $\mathit{Employee \wedge Student \wedge  \neg Pay\_Taxes}$ have both rank $1$  in the rational closure.
As a consequence, the conditional $\mathit{Employee \wedge}$  $\mathit{ Student \ent  Pay\_Taxes}$ 
does not belong to  rational closure of $K'$, since $\rf(\mathit{Em}$- $\mathit{ployee \wedge Student})=$ $\rf(\mathit{Employee \wedge Student \wedge \neg Pay\_Taxes})=1$,
and similarly for  conditional $\mathit{Employee \wedge Student \ent  \neg Pay\_Taxes}$.
The problem is that the same holds for  conditional $\mathit{Employee \wedge Student \ent  Bright}$, which is also not in the rational closure of $K'$  (since $\rf(\mathit{Employee \wedge Student})= 1 =$ $\rf(\mathit{Employee \wedge}$  $\mathit{Student \wedge}$  
$\mathit{ \neg Bright})$). \normalcolor
However, we would like to conclude it, as the property that typical students are bright does not conflict with other properties of typical employees and of typical employed students.


In this example, there are two bases for $\mathit{Employee \wedge Student }$ in the lexicographic closure of $K'$: $D=\{1, 2, 4\}$ and $B=\{2, 3, 4\}$. They represent two alternative scenarios, 
the first one in which typical employed students inherit from typical students the property of not paying taxes, and the second one 
in which typical employed students inherit from typical employees the property of paying taxes.
Observe that $D= \{1, 2, 4\}  $ and  $B= \{2, 3, 4\}$ are not comparable with each other, i.e., none of them is {\em more serious} than the other
(that is, $D \not \prec B$ and $B \not \prec D$),
as the tuples $\langle 0, 1, 2  \rangle_D$ and $\langle 0, 1, 2  \rangle_B$, associated with $D$  and $B$ (respectively), are not comparable
in the lexicographic order.

Both the bases contain the default that normally students are bright and, as intended, this property extends to employed students.
For this reason, the conditional $\mathit{Employee \wedge}$ $\mathit{ Student  \ent  Bright}$ is in the lexicographic closure of $K'$, since
$\tilde{D} \cup \mathit{Em}$- $\mathit{ployee} \wedge \mathit{Student} \models \mathit{ Bright}$ and 
$\tilde{B} \cup {\mathit{Employee \wedge Student}} \models \mathit{ Bright}$.
%
 On the contrary, the lexicographic closure neither contains the conditional  $\mathit{Employee \wedge Student \ent Pay\_}$ $\mathit{Taxes}$ 
 (as  $\tilde{D} \cup {\mathit{Employee \wedge Student }} \not \models \mathit{ Pay\_Taxes}$)
 nor the conditional $\mathit{Employee}$ $\mathit{ \wedge Student \ent \neg Pay\_Taxes}$
  (as  $\tilde{B} \cup {\mathit{Employee \wedge Student }} \not \models \mathit{ \neg Pay\_Taxes}$) \normalcolor
 i.e., each of these conditionals is falsified in one of the two bases of $K'$
 (which are conflicting).

\end{example}
The following variant of  Example  \ref{example-Student-new} 
has a single basis and suggests that the lexicographic closure is sometimes too bold.
\begin{example}  \label{example-differenza_MC_LC}
Let the knowledge base $K''$ contain the following conditionals:
\begin{quote}
1. $\mathit{ Student \ent \neg Pay\_Taxes}$\\
2. $\mathit{Student \ent  Young}$\\
3. $\mathit{Employee \ent  \neg Young \wedge  Pay\_Taxes}$\\
4. $\mathit{Employee \wedge Student  \ent  Busy}$
\end{quote}
 As in previous example, according to the rational closure,  the formulas  $\mathit{Student}$ and $\mathit{Employee}$ have both rank $0$,
while the formula  $\mathit{Employee \wedge Student}$ has rank $1$. Indeed, formula $\mathit{Employee \wedge Student}$ is exceptional for $C_0=K''$, as $C_0 \models_{\Pe} \top \ent \neg ( \mathit{Employee \wedge Student})$. The intuition is that typical employed students cannot be typical students as well as typical employees, 
as typical students and typical employees have conflicting properties (concerning paying taxes).
Therefore, defaults 1, 2 and 3 have rank $0$ in the rational closure, while default 4 has rank $1$. \normalcolor

As a difference with the previous example, the  lexicographic closure has 
a single basis, $D=\{ 1,2, 4\}$.
Indeed,  the two sets of defaults $D=\{1,2, 4 \}$ and  $B=\{ 3, 4 \}$, whose materializations are both consistent with $\mathit{Employee \wedge Student}$,
have the associated tuples $\langle 0, 1, 2  \rangle_D$ and  $ \langle 0, 1, 1  \rangle_B$
and, therefore, $B$ is less serious than $D$ ($B \prec D$).
 As a consequence, there is a single basis $D$ for $\mathit{Employee \wedge Student}$, and we can conclude that typical employed students are not only busy, but (like typical students) they are also young and do not pay taxes.
The conditional
\begin{align} \label{conditional_in_LC-MP}  
\mathit{Employee \wedge Student  \ent  Young  \wedge \neg  Pay\_Taxes }
\end{align}
is in the lexicographic closure of $K''$,
 as
 $\tilde{D} \cup \{ \mathit{Employee \wedge Student\}  \models  Young  \wedge \neg }$ $\mathit{ Pay\_Taxes}$, and $D$ is the only basis for $K''$.
\end{example}
The result above is in line with the choice of  the lexicographic closure that, in the case of contradictory defaults with the same rank, as many defaults as possible  should be satisfied. 
However, the reason to accept that typical employed students are not young and pay taxes (rather than the converse) may be questioned
and, in this last example, the lexicographic closure appears to be {\em too bold}. 
Indeed, the conclusion that normally employed students are young and do not pay taxes, i.e., conditional (\ref{conditional_in_LC-MP}), here follows from the accidental fact that the properties of Employees are expressed by a single default, while the properties of Students are expressed by two defaults.
Notice that, if we replace default $3$ with the two defaults 
  3.1 $\mathit{Employee \ent \neg Young}$ and 
 3.2 $\mathit{Employee \ent  Pay\_Taxes}$,
there would be two bases in the lexicographic closure, and one would not be allowed to conclude any more that typical employed students are young and do not pay taxes. 
As observed by Lehmann, the  lexicographic closure construction 
is ``extremely sensitive to the way defaults are presented" and ``the way defaults are presented is important" \cite{Lehmann95}.
\normalcolor

In the following section, we will consider a different notion of closure, the MP-closure,  that departs 
from the assumption of  the lexicographic closure for which, in case of contradictory defaults with the same rank, as many defaults as possible should be satisfied.
In particular, in Example \ref{example-differenza_MC_LC}, it considers both the sets of defaults $D$ and $B$ to be maximally serious,
and it does not conclude conditional (\ref{conditional_in_LC-MP}).
Although also the MP-closure is somewhat syntax dependent, in this case, differently from the lexicographic closure, it treats in the same way the two different formulations of the knowledge base $K''$ above  (the one with default 3, and the other one with the two defaults 3.1 and 3.2). 
We will show that the MP-closure is stronger than the rational closure but weaker than the lexicographic closure,
%
%
and it defines a preferential consequence relation,
rather than a rational consequence relation,
as well as a more cautious notion of entailment (with respect to the lexicographic closure), 
that does not satisfy the property of Rational Monotonicity.


Following the pattern in \cite{Lehmann95}, in Section \ref{sez:BPsem}, we present both a characterization of the MP-closure in terms of maxiconsistent sets and a model-theoretic construction.
 In Section \ref{sec:rational_relation}, we exploit the idea,
proposed by Lehmann and Magidor in their semantic characterization of the rational closure \cite{whatdoes},
to transform a well-founded preferential model into a ranked model, 
and apply it to the models of the MP-closure to define a rational consequence relation which is a superset of the MP-closure. 
\color{black} We see that such a rational consequence relation, as lexicographic closure, is another rational superset of the rational closure, and is neither stronger nor weaker than the lexicographic closure.
\normalcolor


 \section{The MP-closure revisited} \label{sez:BPsem}

The multipreference closure  (MP-closure, for short), was preliminarily 
introduced in the technical report \cite{Multipref_arXiv2018} as a construction 
which soundly approximates the multipreference semantics proposed by Gliozzi \cite{GliozziNMR2016,GliozziAIIA2016}  for the description logic $\alc$ with typicality, 
thus defining a refinement of the rational closure of $\alc$.
This semantics was originally proposed for separately reasoning about the inheritance of different properties and, hence, to provide a solution to the drowning problem related to rational closure.

The interest of the MP-closure construction goes beyond description logics and 
its definition and semantics can be reconstructed and significantly simplified in the context of propositional logic.
 MP-closure can be regarded as the natural 
variant of lexicographic closure, 
if we are ready to abandon 
 the assumption that the {\em number} (rather than the set) of defaults with the same rank matters
(as illustrated by Example \ref{example-differenza_MC_LC} in the previous section),
an alternative route already considered but not explored by Lehmann \cite{Lehmann95}.
%
In this section 
we reformulate the MP-closure construction from \cite{Multipref_arXiv2018} in the propositional setting and, then, we focus on its semantics, its properties 
and relations with the lexicographic closure. 
In Section \ref{sec:further_issues}, we will investigate the relationships of the MP-closure with  {\em system ARS} \cite{IsbernerRitterskamp2010}, with Brewka's Basic Preference Descriptions \cite{Brewka04}, with the Relevant Closure \cite{Casini2014}, with Geffner and Pearl's Conditional entailment \cite{Geffner1992} and other approaches.


\subsection{The MP-closure construction} \label{sec:construction}

 For a given finite knowledge base $K$, with order $k$, 
in the following, we exploit the MP-closure construction to define the plausible consequences of a formula $A$ with a finite rank $\rf(A)$.
Observe that, for any formula $A$ with infinite rank, i.e., such that $\rf(A)=\infty$, the conditional $A \ent C$ is in the rational closure of $K$, for any $C$.

Given a subset $D$ of conditional assertions in $K$ (a set of defaults), we let $D_i$ be the set of defaults in $D$ with finite rank $i \leq k$, and
$D_\infty$ be the set of defaults in $D$ with no rank. 
The tuple $\langle D_{\infty}, D_k, \ldots, D_1,D_0 \rangle_D$, associated with $D$,
defines a partition of $D$, according to the ranks of the defaults in the rational closure of $K$.

We define a preference relation $\prec^{MP}$ among sets of defaults, by comparing the tuples associated to these sets according to
the natural lexicographic order on such tuples, defined inductively as follows. 
 Given  two tuples  $\langle X_n, \ldots, X_1 \rangle$ and $\langle X'_n, \ldots, X'_1  \rangle$ of sets of defaults in $K$, we let:
\begin{align*}
\langle X_1 \rangle \ll \langle  X'_1  \rangle  & \mbox{ iff  }  X_1 \subset X'_1   \\
\langle X_n, \ldots, X_1 \rangle \ll \langle X'_n, \ldots, X'_1  \rangle  &\mbox{ iff  }  X_n \subset X'_n  \mbox{ or  }  \\
&( X_n = X'_n \mbox{ and } \langle X_{n-1}, \ldots, X_1 \rangle \ll \langle X'_{n-1}, \ldots, X'_1 \rangle )
\end{align*}
As  the (strict) subset inclusion relation $\subset$ among sets is a strict partial order, the lexicographic order $\ll$ on the tuples of sets of defaults is a strict partial order as well. \normalcolor 
This lexicographic order provides a new seriousness ordering among sets of defaults.

\begin{definition}[MP-seriousness ordering]\label{MP-order}
$D \prec^{MP} B$ ($D$ is less serious than $B$ w.r.t. the MP-seriousness ordering) iff
\begin{center}
$\langle D_{\infty}, D_k, \ldots, D_1,D_0 \rangle_D$ $\ll \langle B_{\infty}, B_k, \ldots, B_1,B_0 \rangle_B$.
\end{center}
\end{definition}
Notice that the relation $\prec^{MP}$ defines a seriousness ordering among sets of defaults, which is different from the seriousness ordering used by the lexicographic closure, in that here we pay attention to which defaults are in the $D_i$'s, while in lexicographic closure to the cardinality of the sets $D_i$'s.
In fact, the corresponding tuple associated with $D$ by the lexicographic closure would be $\langle \mid D_{\infty}\mid , \mid D_k\mid , \ldots, \mid D_1\mid ,\mid D_0 \mid \rangle_D$ (see Section \ref{sez:Lex-closure}). 
As  the lexicographic order $\ll$ on tuples of sets of defaults is a strict partial order, $ \prec^{MP} $ is  a strict partial order as well. 
This partial order is not necessarily modular and we will see at the end of this section that the MP-closure does not satisfy Rational Monotonicity.


The difference of the seriousness ordering between lexicographic closure and the MP-closure has an impact on the kind of conclusions one draws in the two cases, as we will see below.
Let us first give a characterization of the MP-closure in terms of bases.
\begin{definition}[MP-basis]
Given a finite knowledge base $K$, and a formula $A$ with finite rank, a set of defaults $D \subseteq K$ is a basis for $A$
if $A$ is consistent with  $\tilde{D}$ (the material counterpart of $D$) and $D$ is maximal w.r.t. the MP-seriousness ordering for this property.
\end{definition}
Notice that the definition of a basis is  exactly the same as in the lexicographic closure \cite{Lehmann95}, but for the fact that it uses a different lexicographic ordering.

\begin{definition}[MP-closure]\label{def:MP-closure}
A default $A \ent B$ is in {\em $MP(K)$, the MP-closure of a knowledge base $K$},  if for all the MP-bases $D$ for $A$:
$$ \tilde{D} \cup \{ A\} \models B,$$%
where $\models$ is logical consequence in the propositional calculus and $\tilde{D}$ is the materialization of $D$. 
\end{definition}
Consider again Example   \ref{example-differenza_MC_LC},
the two sets of defaults $D=\{1,2, 4 \}$ and  $B=\{ 3, 4 \}$ are now  incomparable using the $\prec^{MP}$ preference relation,
as the tuples associated to the sets $D$ and $B$ are respectively:
$\langle \emptyset, \{4\},\{1,2\}  \rangle_D$ and  $ \langle \emptyset, \{4\},\{3\} \rangle_B$
and neither $\langle \emptyset, \{4\},\{1,2\}  \rangle_D$  $ \ll \langle \emptyset, \{4\},\{3\} \rangle_B$
nor $ \langle \emptyset, \{4\},\{3\} \rangle_B$ $ \ll \langle \emptyset, \{4\},\{1,2\}  \rangle_D$  
(on the contrary, as we have seen above, in the lexicographic closure, $D$ is more serious than $B$).
Thus, there are two MP-bases for $\mathit{Employee \wedge Student}$, namely $D=\{1,2, 4 \}$ and  $B=\{ 3, 4 \}$.
Therefore, neither  $\mathit{Employee \wedge Student  \ent  }$ $\mathit{Young }$ nor $\mathit{Employee \wedge }$ $\mathit{ Student  \ent  \neg Young }$ are in the MP-closure of $K''$.
In this example, the MP-closure is less bold than the lexicographic closure, which, as we have seen, includes the default $\mathit{Employee \wedge }$ $\mathit{Student  \ent  Young }$, as $D$ is the only basis for $\mathit{Employee \wedge Student}$ in the lexicographic closure.

Concerning  Examples  \ref{example-Student} and  \ref{example-Student-new} above, it is easy to see that in both of them the MP-closure has the same bases as the lexicographic closure, as well as the same consequences.
 As a further example on which the MP-closure allows weaker conclusions than the lexicographic closure,
consider the following one, which is an instance of the  ``evidence comparison" example from \cite{Weydert03},
and shows that, as a difference with respect to the  lexicographic closure (and with respect to System JLZ  \cite{Weydert03}),
in the MP-closure the ``weight of independent reasons" supporting some conclusion is not taken into account.

\begin{example}\label{exa:evidence_comparison}
Consider the knowledge base $K$ containing the conditionals: 
\begin{quote}
1. $\mathit{ Olympic\_Swimmer \ent Young}$\\
2. $\mathit{Adult \ent  \neg Young}$\\
3. $\mathit{Employee \ent \neg Young}$.
\end{quote}
meaning that normally Olympic  swimmers are young, that normally employees are not young and that normally adults are not young.
The three conditionals 1, 2 and 3 have rank 0 in the rational closure.
The conditionals
\begin{quote}
$\mathit{ Olympic\_Swimmer \wedge Adult \wedge Employee \ent  Young}$  and \\
$\mathit{ Olympic\_Swimmer \wedge Adult \wedge Employee \ent  \neg Young}$ 
\end{quote}
do not belong to the MP-closure of $K$.
Indeed, there are two MP-bases for $A=\mathit{ Olympic\_Swimmer \wedge Adult \wedge Employee}$, namely $D=\{1 \}$ and  $B=\{ 2, 3 \}$ (which are incomparable using the $\prec^{MP}$ preference relation), and
 $ \tilde{D} \cup \{ A\} \models \mathit{Young}$, while $ \tilde{B} \cup \{ A\} \models \mathit{\neg Young}$.

On the contrary, the conditional $\mathit{ Olympic\_Swimmer \wedge Adult \wedge Employee \ent}$ 
 \linebreak $\mathit{ \neg Young}$ belongs to the lexicographic closure of $K$,
as the only basis for $A$ in the lexicographic closure is $B=\{ 2, 3 \}$,
which is preferred to $D$ as $B$ contains two defaults with rank 0, while $D$ just one.
In this example, the lexicographic closure (as System JLZ) accepts the conditional $A \ent \mathit{\neg Young}$, as it is supported by two independent defaults (2 and 3),
while $A \ent \mathit{ Young}$ is only supported by default 1.

Notice that 
in the lexicographic closure the same result would be obtained if a fourth conditional  $\mathit{Employee \ent Adult}$ were added in $K$
and,  clearly, in this case, conditionals 2 and 3 would not be independent evidences for $\mathit{\neg Young}$.
\end{example}
\normalcolor

It can be proved that the MP-closure is stronger than the rational closure, but weaker than the lexicographic closure.
We prove the second result (Corollary  \ref{cor:MP-LC}), while postponing the proof of the first one until the introduction of the semantics of the MP-closure in Section \ref{sec:semantics}.
 In order to prove Corollary 2, let us first prove the following proposition and Corollary  \ref{corollary:bases}. We show that $ \prec^{MP}$ is  at least as coarse as $ \prec$.   

\begin{proposition} \label{Prop:coarser}
$ \prec^{MP}$ is at least as coarse as $ \prec$, that is, for all the sets of defaults $D$ and $B$, if $D \prec^{MP} B$ then $D \prec B$. 
\end{proposition}
\begin{proof}
Given a knowledge base $K$ and  two sets of conditionals $D, B \subseteq K$, let us assume that $D \prec^{MP} B$. 
 As $D$ is less serious than $B$ in the MP-ordering, it must be that:
\begin{center}
$\langle D_{\infty}, D_k, \ldots, D_1,D_0 \rangle_D$ $\ll \langle B_{\infty}, B_k, \ldots, B_1,B_0 \rangle_B$.
\end{center}
consider  the highest  $j$, with $0 \leq j \leq k$, such that $D_j \neq B_j$.
For such a $j$, it must be that  $D_j \subset B_j$, while $D_r = B_r$, for all $r$ such that $k \geq r >j$.

Let us now consider the two tuples of numbers 
\begin{center}
$\langle n_{\infty}, n_k, \ldots, n_1,n_0 \rangle_D$  and $ \langle m_{\infty}, m_k, \ldots, m_1,m_0 \rangle_B$
\end{center}
associated with $D$ and $B$, respectively, in the lexicographic closure construction.
Notice that, $n_i= \mid D_i \mid$, for all $i$, and  $m_i= \mid B_i \mid$, for all $i$.
Furthermore, $D_{\infty}(=B_{\infty})$ is the set of all the conditionals with rank $\infty$ in the rational closure and,
hence, $n_{\infty}=m_{\infty}$.
For all $r$ such that $k \geq r >j$, as $D_r = B_r$, it must be that $n_r=m_r$.
Also, from $D_j \subset B_j$, we get $n_j <m_j$.
Thus, using the lexicographic ordering on numbers
tuple
$\langle n_{\infty}, n_k, \ldots, n_1,n_0 \rangle_D$ comes before tuple $ \langle m_{\infty}, m_k, \ldots, m_1,m_0 \rangle_B$,
and, therefore, $D \prec B$.
\hfill $\Box$
\end{proof}
As a consequence of this result, it is easy to prove the following corollaries. 
\begin{corollary} \label{corollary:bases}
Let $K$ be a knowledge base, $A$ a formula and $D \subseteq K$ a set of defaults. 
If $D$ is a basis for $A$ in the lexicographic closure, then $D$ is a basis for $A$ in the MP-closure.
\end{corollary}
\begin{proof}
Let $D$ be a basis for $A$ in the lexicographic closure, i.e., 
$A$ is consistent with  $\tilde{D}$ and $D$ is maximal w.r.t. $\prec$-seriousness ordering for this property.

We show that $D$ is also a basis in the MP-closure. 
If not, there is a set of defaults $B \subseteq K$ such that $A$ is consistent with $\tilde{B}$ and $D \prec^{MP} B$.
But, then, by Proposition \ref{Prop:coarser}, $D \prec B$, and $D$ is not  maximal w.r.t. $\prec$ among the sets of default whose materialization is consistent with $A$,
which contradicts the hypothesis that $D$ is a basis for $A$ in the lexicographic closure.
\hfill $\Box$
\end{proof}

\begin{corollary} \label{cor:MP-LC}
Let $K$ be a knowledge base and $A$ a formula.
If $A \ent C$ is in the MP-closure of $K$, then $A \ent C$ is in the lexicographic closure of $K$.
\end{corollary}
\begin{proof}
If $A \ent C$ is in the MP-closure of $K$, then, in all the bases $D$ for $A$ in the MP-closure, $\tilde{D} \cup \{ A \} \models C$.
Let $D$ be any basis for $A$ in the lexicographic closure of $K$.
As,  by Corollary  \ref{corollary:bases}, all the bases for $A$ in lexicographic closure of $K$ are also MP-bases for $A$,
$\tilde{D} \cup \{ A \} \models C$.
Hence,  for all the bases $D$ for $A$ in the lexicographic closure of $K$, $\tilde{D} \cup \{ A \} \models C$,
and $A \ent C$ is in the lexicographic closure of $K$.
\hfill $\Box$
\end{proof}

To conclude this section, we show that the MP-closure, 
differently from the lexicographic closure, does not define a rational consequence relation.  
 For a knowledge base $K$, let ${\cal MP}_K$ be  the set of conditionals in the MP-closure of $K$.
The following counterexample shows a knowledge base $K$ such that  ${\cal MP}_K$  does not satisfy the property of Rational Monotonicity, and is a reformulation of Lehmann's musician example \cite{Lehmann95}. 

\begin{example}\label{exa:counterexa_RM}
Given the following knowledge base $K$: 
\begin{quote}
1. $\mathit{ Student \ent Merry}$\\
2. $\mathit{Student \ent  Young}$\\
3. $\mathit{Adult \ent Serious}$\\
4. $\mathit{Student \wedge Adult \ent (\neg Young \wedge \neg Merry) \vee \neg Serious}$
\end{quote}
the conditionals 1, 2 and 3 have rank 0 in the rational closure, while conditional 4 has rank 1.
There are two bases for $\mathit{ Student \wedge Adult}$  in the MP-closure of $K$, $D=\{1,2,4\}$ and $B=\{3,4\}$,
and the conditional $\mathit{ Student \wedge Adult \ent  Young \leftrightarrow Merry}$ is in the MP-closure of $K$ (in ${\cal MP}_K$).
On the contrary,   the conditional $\mathit{ Student \wedge Adult \ent }$ $\mathit{  Young}$ is not in ${\cal MP}_K$ (as $Young$ does not hold in the basis $B$), that is, $\mathit{ Student \wedge Adult }$ $\mathit{\not \ent  Young}$.
By the property of Rational Monotonicity, the conditional
\begin{quote}
$\mathit{ Student \wedge Adult \wedge \neg Young \ent Young \leftrightarrow Merry}$
\end{quote}
should be in ${\cal MP}_K$.
On the contrary, the last conditional 
 is not in the MP-closure of $K$.
In fact, there are two bases for $\mathit{ Student \wedge Adult  }$
$\mathit{ \wedge \neg Young}$,
namely $D'=\{1,4\}$ and $B'=\{3,4\}$, and the formula $\mathit{Young \leftrightarrow Merry}$ does not hold in the first basis, as  $\mathit{ \tilde{D'} \cup \{ Student \wedge Adult \wedge \neg Young\}  \models}$
$\mathit{  \neg Young \wedge Merry}$.
\end{example}
Notice that the example above is not a counterexample to Rational Monotonicity for the lexicographic closure, which is known to define a rational consequence relation.
In fact,  $D=\{1,2,4\}$ is the only basis for $\mathit{ Student \wedge Adult}$  in the lexicographic
closure of $K$  and, hence, the conditional $\mathit{ Student \wedge }$
$\mathit{  Adult \ent  Young}$ is in the lexicographic closure of $K$. 

\subsection{A semantic characterization for MP-closure}\label{sec:semantics}

A semantics for  MP-closure is defined in \cite{arXiv_Skeptical_closure}  building on the preferential semantics for rational closure of $\alctr$, 
introducing a notion of {\em refined, bi-preference interpretation}, 
which contains two preference relations, let us call them $<$ and $<'$:  the first one plays the role of the preference relation in a model of the RC in $\alctr$, while the second one $<'$ is built from $<$ exploiting a specificity criterium, and represents a refinement of $<$.

In this section we define a simpler semantic characterization of the MP-closure of a propositional knowledge base, starting from the propositional models of the rational closure. 
This simplified setting, that corresponds to the one considered by Lehmann in his semantic characterization of the lexicographic closure \cite{Lehmann95},
also allows for an easy comparison among the two semantics.

Given a finite satisfiable knowledge base $K$,
in the following we  define the semantics of the MP-closure by means of some preferential models of $K$ (that we call MP-models) and, then, we prove a characterization result.
To this purpose, we introduce a functor ${\cal F}$ associating a preferential interpretation $\enne$ to each finite minimal canonical ranked model  $\emme \in Min_{RC}(K)$ characterizing the rational closure of $K$ according to  Theorem \ref{rat_closure_modelli_minimali}.
As we will see, $\enne$ will be a model of the MP-closure.
\normalcolor

\begin{definition}[Functor ${\cal F}$]\label{def:functor}
Given a minimal canonical ranked interpretation $\emme=\langle \WW, <, v \rangle$ in  $\mathit{Min_{RC}(K)}$, we let
${\cal F}(\emme)= \enne$ such that: $\enne=\langle \WW, <', v \rangle$ and 
\begin{align}  \label{order_on_worlds}
 x <' y \mbox{\ \ {\bf \em iff} } V(x)  \prec^{MP} V(y)
 \end{align}
 where, for $z \in \WW$, $V(z)$ is the set of defaults in $K$ which are violated by $z$ (i.e., the set of conditionals $A \ent C \in K$ such that $\emme, z \models A \wedge \neg C$).
\end{definition}
As $\prec^{MP}$, introduced in  Definition \ref{MP-order}, 
is a strict partial order,
it is easy to see  that $<'$ in the definition above is a strict partial order as well.
Indeed, $<'$ is irreflexive: if $x <' x$ then $ V(x)  \prec^{MP} V(x)$ but, by irreflexivity of $\prec^{MP} $, $ V(x) \not  \prec^{MP} V(x)$, a contradiction. Hence, $x <'x$ does not hold.
Also, $<'$ is transitive: if $x <' y$ and $y <' z$, then, by definition of $<'$, $ V(x)  \prec^{MP} V(y)$ and $ V(y)  \prec^{MP} V(z)$. From the transitivity of 
$\prec^{MP} $, $ V(x)  \prec^{MP} V(z)$. Therefore, $x<'z$.

Hence, $\enne=\langle \WW, <', v \rangle$, in Definition \ref{def:functor},  is a preferential interpretation.
Propositions \ref{MP_stronger _than_RC} and \ref{enne_is_model_of_K} below  will show that $< \subseteq <'$, i.e., the preference relation $<'$ is 
 at least as fine as the modular preference relation $<$, and that $\enne$ is a model of $K$. 
Thus, $\enne$ is a preferential model of $K$ which is the refinement of the model $\emme$ of the rational closure of $K$ (in the sense that the preference relation in $\enne$ is 
at least as fine as the preference relation in $\emme$).

Before stating Propositions \ref{MP_stronger _than_RC} and \ref{enne_is_model_of_K}, let us consider a simple example.

\begin{example}
Consider again the minimal canonical model $\emme=\langle \WW, <, v \rangle$ of knowledge base $K$ in Example  \ref{es:modello_canonico}:
\begin{quote}
 1. $\mathit{ Student \ent \neg Pay\_Taxes}$\\
 2. $\mathit{Student \ent  Young}$\\
 3. $\mathit{ Employee \wedge Student \ent Pay\_Taxes}$
\end{quote}
As we have seen, default 1 and 2 have rank $0$, while default $3$ has rank $1$.
Given the following ranking of worlds:

Rank $0$: \ $\emptyset$, $\{ p\}$, $\{ y \}$, $\{ p, y \}$,  $\{ e\}$, $\{ e,y\}$, $\{ e,p\}$, $\{ e,y,p\}$, $\{ s,y\}$

Rank $1$: \   $\{ s,e,p\}$, $\{ s,e,p,y\}$,  $\{ s\}$, $\{ s,p\}$, $\{ s,p,y\}$

Rank $2$: \   $\{ s,e\}$, $\{ s,e,y\}$, 

\noindent
the relation $<$ in $\emme$ is defined as usual (see Section \ref{sez:RC}): for each world $w$ with rank $i$ and each world $z$ with rank $j>i$, $w<z$ holds.
 
Let us consider the model $\enne=\langle \WW, <', v \rangle$ such that $\enne={\cal F}(\emme)$.
Consider the two worlds: $x=\{s,y\}$ and $y=\{s,e,p,y\}$. As $x$ has rank $0$ and $y$ has rank $1$, $x <y$ holds in $\emme$.
Does $x <'y$ hold as well?
As $x$ violates no default (i.e., $V(x)=\emptyset$) and $y$ violates default 1 with rank $0$ (i.e., $V(x)=\{1\}$), the violation of $x$ is less serious than the violation of $y$ and $V(x)  \prec^{MP} V(y)$. Then, by condition  (\ref{order_on_worlds}), $x <'y$ holds in $\enne$. As expected, $x$ and $y$ keep their relative order from $\emme$.

Let us now consider two worlds having the same rank in $\emme$, namely $w= \{ s,e,p,y\}$ and $z=\{ s,e,p\}$. 
As $w$ violates default 1, while $z$ violates defaults 1 and 2, $V(w)  \prec^{MP} V(z)$, that is,
the violations in $w$ are less serious  than those  in $z$ and, hence, $w<'z$.
As a consequence, in $\enne$ $z$ is no more a minimal world satisfying  $\mathit{ Employee \wedge Student }$.
On the contrary,  $w$ is the only minimal world satisfying $\mathit{ Employee \wedge Student }$
and, as $\mathit{Young}$ is true in $w$, the conditional $\mathit{ Employee \wedge Student \ent Young}$ is satisfied in $\emme$.

\end{example}
\normalcolor

Let us prove that $<'$ is a refinement of $<$.

\begin{proposition}\label{MP_stronger _than_RC} 
For all $\emme=\langle \WW, <, v \rangle \in Min_{RC}(K)$ and $\enne=\langle \WW, <', v \rangle$ such that
$\enne= {\cal F}(\emme)$, it holds that  $< \subseteq <'$, i.e., the preference relation $<'$ is  at least as fine as $<$.

\end{proposition}
\begin{proof}
We show that, for all $x,y \in \WW$, $x <y$ implies $x <'y$.

If $x <y$ in $\emme$, then for some $j,h$, $k_{\emme}(x)=j < h= k_{\emme}(y)$. 
As $\emme$ is a ranked model of $K$, 
by Proposition 2 in \cite{AIJ15},
$\emme, x \models A \ri C$ for all the conditionals  $A \ent C \in C_j$
(i.e., for all the conditionals $A \ent C$ with rank $r \geq j$). 
In particular, letting $V_r(x)$ be the defaults with rank $r$ violated by $x$, we have:
$V_r(x)=\emptyset$, for all $r\geq j$. 
The tuple associated with $V(x)$ has the form 
$$\langle C_{\infty}, \emptyset, \ldots, \emptyset,V_{j-1}(x),  \ldots, V_0(x) \rangle_{V(x)}.$$
where $C_\infty$ is the set of conditionals in $K$ with rank $\infty$.
Such conditionals must be in $V(x)$ as they cannot be satisfied in any model of $K$.

As $k_{\emme}(y)=h>j$, there must be some conditional $A \ent B$ with $\rf(A)= h-1\geq j$, which is falsified by $y$. 
Hence, $V_{h-1}(y)$ is nonempty.
The tuple associated with $V(y)$ has the form 
$$\langle C_{\infty}, \emptyset, \ldots, \emptyset,V_{h-1}(y), \ldots, V_{j}(y), V_{j-1}(y),  \ldots, V_0(y) \rangle_{V(y)},$$
with  $V_{h-1}(y)\neq \emptyset$ and $h-1\geq j$.
Then, the tuple associated with $V(x)$ is less serious (in the MP-seriousness ordering) than the tuple associated to $V(y)$. Hence, $V(x) \prec^{MP} V(y)$.
Therefore, by (\ref{order_on_worlds}), $x<' y$ holds,
and we can then conclude that,  $< \subseteq <'$.
\hfill $\Box$
\end{proof}
To prove that a preferential interpretation $\enne$ such that  $\enne= {\cal F}(\emme)$, for some $\emme \in Min_{RC}(K)$, is a model of $K$,
we exploit the following general property of preferential interpretations.

\begin{proposition}\label{prop:<' subseteq <''}
Let $\enne' = \langle \WW, <', v \rangle$ and $\enne'' = \langle \WW, <'', v \rangle$ be two  preferential interpretations such that $<' \subseteq <''$. 
For all conditionals $A \ent C$, if $\enne' \models A \ent C$ then  $\enne'' \models A \ent C$.
\end{proposition}
\begin{proof}
Let $\enne' \models A \ent C$, i.e., for all $w \in \WW$, if $w \in Min_{<}^{\enne'}(A)$ then $\enne', w \modello C$.
We prove that $\enne'' \models A \ent C$.
Assume that $w \in Min_{<''}^{\enne''}(A)$.
We show that $\enne'', w \modello C$.

As $w \in Min_{<''}^{\enne''}(A)$,  clearly $w \in \WW$ and $\enne'', w \modello A$.
Since $\enne'$ and $\enne''$ have the same set of worlds $\WW$ and valuation function $v$, it follows that $ \enne', w \modello A$.
Furthermore, as the worlds satisfying $A$ in $\enne'$ and in $\enne''$ are the same and,
from $<' \subseteq <''$,
it follows that $Min_{<''}^{\enne''}(A) \subseteq Min_{<'}^{\enne'}(A)$.

Therefore $w \in Min_{<'}^{\enne'}(A)$ and, $\enne' \models A \ent C$,  then $ \enne', w \modello C$.
Again, as $\enne'$ and $\enne''$ have the same set of worlds $\WW$ and valuation function $v$, it follows that $ \enne'', w  \modello C$, which concludes the proof.
\hfill $\Box$
\end{proof}
As a consequence of the proposition above, an interpretation $\enne$ such that $\enne= {\cal F}(\emme)$ is a model of $K$.

\begin{proposition}\label{enne_is_model_of_K}
For all $\emme=\langle \WW, <, v \rangle \in Min_{RC}(K)$ and $\enne=\langle \WW, <', v \rangle$ such that
$\enne= {\cal F}(\emme)$, it holds that $\enne$ is a model of $K$.

\end{proposition}
\begin{proof}
From the hypothesis we know that $\emme$ is a minimal canonical ranked model of $K$.
Hence, $\emme$ satisfies all the conditionals in $K$, i.e., for all conditionals $A \ent {C} \in K$, $\emme \models A \ent C$.
As $\emme$ and $\enne$ are preferential models and, by Proposition \ref{MP_stronger _than_RC}, $< \subseteq <'$,
we can use Proposition \ref{prop:<' subseteq <''} to conclude that,
for all conditionals $A \ent  {C} \in K$, as $\emme \models A \ent C$, also 
$\enne \models A \ent C$. Thus $\enne$ is a model of $K$.
\hfill $\Box$
\end{proof}

%
%
We can naturally extend the functor ${\cal F}$ to a set of models, and let
$$ {\cal F}(Min_{RC}(K)) =\{  \enne \mid  \enne= {\cal F}(\emme) \mbox{  for some } \emme \in Min_{RC}(K) \}.$$
${\cal F}(Min_{RC}(K))$ is a set of preferential models of $K$ that we call  MP-models of $K$.


\begin{definition}[MP-models of $K$] \label{defi:MP-models}
Given a knowledge base $K$,  an {\em MP-model of $K$} is  any preferential model $\enne \in {\cal F}(Min_{RC}(K))$.  
\end{definition}
We prove that MP-models provide a semantic characterization of the MP-closure.
\begin{theorem}[Characterization result for the MP-closure] \label{MP-characterization}
Given a satisfiable  knowledge base $K$, 
a conditional $A \ent C$ is true in all the MP-models of $K$
if and only if $A \ent C$ belongs to the MP-closure of $K$.
\end{theorem} \normalcolor  
\begin{proof}
$(\Rightarrow)$ By contraposition, let us assume that $A \ent C$ does not belong to the MP-closure of $K$.
We show that there is a model $\enne$ in ${\cal F}(Min_{RC}(K))$ such that $A \ent C$ is not satisfied in $\enne$.

If $A \ent C$ does not belong to the MP-closure of $K$, $A$ must have a finite rank $i$ and 
there must be some basis $D$ for $A$ in $K$ such that $\tilde{D} \cup \{ A\} \not \models C$.
Therefore, there is some propositional interpretation satisfying $ \tilde{D} \wedge A \wedge \neg  C$.

Let $\emme =\langle \WW, <, v \rangle$ be any minimal canonical ranked model  in $Min_{RC}(K)$. 
We have seen in Section \ref{sez:RC} that, for a satisfiable knowledge base $K$, a minimal canonical model of $K$ always exists (Theorem 1 in \cite{AIJ15}).

First we prove that here must be a world $w \in \WW$ such that 
$\emme, w \models \tilde{D} \wedge A \wedge \neg  C$ (here $\tilde{D}$ stands for the conjunction $\bigwedge  \tilde{D}$).

As $A$ has rank $i$ in the rational closure, $A$ is not exceptional for $C_i$, i.e., $C_i \not \models_\Pe \top \ent \neg A$,
where  $C_i$ is the set of defaults with rank greater or equal to $i$, defined according to the rational closure construction in Section \ref{sez:RC}.
Also $\tilde{D} \wedge A \wedge \neg  C$ is not exceptional for $C_i$.
Suppose that this is not the case and that 
$C_i  \models_\Pe \top \ent \neg (\tilde{D} \wedge A \wedge \neg  C)$.
From the equivalence $\tilde{K} \models \alpha$ iff $K \models_\Pe \top \ent \alpha$,proved by Lehmann and Magidor \cite{whatdoes} ,
it would follow that $\tilde{C_i}  \models \neg (\tilde{D} \wedge A \wedge \neg  C)$ and, hence 
$\tilde{C_i}  \models \tilde{D} \wedge A \ri  C$ and, by the deduction theorem, $\tilde{C_i} \cup \tilde{D} \cup \{A \} \models  C$.
Notice that all the defaults in $C_i$ have rank greater or equal than $i$  and they must be in $\tilde{D}$.
In fact, the materialization of $C_i$ is consistent with $A$ (as, from $C_i \not \models_\Pe \top \ent \neg A$, it follows that
$\tilde{C_i} \not \models \neg A$) and the base $D$ is maximally serious among the sets of defaults consistent with $A$.
Hence, we would conclude $\tilde{D} \cup \{A\}  \models  C$, a contradiction.
Therefore, $\tilde{D} \wedge A \wedge \neg  C$ cannot be exceptional for $C_i$, and then  $\rf(\tilde{D} \wedge A \wedge \neg  C)=i$.
This also means that, in the minimal canonical model $\emme$ of $K$ there must be a world $w$ with rank $i$ 
satisfying $\tilde{D} \wedge A \wedge \neg  C$.

Let $\enne={\cal F}(\emme)$, i.e., $\enne$ is an MP-model of $K$. By construction, $\enne=\langle \WW, <', v \rangle $  and
$ x <' y \mbox{\ \ { \em iff} } V(x)  \prec^{MP} V(y)$.
We have to prove that $\enne$ falsifies $A \ent C$. 
As $\WW$ and $v$ in $\enne$ are the same as in $\emme$, it must be that $\enne, w \models \tilde{D} \wedge A \wedge \neg  C$.
We prove that $w$ is a minimal world satisfying $A$ in $\enne$, i.e., that $w \in Min_{<'}^{\enne}(A)$. As  $w$ falsifies $C$, 
it follows that $\enne$ falsifies $A \ent C$. 
\normalcolor

By absurd, if $w$ were not in $ Min_{<'}^{\enne}(A)$, there would be an element $z \in  Min_{<'}^{\enne}(A)$, such that $z <'w$.
In the following, we show that this leads to a contradiction with the hypothesis that $D$ is a basis for $A$ in $K$.

Let $B$ be the set of defaults $A' \ent C'$ in $K$  that are not violated in $z$, i.e., $B = K \backslash V(z)$.
As $z <'w$, by definition of $<'$,  $V(z)  \prec^{MP} V(w)$, that is, there is some $h$ such that $V_h(z) \subset V_h(w)$ and
$V_j(z)=V_j(w)$ for all $j$ such that $k \geq j>h$ (all defaults with no rank are violated both in $z$ and in $w$).
Hence, there is some $E \ent F \in V_h(w) \backslash V_h(z)$, 
 a conditional with rank $h$ which is violated by $w$ and not by $z$).
Clearly, $E \ent F \in B$. On the contrary, $E \ent F \not \in D$, as default $E \ent F$ is violated in $w$ while, by construction, all the defaults in $D$ must be satisfied in $w$ (as $\enne, w \models \tilde{D}$).

To prove that $D \prec^{MP} B$,
we have to show that all defaults $G \ent G'$ with rank $l \geq h$ that belong to $D$ also belong to $B$.
Let $G \ent G' \in D$, then $G \ri G' \in \tilde{D}$ and hence $\enne, w \models G \ri G'$.
As $G \ent G' \not \in V(w)$, then $G \ent G' \not \in V(z)$, since we know that 
$V_h(z) \subset V_h(w)$ and, for all $j$ such that $j>h$, $V_j(z)=V_j(w)$.
Therefore, by definition of $B$, $G \ent G'  \in B$,
and $D \prec^{MP} B$ follows.

%
This contradicts the hypothesis. 
In fact, as $D$ is a basis for $A$ in $K$, 
there cannot be a set of defaults $B$ such that  $\tilde{B} \cup \{A\}$  is consistent and $D \prec^{MP} B$.
We can then conclude that $w \in Min_{<'}^{\enne}(A)$ and that the conditional $A \ent C$ is false in $\enne$.

$(\Leftarrow)$ By contraposition, let us assume that there is an MP-model $\enne$ in ${\cal F}(Min_{RC}(K))$ such that $A \ent C$ is false in $\enne$.
Let $\enne= {\cal F}(\emme)$ for some $\emme =\langle \WW, <, v \rangle  \in Min_{RC}(K)$.
We prove that  $A \ent C$ does not belong to the MP-closure of $K$,  i.e., there is a basis $D$ for $A$ such that $\tilde{D} \cup \{A\} \not \models C$.

As $A \ent C$ is false in $\enne$, there is a world $x \in Min_{<'}^{\enne}(A)$ such that $\enne, x \models \neg C$.
We define $D$ as $K \backslash V(x)$,  the set of conditionals $F \ent E$ in $K$ which are not violated in $x$, i.e., the conditionals $F \ent E$  such that $\enne, x \models \neg F \vee E$.
We have to prove that $D$ is a basis for $A$ in $K$ and that $\tilde{D} \cup \{A\} \not \models C$. 

Let us prove that  $D$ is a basis for $A$.

(1) $\tilde{D} \cup \{A\} $ is consistent. In fact, by construction, $\enne, x \models \neg F \vee E$, for all  $F \ent E$ in $D$. Hence, $\enne, x \models \tilde{D}$
(where $\tilde{D}$ stands for the conjunction $\bigwedge  \tilde{D}$).
Also,  $\enne, x \models A$, as $x \in Min_{<'}^{\enne}(A)$. Therefore, $\enne, x \models \tilde{D} \wedge A$.

(2)  $D$ is maximal w.r.t. $\prec^{MP}$ among the sets of defaults $B \subseteq K$ such that $\tilde{B} \cup \{A\} $ is consistent.
The proof of this point is by contradiction.

By absurd, suppose that there is a set $B$ of defaults in $K$ such that $\tilde{B} \cup \{A\} $ is consistent and $D \prec^{MP} B$.
As the model $\emme  \in Min_{RC}(K)$ is a canonical model of $K$, from the fact that $\tilde{B} \cup \{A\} $ is consistent, we will prove that there must be a world $w \in \WW$ such that $\emme, w \models  \tilde{B} \cup \{A\} $.
From $\enne= {\cal F}(\emme)$ it will also follow that $\enne, w \models  \tilde{B} \wedge \{A\}) $.

 Let us suppose that $A$ has rank $i$ in the rational closure.
We prove that proposition $\tilde{B} \wedge A$, has rank $i$ as well.
We have to exclude that it has a rank higher than $i$.
As $\tilde{B} \cup \{A\} $ is consistent, 
$\not \models \tilde{B} \ri \neg A$. 
Let us refer to the sets $C_i$ in rational closure construction (Section \ref{sez:RC}).
As all defaults in $C_i$ have rank greater or equal than $i$, and $B$ is a basis for a proposition $A$ with rank $i$, 
the materialization of $C_i$ is consistent with $A$ (i.e., $\tilde{C_i} \not \models \neg A$)
and  $C_i \subseteq B$. In fact, $B$ is maximal (w.r.t. the MP-seriousness ordering) among the sets of defaults whose materialization is consistent  with $A$.
Then,  $\not \models \tilde{C_i}  \wedge \tilde{B} \ri \neg A$ and, by propositional reasoning, $\not \models \tilde{C_i}  \ri( \tilde{B} \ri \neg A)$.
By the deduction theorem for the propositional calculus, $\tilde{C_i}  \not \models  \tilde{B} \ri \neg A$, that is, $\tilde{C_i}  \not  \models \neg ( \tilde{B} \wedge A)$.
Using Lehmann and Magidor's equivalence \cite{whatdoes}  (mentioned above), we get $C_i \not \models_\Pe \top \ent  \neg (\tilde{B} \wedge A)$.
Therefore, proposition $\tilde{B} \wedge A$ is not exceptional for $C_i$ and must have rank $i$. 

As $\emme  \in Min_{RC}(K)$ is a minimal canonical model of $K$, the rank of a proposition in a minimal canonical model of $K$ must be the same as its rank in the rational closure of $K$, formula $\tilde{B} \wedge A$ must have rank $i$ in $\emme$.
Therefore, there must be a world $w \in \WW$ in $\emme$
such that $\emme, w \models  \tilde{B} \wedge A$ 
and, similarly, $\enne, w \models  \tilde{B} \wedge A$.
\normalcolor

Let us consider the tuples $\langle D_{\infty}, D_k, \ldots, D_1,D_0 \rangle_D$ 
and $ \langle B_{\infty}, B_k, \ldots, B_1,B_0 \rangle_B$ associated with the sets $D$ and $B$, respectively.
As $D \prec^{MP} B$,
there is some $h$ such that $D_h \subset B_h$ and
$D_j=B_j$ for all $j$ such that $k \geq j>h$.

Let $F \ent G \in B_h \backslash D_h$.  
By construction, $V(x) = K \backslash D$ and, hence, $F \ent G \in V_h(x)$.
On the contrary, $F \ent G \not \in V_h(w)$, as $F \ent G \in B$ and $\enne,w \models \tilde{B}$.
Therefore, $F \ent G \in V_h(x) \backslash V_h(w)$.

We prove that, for all $r \geq h$, $V_r(w) \subseteq V_r(x)$. 
From this, together with $F \ent G \in V_h(x) \backslash V_h(w)$, it follows that 
$V_h(w) \subset V_h(x)$,
and that
$V(w) \prec^{MP} V(x)$.

To show that,  for all $r \geq h$, $V_r(w) \subseteq V_r(x)$,
let $E \ent E' \in V_j(w)$, for some $j \geq h$ ($E \ent E'$ is a default of rank $j$ violated in $w$).
Since all defaults in $B$ are satisfied in $w$, $E \ent E' \not \in B_j$.
As,  for all $e \geq h$, $D_r \subseteq B_r$, $E \ent E' \not \in D_j$ (and $E \ent E' \not \in D$).
Hence, $E \ent E'  \in V(x)$ and, specifically, $E \ent E'  \in V_j(x)$ (as $E \ent E' $ has rank $j$).

Then we can conclude that $V(w) \prec^{MP} V(x)$
 and, by definition of $<'$, that $w <'x$.
As $\enne, w \models A$, this contradicts the hypothesis that $x$ is minimal among the worlds satisfying $A$ in $\enne$,
i.e., $x \in Min_{<'}^{\enne}(A)$. The assumption that $D$ is not maximal w.r.t. $\prec^{MP}$ leads to a contradiction.
This concludes point (2) of the proof.

By points (1) and (2) we have proved that  $D$ is a basis for $A$ in $K$.
We have still to prove that  
 $\tilde{D} \cup \{A\} \not \models C$.
 Remember that $\enne, x \models \neg C$ and that $\enne, x \models \tilde{D} \wedge A$.
Therefore, the propositional valuation $v(x)$ satisfies $\tilde{D} \cup \{A\}$ but falsifies $C$.

We have proved that there is a basis $D$ for $A$ in $K$ such that $\tilde{D} \cup \{A\} \not \models C$. It follows that 
$A \ent C$ does not belong to the MP-closure of $K$. \hfill $\Box$
\end{proof} 
%
%

We conclude this section by showing that the functor ${\cal F}$ in Definition \ref{def:functor} can be reformulated by making it explicit the dependency of the preference relation  $<'$ in $\enne$ from the preference relation $<$ in $\emme$.
Observe that, by Proposition 13 in \cite{AIJ15}, for a model $\emme \in  Min_{RC}(K)$, the ranking function $k_\emme$ associated  with $\emme$ 
is such that, for all formulas $A$ in the language of $K$, $k_\emme(A)=\rf(A)$, where $\rf(A)$ is the rank of $A$ in the rational closure of $K$.
We can therefore give an alternative  equivalent definition of the functor ${\cal F}$ above replacing $\prec^{MP}$ with the preference relation $\prec^{k_\emme}$ defined below, 
which makes  the dependency of $<'$ on ${k_\emme}$ (and thus on $<$) explicit.
\normalcolor

\begin{definition}[$k_\emme$-seriousness ordering]\label{k-order}
Given a finite ranked model $\emme$ of $K$ and a set of defaults $D \subseteq K$, 
 let $k$ be the greatest rank of a world in the finite model $\emme$, and
 let the {\em tuple associated with $D$ in $\emme$}  be the following tuple of subsets of $D$ 
$$\langle D_{\infty}, D_k, \ldots, D_1,D_0 \rangle_D^\emme, $$
where 
  $D_{\infty}$ is the set of defaults $ A \ent C$ in $ D$ such that $A$ has no rank in the model $\emme$, and,
 for all finite $i$, $D_i$ is the set of defaults $ A \ent C$ in $ D$ such that $k_\emme(A)=i$.

Given two sets of defaults $D$ and $B$, 
 {\em $D$ is less serious than $B$ w.r.t. the $k_\emme$-seriousness ordering}, written $D \prec^{k_\emme} B$, if and only if
\begin{center}
$\langle D_{\infty}, D_k, \ldots, D_1,D_0 \rangle_D^\emme$ $\ll\langle B_{\infty}, B_k, \ldots, B_1,B_0 \rangle_B^\emme$. 
\end{center}
where, $\langle D_{\infty}, D_k, \ldots, D_1,D_0 \rangle_D^\emme$ and $ \langle B_{\infty}, B_k, \ldots, B_1,B_0 \rangle_B^\emme$ are,
respectively, the tuples associated with $D$ and with $B$ in the model $\emme$.
\end{definition}
Observe that the definition above uses the same lexicographic order $\ll$ on tuples of sets of defaults used in the definition of the MP-seriousness ordering $\prec^{MP}$, which is  inductively defined just before  Definition \ref{MP-order}.
 Observe also that, for a set of defaults $D$, $\langle D_{\infty}, D_k, \ldots, D_1,D_0 \rangle_D^\emme$ corresponds to the tolerance partitioning of $D$ in system $Z$ \cite{GeffnerAIJ1992}. 

We can then reformulate the functor ${\cal F}$ according to the following Proposition:
\begin{proposition}\label{prop:functor}
Given a ranked interpretation $\emme=\langle \WW, <, v \rangle \in Min_{RC}(K)$
and a preferential interpretation $\enne= \langle \WW, <', v \rangle$
such that $\enne= {\cal F}(\emme)$, then: 
\begin{align}\label{defi:<inBP}
 x <' y \mbox{\ \ {\bf \em iff} } V(x)  \prec^{k_\emme} V(y)
 \end{align}
\end{proposition}
\begin{proof}
The proof is trivial since, as a consequence of Proposition 13 in \cite{AIJ15}, 
the rank of a conditional $A \ent B$ in the rational closure of $K$
corresponds to the rank of $A$ in any minimal canonical model of $K$ (and hence in  $\emme$). In particular, when $A \ent B$ has no rank in the rational closure,
$A$ has no rank in $\emme$.
For any set of defaults $D$ (and in particular for $V(x)$ and for $V(y)$), the tuple $\langle D_{\infty}, D_k, \ldots, D_1,D_0 \rangle_D$ associated with $D$ in the definition of MP-seriousness ordering (Definition \ref{MP-order}) exactly corresponds to the tuple $\langle D_{\infty}, D_k, \ldots, D_1,D_0 \rangle_D^\emme$
associated with $D$ in $\emme$ according to the Definition \ref{k-order} of $k_\emme$-seriousness ordering, when $\emme$ is a minimal canonical model of $K$.
Therefore, the $k_\emme$-seriousness ordering $ \prec^{k_\emme}$  and the MP-seriousness ordering  $\prec^{MP}$ coincide,
when $\emme \in Min_{RC}(K)$.
Under this condition, $V(x)  \prec^{k_\emme} V(y)$ iff $V(x)  \prec^{MP} V(y)$, and
the claim follows.
\hfill $\Box$
\end{proof}
This last reformulation of the functor ${\cal F}$, besides making the dependency of $<'$ on $<$ more evident, also allows 
to point at the differences
between MP-models introduced above and the notion of bi-preference interpretation in \cite{arXiv_Skeptical_closure}.  

 As we have mentioned above, a bi-preference interpretation is a quadruple $\langle \WW,<,  <', v \rangle$, where $<$ is a modular preference relation while $<'$ is a (irreflexive and transitive) preference relation, which additionally satisfies the {\em if part} of Condition (\ref{defi:<inBP}).  
Minimization was then used therein to define the minimal BP-models of $K$ characterizing the MP-closure. 
In minimal BP-models  $\langle \WW,<,  <', v \rangle$ of a knowledge base $K$, essentially, $\langle \WW,<, v \rangle$ corresponds to a minimal model 
in $Min_{RC}(K)$ (a model of the rational closure of $K$),  while $<'$ has to satisfy the if part of Condition (\ref{defi:<inBP}).  
In this paper, we have chosen a different and simpler route, and we have 
defined a functor ${\cal F}$ which builds an MP-model $\enne$ of a knowledge base $K$ starting from a minimal canonical model $\emme \in Min_{RC}(K)$ 
without requiring any further minimization step.
As a difference, the {\em if and only if } Condition (\ref{defi:<inBP}) 
is used to bind the preference relation $<'$ in $\enne$ to the preference relation $<$ in $\emme$, when $\enne \in {\cal F}(\emme)$.
As $\emme$ is already a minimal model of $K$, no further minimization is needed to define the MP-model $\enne$.

%
 \normalcolor 

%
%
%

\subsection{Some properties and relation with lexicographic closure}

Observe that there is a strong correspondence between the ordering on models defined by Condition (\ref{order_on_worlds}) in Definition \ref{def:functor}  
and the semantics of  lexicographic closure given by Lehmann \cite{Lehmann95}, 
where the seriousness ordering $\prec$ based on the lexicographic order on the tuples of numbers
 is used to order the propositional models as follows:
\begin{align}\label{cond:preference_prop_int}
m \prec m' \mbox{\ \ {\bf \em iff} } V(m)  \prec V(m')
\end{align}
 where the seriousness ordering $\prec$ is the one based on the lexicographic order on the tuples of numbers in Definition \ref{def:seriousness_ord_LC}.
``This modular ordering on models defines a modular preferential model that in turn defines a consequence relation"   \cite{Lehmann95}, the lexicographic closure of the knowledge base.
It was proven by Lehmann \cite{Lehmann95} (Theorem 2) 
that the relation $\prec$  on propositional models is finer than $\ll$, a modular relation defining a model of the rational closure. 
In Proposition \ref{Prop:coarser} in Section  \ref{sec:construction} 
we have proven that $ \prec^{MP}$ is  at least as coarse as $ \prec$.
As a consequence, all conditionals belonging to the MP-closure of a knowledge base $K$ also belong to the lexicographic closure of $K$.

Let   ${\cal LC}_K$ be  the set of conditionals belonging to the lexicographic closure of $K$, and
    ${\cal RC}_K$ be  the set of conditionals belonging to the rational closure of $K$.
Remember that ${\cal MP}_K$ is  the set of conditionals belonging to the MP-closure of $K$.
We have already proven that ${\cal MP}_K \subseteq {\cal LC}_K$.
Example \ref{example-differenza_MC_LC} shows a knowledge base $K$ for which the inclusion is strict,
i.e., ${\cal MP}_K \subset {\cal LC}_K$. In fact, conditional (\ref{conditional_in_LC-MP}) belongs to ${\cal LC}_K$ but not to  ${\cal MP}_K$.
Therefore the converse inclusion ($ {\cal LC}_K \subseteq {\cal MP}_K$) does not hold.

We will now exploit the correspondence result in Theorem  \ref{MP-characterization} to show that 
 the conditionals belonging to the rational closure of a knowledge base $K$ also belong to the MP-closure of $K$, i.e., that
 ${\cal RC}_K \subseteq {\cal MP}_K$.

\begin{proposition}\label{prop:RC_MP}
Given a knowledge base $K$,
if $A \ent C$ is in the rational closure of $K$, than  $A \ent C$ is in the MP-closure of $K$.
\end{proposition}
\begin{proof}
Let $A \ent C$ be in the rational closure of $K$, then, by the correspondence  Theorem \ref{rat_closure_modelli_minimali} 
$\emme \models A \ent C$, for all $\emme$ in $Min_{RC}(K)$.
To show that $A \ent C$ is in the MP-closure of $K$, we prove that for all  $\enne \in {\cal F}(Min_{RC}(K))$, $\enne \models A \ent C$.

Let us consider any $\enne \in {\cal F}(Min_{RC}(K))$.
It must be $\enne= {\cal F}(\emme)$ for some $\emme =\langle \WW, <, v \rangle$  in $Min_{RC}(K)$.
From the hypothesis we know that that $\emme \models A \ent C$ and, hence, for all $w \in Min_<^{\emme}(A)$, $\emme, w \models C$.

We prove that $\enne \models A \ent C$, i.e., that for all $w' \in Min_{<'}^{\enne}(A)$, $\enne, w' \models C$.
From $w' \in Min_{<'}^{\enne}(A)$, we know that $w' \in \WW$ and that $\enne, w' \models A$.
But then $\emme, w' \models A$, as $\WW$ and $v$ are the same in $\emme$ and $\enne$:
the set of the worlds satisfying $A$ in $\enne$ is the same as the set of the worlds satisfying $A$ in $\emme$.
By Proposition \ref{MP_stronger _than_RC},  $< \subseteq <'$ and, hence, $Min_{<'}^{\enne}(A) \subseteq Min_{<}^{\emme}(A)$.
Thus, $w' \in Min_{<}^{\emme}(A)$.
As $\emme$ satisfies $A \ent C$, then $\emme, w' \models C$.
Again, as $\WW$ and $v$ are the same  in $\emme$ and $\enne$, $\enne, w' \models C$.
We can then conclude that $\enne \models A \ent C$.
As this is true for all $\enne \in {\cal F}(Min_{RC}(K))$, this proves the claim.
\hfill $\Box$
\end{proof}
The next corollary follows:

\begin{corollary}
 ${\cal RC}_K \subseteq {\cal MP}_K  \subseteq {\cal LC}_K$.
\end{corollary}
To show that, for some knowledge base $K$, the strict inclusion  ${\cal RC}_K \subset {\cal MP}_K$ holds, it is enough to observe that, in Example \ref{example-Student},
conditional $\mathit{Employee \wedge Student \ent  Young}$ is in the MP-closure of $K$, but not in the rational closure of $K$.

Notice that
Lehmann in his semantics considers a single model including the propositional interpretations $m$, ordered according to the preference relation $\prec$,
where, for any propositional interpretations $m$ and $m'$,
$m \prec m'$ is defined by condition (\ref{cond:preference_prop_int}).
On the contrary, in our semantics we consider a set of  models ${\cal F}(Min_{RC}(K))$, where the models $\enne  \in {\cal F}(Min_{RC}(K))$ may differ w.r.t. the set of worlds $\WW$ and the valuation function $v$.
Each model $\enne = \langle \WW, <', v \rangle$ in $ {\cal F}(Min_{RC}(K))$, however, is canonical (according to Definition \ref{canonical_model}), and it must contain at least one world for each propositional truth assignment (over the language ${\cal L}_{K,Q}$) which is compatible with $K$.
We will see that each model $\enne$ in $ {\cal F}(Min_{RC}(K))$ satisfies exactly the same conditionals.
\normalcolor
%
%
%
%
This is a consequence of the fact that the relative ordering of two worlds  $x,y \in \WW$ in any interpretation $\enne \in  {\cal F}(Min_{RC}(K))$  
only depends on the set of defaults 
violated in  $x$  and 
in $y$. 
In turn, the defaults violated in $x$ and in $y$ only depend on the valuation of the formulas $B \in {\cal L}_{K,Q}$  in $x$ and $y$. 
\normalcolor

The following lemma shows that the preference relation $<'$ among worlds in a model $\enne$ is determined by the valuation of the propositions in the language ${\cal L}_{K,Q}$ of the knowledge base $K$ and of a query $Q$ (a conditional).
The lemma will be used to prove the next proposition.

%

Given two models $\enne'=\langle \WW', <', v' \rangle $ and $\enne''=\langle \WW'', <'', v'' \rangle $, and two worlds $x \in \WW'$ and $y \in \WW''$,
we say that {\em the valuations  $v'(x)$ and $v''(y)$ coincide over the language ${\cal L}_{K,Q}$},
when,  for all $B \in {\cal L}_{K,Q}$,  
$\enne', x \models B$ iff $\enne'', y \models B$.
That is, the propositional formulas of the language ${\cal L}_{K,Q}$ have the same truth value in $x$ and in $y$.
\begin{lemma}\label{same valuation}
Given two models $\enne'=\langle \WW', <', v' \rangle $ and $\enne''=\langle \WW'', <'', v'' \rangle $ such that $\enne', \enne'' \in {\cal F}(Min_{RC}(K))$,
let $x',y' \in \WW'$ and $x'',y'' \in \WW''$. 

If the valuations  $v'(x')$ and $v''(x'')$ coincide over  language ${\cal L}_{K,Q}$,
and the valuations $v'(y')$ and $v''(y'')$ coincide over language ${\cal L}_{K,Q}$,
then, 
\begin{center}
$x' <' y'$ iff $x'' <'' y''$.
\end{center}

\end{lemma}

\begin{proof}
%

Let $x',y' \in \WW'$ and $x'',y'' \in \WW''$ and assume that  valuations  $v'(x')$ and $v''(x'')$ coincide over language ${\cal L}_{K,Q}$,
and that   $v'(y')$ and $v''(y'')$ coincide over  language ${\cal L}_{K,Q}$.
Hence, for all $B \in {\cal L}_{K,Q}$,  

$\enne', x' \models B$ iff $\enne'', x'' \models B$,  and, 

$\enne', y' \models B$ iff $\enne'', y'' \models B$.\\
Then the defaults of $K$ which are violated in $x'$ and in $x''$ are the same ones, i.e., $V(x')=V(x'')$ and,
similarly, $V(y')=V(y'')$.
As $\enne \in {\cal F}(Min_{RC}(K))$, by construction, condition (\ref{order_on_worlds}) holds and it must be:  
\begin{center}
$x' <' y'$ iff $V(x') \prec^{MP} V(y')$.
\end{center}
As $\enne' \in {\cal F}(Min_{RC}(K))$, again by condition  (\ref{order_on_worlds}), 
\begin{center}
$x'' <'' y''$ iff $V(x'') \prec^{MP} V(y'')$.
\end{center}
From the two equivalences above, given that $V(x')=V(x'')$ and  $V(y')=V(y'')$, it follows that 
\begin{center}
$x' <' y'$ iff $x'' <'' y''$.
\end{center}
\hfill $\Box$
\end{proof}
The property that the preference relation $<'$ among worlds in each  $\enne \in {\cal F}(Min_{RC}(K))$ is only determined by the valuation of the propositional formulas of the language ${\cal L}_{K,Q}$ at the worlds, 
can be used to prove that all the models in ${\cal F}(Min_{RC}(K))$ satisfy the same conditional formulas.

\begin{proposition}\label{prop:same_conditionals}
Given two models $\enne'=\langle \WW', <', v' \rangle $ and $\enne''=\langle \WW'', <'', v'' \rangle $ such that $\enne', \enne'' \in {\cal F}(Min_{RC}(K))$ and a query $Q= A \ent C$. Then, 
\begin{center}
$\enne' \models A \ent C$ \  $\iff$  \ $\enne'' \models A \ent C$.
\end{center}
\end{proposition}
\begin{proof}
($\Ri$) Assume $\enne' \models A \ent C$. We want to show that $\enne'' \models A \ent C$, i.e., for all $w \in \WW''$, if $w \in Min_{<''}^{\enne''}(A)$ then  $\enne'', w \models C$.
Suppose, by absurd, that, for some $w \in \WW''$, $w \in Min_{<''}^{\enne''}(A)$ but  $\enne, w \not \models C$.
Then there is truth assignment $v_0:  \mathit{ATM_{K,Q}} \longrightarrow \{true, false\}$ such that
$\enne'', w \models B$ iff $v_0(B)=true$, for all formulas $B \in {\cal L}_{K,Q}$.
Clearly, $v_0$ is compatible with $K$, as $\enne''$ is a model of $K$, by Proposition \ref{enne_is_model_of_K}.

As $\enne'$ is a canonical model of $K$ ($\enne' \in {\cal F}( \emme)$ for some canonical model $\emme \in Min_{RC}(K)$),
there must be a world $w' \in \WW'$ such that $\enne', w' \models B$ iff $v_0(B)=true$, for all formulas $B \in {\cal L}_{K,Q}$.
In particular, $\enne', w' \models A$ and $\enne', w' \not \models C$.
We show that $w' \in Min_{<'}^{\enne'}(A)$. 

If $w' \not \in Min_{<'}^{\enne'}(A)$, there must be a $z' \in \WW'$ such that $z' <' w'$ and $\enne', z' \models A$. 
This leads to a contradiction with the assumption that $w \in Min_{<''}^{\enne''}(A)$.
As $\enne''$ 
is equal to ${\cal F}(\emme'')$ for some minimal canonical model $\emme''$ of $K$, there must be an element $z'' \in \WW''$,
whose valuation corresponds to the same truth assignment over ${\cal L}_{K,Q}$ of $z'$ in $\enne'$.
In particular, $z'$ satisfies $A$ and, hence, $z''$ satisfies $A$ in $\enne''$.
By Lemma \ref{same valuation}, as $v''(w)$ and $v'(w')$  coincide over the language ${\cal L}_{K,Q}$,
and  $v''(z'')$ and $v'(z')$ coincide over the language ${\cal L}_{K,Q}$, then 
\begin{center}
$z' <' w'$ iff $z'' <'' w$.
\end{center}
Then, it must be $z'' <'' w$, which contradicts the assumption that $w \in Min_{<''}^{\enne''}(A)$.
Hence, it must be that $w \in min_{<'}^{\enne'}(A)$.
As $\enne', w \not \models C$, we can conclude that $\enne' \not \models A \ent C$, which contradicts the hypothesis.

 ($\Leftarrow$)  Assuming $\enne'' \models A \ent C$, we can prove that $\enne' \models A \ent C$ as in previous case, replacing $\enne'$ with $\enne''$ and vice-versa, as $\enne'$ and $\enne''$ are any two models in ${\cal F}(Min_{RC}(K))$.
%
\hfill $\Box$
\end{proof}
As a consequence of this result, all the models in ${\cal F}(Min_{RC}(K))$ must satisfy the same conditionals. 
Hence, for  a satisfiable knowledge base $K$
we can take {\em any}  (arbitrarily chosen) model $\enne_{MP}$ in $ {\cal F}(Min_{RC}(K))$ as the semantic characterization of the MP-closure, thus coming up with a semantics quite similar to the model-theoretic semantics by Lehmann \cite{Lehmann95}.
In particular, by the characterization Theorem \ref{MP-characterization},
\begin{align*}
\enne_{MP} \models A \ent C \mbox{ iff } A \ent C \in {\cal MP}_K. 
\end{align*}
Then,  MP-closure is the consequence relation defined by a preferential model, $\enne_{MP}$.
{\color{black} As an alternative, given the canonical model constructions for rational closure developed 
 by Booth and Paris \cite{BoothParis98} and by Pearl \cite{PearlTARK90}, mentioned in Section \ref{sez:RC}, 
$\enne_{MP}$ could have been chosen as the preferential model obtained by applying the functor ${\cal F}$ to the canonical model of rational closure.}
As a consequence of the property, proved by Lehmann and Magidor \cite{whatdoes}, that  any preferential model defines a preferential consequence relation,
it follows that the MP-closure is a preferential consequence relation.

\begin{corollary}
MP-closure is a preferential consequence relation
\end{corollary}
We have already seen in Section \ref{sec:construction} that the MP-closure, on the contrary, is not a rational consequence relation as it violates the property of Rational Monotonicity.

\section{Rational supersets of the MP-closure}  \label{sec:rational_relation}

We have seen that the  MP-closure does not define a rational consequence relation and that the models characterizing the MP-closure
by Theorem \ref{MP-characterization} are preferential models, in which the preference relation is not necessarily modular. Indeed, the MP-closure does not satisfy the property of rational monotonicity, as we have seen from Example \ref{exa:counterexa_RM}. 
 While the adequacy of the property of Rational Monotonicity might be subject of discussion, in this section 
we consider rational supersets of the MP-closure.  We know that the lexicographic closure is such a superset, by Corollary 3.
Another rational superset of the MP-closure (and of the rational closure) can be obtained by applying to the MP-closure a simplified version of the semantic construction proposed by Lehmann and Magidor \cite{whatdoes}, in their semantic characterization of the rational closure,
to transform a well-founded preferential model into a ranked model. 
This construction allows a rational consequence relation 
to be defined from the MP-closure.
 In this section we investigate its relations with  the lexicographic closure and show that 
these two closures do not coincide and are incomparable  i.e., none of them is stronger than the other one. \normalcolor

Given any MP-model of $K$, i.e., a preferential model $\enne = \langle \WW, <', v \rangle$ such that  $\enne \in {\cal F}(Min_{RC}(K))$, we want to minimally extend the preference relation $<'$ to a ranked preference relation $<^\Ra$, so to define a ranked model $\enne^\Ra = \langle \WW, <^\Ra, v \rangle$ of $K$.
%
The idea is that the height of a world in a preferential structure has to be regarded as its rank,
and we introduce two alternative (but equivalent) definitions of the rank of a world.
The first one (in Definition \ref{def:rank_in_MP}) is based on the definition of rank introduced in Section \ref{sez:RC} for ranked models (Definition \ref{definition_rank_prop}), which extends to preferential models\footnote{ Observe that, for a finite preferential model, the preference relation among worlds can be represented as a finite directed acyclic graph,
where the nodes are the worlds and the edges are given by the relation $<$.
As in the modular case,  the notion of the longest path between two nodes is well-defined, and for this reason we can exploit the same Definition \ref{definition_rank_prop} as for ranked models in Section \ref{sez:RC}. In the following proofs we will sometimes refer to the length of such paths (or chains of worlds).}. The second one is inspired to Lehmann and Magidor's construction.


\begin{definition} \label{def:rank_in_MP}
Given  $\enne = \langle \WW, <', v \rangle$ a preferential interpretation, we let $\enne^\Ra = \langle \WW, <^\Ra, v \rangle$
be the ranked interpretation where: the rank $k_{\enne^\Ra}(w)$ of a world $w \in \WW$ in $\enne^\Ra$
is the length of 
a maximal chain $w_0 <'  \ldots <' w$ from $w$ to a minimal world $w_0$ 
(i.e., a world for which there is no $z \in \WW$ such that $z<'w_0$). 
In particular, each minimal world $w_0$ is given rank $k_{\enne^\Ra}(w_0)=0$.
\end{definition}
As mentioned above, in their model-theoretic description of the rational closure, \linebreak Lehmann and Magidor introduce
``a construction that transforms a preferential model $W$ into a ranked model $W'$ letting all the states of $W$ sink as low as they can respecting the order of $W$, i.e., ranks the states in $W$ by their height in $W$'' \cite{whatdoes}. For finite knowledge bases over a finite language, we consider a simplified version of their iterative construction, and exploit it for transforming a finite preferential model $\enne$ of the MP-closure into a ranked one.
This provides an alternative inductive definition of the rank $k_{\enne^R}(w)$ of a world $w \in \WW$, 
which is equivalent to Definition \ref{def:rank_in_MP}. 

Given a preferential interpretation $\enne = \langle \WW, <', v \rangle$, let  $U_i$ be
the set of  worlds in $\WW$ with rank $i$, defined inductively as follows: 
\begin{align*}
U_0=& \{ w \mid w \in \WW \mbox{ and  there is no } z\in \WW, \mbox{ such that }z<'w\}\\
U_i=& \{ w \mid w \in \WW - (U_1 \cup \ldots U_{i-1}) \mbox{ and there is no } z \in \WW- (U_1 \cup \ldots U_{i-1}) \\
&\mbox{ \ \ \ \ \ \ \ \ \ such that } z<'w \};
\end{align*} 
One can easily prove, by induction on the rank $i$, that:
\begin{proposition}\label{ranks_MP_models} 
Given a preferential interpretation  $\enne = \langle \WW, <', v \rangle$, 
the ranked interpretation $\enne^R = \langle \WW, <^R, v \rangle$ can be defined  by letting, for all $w \in \WW$,
$k_{\enne^R}(w)=i$ if $w \in U_i$.
\end{proposition}
\normalcolor

To show that when $\enne$ is an MP-model of $K$ $\enne^\Ra$ is a model of $K$  as well, 
we need the following proposition.

\begin{proposition}\label{prop:<' subseteq <^R}
Given a preferential interpretation $\enne = \langle \WW, <', v \rangle$ and the ranked interpretation $\enne^\Ra = \langle \WW, <^{\Ra}, v \rangle$  defined as above,
it holds that $<' \subseteq <^\Ra$.
\end{proposition}
\begin{proof}
The proof is immediate. Just observe that, if  $x <' y$ holds for some $x,y \in \WW$ in $\enne$, in the associated graph,
 the length of the longest path  $w_0 <'  \ldots <' y$ from $y$ to a minimal world $w_0$
 must be larger than
the length of the longest path  $w'_0 <'  \ldots <' x$ from $x$ to a minimal world $w'_0$.
Hence, $k_{\enne^\Ra}(x) < k_{\enne^\Ra}(y)$, and then $x <^\Ra y$.
\hfill $\Box$
\end{proof}
From Proposition \ref{prop:<' subseteq <^R} and Proposition \ref{prop:<' subseteq <''},
it is easy to prove that when $\enne$ is an MP-model of a knowledge base $K$, then  $\enne^\Ra$ is a model of $K$. 
\begin{proposition}\label{prop:N^R model of K}
Let $\enne$ be an MP-model of $K$ and let $\enne^\Ra$ be the ranked interpretation as defined above.
$\enne^\Ra$  satisfies a superset of the conditionals satisfied by $\enne$ and is a model of $K$.
\end{proposition}
\begin{proof} 
Let $\enne= \langle \WW, <', v \rangle$ be an MP-model of $K$.
As both $\enne$ and $\enne^\Ra$ are preferential interpretations and, by Proposition \ref{prop:<' subseteq <^R}, $<' \subseteq <^\Ra$,
we can conclude, as a consequence of Proposition  \ref{prop:<' subseteq <''},
that all the conditionals $A \ent C$ satisfied by $\enne$ are satisfied by $\enne^\Ra$ as well.
In particular, as $\enne$ is a model of $K$, $\enne$ satisfies all the conditionals in $K$ and $\enne^\Ra$ does as well.
\hfill $\Box$
%
%
\end{proof}

We will now consider the rational consequence relation containing all the conditionals satisfied in $\enne^\Ra$, for some $\enne \in Min_{RC}(K)$.
Let us define a new functor ${\cal F}^\Ra$ as follows:
\begin{align} \label{eq:functorF^R}
{\cal F}^\Ra(\emme)= \{\enne^\Ra \mid \enne= {\cal F}(\emme)\}.
\end{align}
In particular, given a model $\emme  \in Min_{RC}(K)$, the functor ${\cal F}^\Ra$ first applies the functor ${\cal F}$ to $\emme$ and then transforms the resulting preferential model $\enne$ into a ranked model $\enne^\Ra$.

For a satisfiable knowledge base $K$,
we have shown in Proposition \ref{prop:same_conditionals}  that all models $\enne$ in $ {\cal F}(Min_{RC}(K))$ satisfy the same conditionals and, essentially, define the same preference relation among worlds. 
As a consequence, a single model in $ {\cal F}(Min_{RC}(K))$,  call it $\enne_{MP}$, 
already provides a semantic characterization of the MP-closure of a satisfiable knowledge base $K$,
in the following sense:
\begin{align*}
\enne_{MP} \models A \ent C \mbox{ iff } A \ent C \in {\cal MP}_K. 
\end{align*}
Let us now consider the ranked model  $\enne_{MP}^\Ra \in  {\cal F}^\Ra (Min_{RC}(K))$, obtained from $\enne_{MP}$ according to the construction above. 
As a consequence of the result proven by  Lehmann and Magidor \cite{whatdoes}, that the consequence relation defined by any ranked model is rational,
the model $\enne_{MP}^\Ra$ defines a rational consequence relation, that we call  ${\cal MP}_K^\Ra$, where
\begin{align*}
{\cal MP}_K^\Ra = \{ A \ent C \mid \;  \enne_{MP}^\Ra \models A \ent C \}.
\end{align*}
%
The definition of the rational consequence relation ${\cal MP}_K^\Ra$ does not depend on the choice of the model $\enne_{MP}$ in $ {\cal F}(Min_{RC}(K))$,
as shown by the following proposition.
\begin{proposition} \label{prop:modelli  N^R equivalenti}
For all $\enne_1, \enne_2 \in {\cal F}(Min_{RC}(K))$, $\enne_1^\Ra$ and $ \enne_2^\Ra$ satisfy the same conditionals. 
\end{proposition}

\begin{proof}
We prove that, for any conditional $A \ent C \in K$,
if $\enne_1^\Ra \models A \ent C$, then $\enne_2^\Ra \models A \ent C$. 
As $\enne_1$ and $ \enne_2 $ are any two models in $ {\cal F}(Min_{RC}(K))$, the claim follows.

Let $\enne_1 = \langle \WW_1, <_1, v_1 \rangle$ and $\enne_2 = \langle \WW_2, <_2, v_2 \rangle$
and assume that $\enne_1^\Ra \models A \ent C$.
We show that $\enne_2^\Ra \models A \ent C$, i.e., for all $w \in \WW_2$, if $w \in Min_{<_2}^{\enne_2^\Ra}(A)$, then $\enne_2^\Ra, w \models C$. 

Suppose, by absurd, that, for some $w \in \WW_2$, $w \in Min_{<_2}^{\enne_2^\Ra}(A)$ and  $\enne_2^\Ra, w \not \models C$.
Then there is truth assignment $v_0:  \mathit{ATM_{K,Q}} \longrightarrow \{true, false\}$ such that
$\enne_2^\Ra, w \models B$ iff $v_0(B)=true$, for all formulas $B \in {\cal L}_{K,Q}$.
Clearly, $v_0$ is compatible with $K$, as $\enne_2^\Ra$ is a model of $K$, by Proposition \ref{prop:N^R model of K}.

As $\enne_1$ is a canonical model of $K$
 ($\enne_1 \in {\cal F}( \emme)$ for some canonical model $\emme \in Min_{RC}(K)$),
$\enne_1^\Ra$is a canonical model of $K$ as well, and
there must be a world $w' \in \WW_1$ such that $\enne_1^\Ra, w' \models B$ iff $v_0(B)=true$, for all formulas $B \in {\cal L}_{K,Q}$.
In particular, $\enne_1^\Ra, w' \models A$ and $\enne_1^\Ra, w' \not \models C$.
We show that $w' \in Min_{<_1}^{\enne_1^\Ra}(A)$. 

If $w' \not \in Min_{<_1}^{\enne_1^\Ra}(A)$, there must be a $z' \in \WW_1$ such that $z' <_1 w'$ and $\enne_1^\Ra, z' \models A$. 
We show that we arrive at a contradiction with the assumption that $w \in Min_{<_2}^{\enne_2^\Ra}(A)$.

In fact, as $\enne_2$ 
is equal to ${\cal F}(\emme'')$ for some minimal canonical model $\emme''$ of $K$, there must be an element $z'' \in \WW_2$,
whose valuation corresponds to the same truth assignment over ${\cal L}_{K,Q}$ of $z'$ in $\enne_1^\Ra$.
In particular, as $z'$ satisfies $A$ in $\enne_1^\Ra$,  $z''$ satisfies $A$ in $\enne_2^\Ra$.
By Lemma \ref{same valuation}, as $v_2(w)$ and $v_1(w')$  coincide over the language ${\cal L}_{K,Q}$,
and  $v_2(z'')$ and $v_1(z')$ coincide over the language ${\cal L}_{K,Q}$, then 
\begin{center}
$z' <_1 w'$ iff $z'' <_2 w$.
\end{center}
Then, it must be $z'' <_2 w$, which contradicts the assumption that $w \in Min_{<_2}^{\enne_2^\Ra}(A)$.
Hence, it must be that $w' \in Min_{<_1}^{\enne_1^\Ra}(A)$.
As $\enne_1^\Ra, w' \not \models C$, we can conclude that $\enne_1^\Ra, w'  \not \models A \ent C$, which contradicts the hypothesis.
%
\hfill $\Box$
\end{proof}

Next proposition proves that the rational consequence relation ${\cal MP}_K^\Ra$ is a superset of ${\cal MP}_K$, the set of conditionals which belong to the MP-closure of $K$. 

\begin{proposition}
For a satisfiable knowledge base $K$,
${\cal MP}_K \subseteq {\cal MP}_K^\Ra$.
\end{proposition}
\begin{proof}
We have seen that ${\cal MP}_K$ is the set of  conditionals satisfied by an MP-model $\enne_{MP}$ and that, by construction, 
$ {\cal MP}_K^\Ra$ is the set of conditionals satisfied by the ranked model $\enne_{MP}^\Ra$.
As, by Proposition \ref{prop:N^R model of K},  $\enne_{MP}^\Ra$ satisfies a superset of the conditionals satisfied by $\enne_{MP}$,
we can conclude that ${\cal MP}_K \subseteq {\cal MP}_K^\Ra$.
%
\hfill $\Box$
\end{proof}
To see that the converse inclusion $ {\cal MP}_K^\Ra \subseteq {\cal MP}_K $ does not hold, let us consider again Example \ref{exa:counterexa_RM}, showing a knowledge base $K$ such that $ {\cal MP}_K^\Ra$ contains a conditional  which is not in $ {\cal MP}_K$. 
\begin{example}
Let $K$ be the knowledge base:
\begin{quote}
1. $\mathit{ Student \ent Merry}$\\
2. $\mathit{Student \ent  Young}$\\
3. $\mathit{Adult \ent Serious}$\\
4. $\mathit{Student \wedge Adult \ent (\neg Young \wedge \neg Merry) \vee \neg Serious}$
\end{quote}
It is possible to see that the conditional $\mathit{Student \wedge Adult \ent Young}$ which (as we have seen in Example \ref{exa:counterexa_RM}) does not belong to ${\cal MP}_K$,
on the contrary, belongs to ${\cal MP}_K^\Ra$.
Indeed, if we consider any preferential model $\enne_{MP} = \langle \WW, <^R v \rangle$ in  ${\cal F} (Min_{RC}(K))$,  the ranked model $\enne_{MP}^\Ra=  \langle \WW, <^R v \rangle$ satisfies  
$\mathit{Student \wedge Adult \ent Young}$.
Let us explain why.

First, there must be a world $w \in \WW$ whose propositional valuation is
 $\mathit{ v(w)= }$ $\mathit{ \{Student, Adult, Merry, Young \}}$, as this valuation is compatible with $K$ and the model $\enne_{MP}$ is a canonical model.
The world $w$ satisfies all the conditionals in $K$ except  conditional 3: $\mathit{ Adult \ent Serious}$.
Thus the set $V(w)$ of the conditionals violated by $w$ is associated with the tuple $\langle \emptyset, \emptyset, \{3\} \rangle_{V(w)}$ (conditional 3 has rank 0).
All the worlds $z\in \WW$ such that $z <' w$ must have the associated tuple  $\langle \emptyset, \emptyset, \emptyset \rangle_{V(z)}$
that is, they violate no conditional, and they must have rank 0 in $\enne_{MP}^\Ra$ (i.e., $k_{\enne_{MP}^\Ra}(z)=0$).
Therefore, $w$ must have rank $1$ in $\enne_{MP}^\Ra$ (i.e., $k_{\enne_{MP}^\Ra}(w)=1$).
One can see that 
$w$ is the unique minimal world in $\enne_{MP}^\Ra$ satisfying $\mathit{Student \wedge Adult}$.
In fact, the world $x \in \WW$ such that $\mathit{ v(x)= \{Student, Adult, Serious \}}$ (which falsifies $Young$) has rank $2$ in $\enne_{MP}^\Ra$.
While in the preferential model $\enne_{MP}$ both $x$ and $w$ are minimal worlds satisfying  $\mathit{Student \wedge Adult}$
(i.e., $x,w \in Min_<^{\enne_{MP}} (Student \wedge Adult)$),
so that $\enne_{MP} \not \models \mathit{Student \wedge Adult \ent Young}$,
in the ranked model $\enne_{MP}^\Ra$ the world $x$ has rank $2$ and is not a minimal world satisfying  $\mathit{Student \wedge Adult}$. Indeed $x$ violates two conditionals (conditionals 1 and 2) and it has a rank higher than the worlds falsifying just one of these two defaults. In particular, a world $y$ such that  $\mathit{ v(y)= \{Student, Adult, Merry,}$ $\mathit{ Serious \}}$ must be as well in $\WW$ 
and $y<' x$, as $y$ violates a (strict) subset of the conditionals violated by $x$ (conditional 2, but not 1). World $y$ has rank $1$ in $\enne_{MP}^\Ra$.
Therefore, $x$ is not minimal among the worlds satisfying $Student \wedge Adult$ in $\enne_{MP}^\Ra$
and it holds that $\mathit{Min_{<^\Ra}^{\enne_{MP}^\Ra} (Student \wedge Adult)}=\{w\}$. Thus, $\enne_{MP}^\Ra \models \mathit{Student \wedge Adult \ent Young}$.
\end{example}

In this example, the model $\enne_{MP}^\Ra$ looks  similar to a model of the lexicographic closure (according to the semantics by Lehmann),
as the rank of a world in $\enne_{MP}^\Ra$, as shown above,  also depends on the number of defaults it satisfies.
One may wonder whether the rational consequence relation ${\cal MP}_K^\Ra$  coincides with the lexicographic closure of $K$, ${\cal LC}_K$.
The next example shows that this is not the case, 
providing a  knowledge base $K$ which falsifies both the inclusions  ${\cal MP}_K^\Ra \subseteq {\cal LC}_K$ and 
${\cal LC}_K \subseteq {\cal MP}_K^\Ra$.

\begin{example}
Let $K$ be the knowledge base:
\begin{quote}
1. $\mathit{ A \ent E\wedge F}$\\
2. $\mathit{A \ent  E\wedge F \wedge E}$\\
3. $\mathit{A \ent  E\wedge F \wedge E \wedge F}$\\
4. $\mathit{C \ent \neg E}$\\
5. $\mathit{C \ent \neg F}$
\end{quote}
The knowledge base $K$ is satisfiable and, in the rational closure construction, all the conditionals have rank $0$. 
In the lexicographic closure, $D=\{1,2,3\}$ is the only basis for $A \wedge C$, and the conditional $\mathit{ A \wedge C \ent E}$ is in ${\cal LC}_K$.
In the MP-closure, there are two bases  for $A \wedge C$: $D=\{1,2,3\}$ and $D'=\{4,5\}$, and neither  $\mathit{ A \wedge C \ent E}$ nor
$\mathit{ A \wedge C \ent \neg E}$ are in ${\cal MP}_K$.
We show that the ${\cal MP}_K^\Ra$ contains $\mathit{ A \wedge C \ent \neg E}$.

Consider an MP-model $\enne=\langle \WW, <' v \rangle$ of $K$ and the model $\enne^\Ra$.
Consistently with the fact that in the MP-closure there are two bases $D$ and $D'$ for $A\wedge C$,  $\enne$ contains two minimal worlds  $x$ and $y$ satisfying $A \wedge C$, where
$v(x)=\{A,C,E,F\}$ and $v(y)=\{A,C\}$, so that $x$ violates defaults 4 and 5 ($V(x)= \langle \emptyset, \{4,5\} \rangle$) and $y$ violates defaults 1, 2 and 3 ($V(y)= \langle \emptyset, \{1,2,3\} \rangle$) .

All worlds $z$ such that $z <'y$ do not violate any default and have rank $0$ in $\enne^\Ra$, so that $y$ has  rank $1$ in $\enne^\Ra$.
In fact, there is no propositional assignment (and no world in $\enne$)  that may violate a subset of the defaults 1, 2 or 3, but not all of them.

Concerning $x$, there are two propositional assignments, $\{C,E\}$ and $\{C,F\}$violating a (strict) subset of the defaults violated by $x$.
As $\enne$ is canonical there are two worlds $z$ and $w$ such that
$v(z)=\{C,E\}$, $v(w)=\{C,F\}$,  $z<'x$ and $w<'x$.
In fact, $z$ violates default  4 ($V(z)= \langle \emptyset, \{5\} \rangle$) and $w$ violates default 5 ($V(w)= \langle \emptyset, \{4\} \rangle$) .
Both $z$ and $w$ have rank $1$ in $\enne^\Ra$, so that $x$ has rank $2$.

While $x$ and $y$ are both minimal $A \wedge C$ worlds in the MP-model $\enne$, on the contrary, $y <^\Ra x$ in $\enne^\Ra$, as $x$ has rank $2$
while $y$ has rank $1$. 
It can be seen that 
all the minimal worlds satisfying $A \wedge C$ in $\enne^\Ra$ have the same truth assignment as $y$, and falsify $E$ and $F$.
Hence, $\mathit{ A \wedge C \ent \neg E}$ is satisfied by $\enne^\Ra$ and is in ${\cal MP}_K^\Ra$.
On the contrary, $\mathit{ A \wedge C \ent E}$ is not in ${\cal MP}_K^\Ra$.

As $\mathit{ A \wedge C \ent \neg E}$ is not in the lexicographic closure of $K$,
 we can conclude that $\mathit{ A \wedge C \ent \neg E}$ is in ${\cal MP}_K^\Ra - {\cal LC}_K$ while 
$\mathit{ A \wedge C \ent E}$ is in ${\cal LC}_K- {\cal MP}_K^\Ra$.
\end{example}
We have seen previously that ${\cal RC}_K \subseteq {\cal MP}_K  \subseteq {\cal LC}_K$, and that, for a satisfiable knowledge base $K$,
$ {\cal MP}_K  \subseteq {\cal MP}_K^\Ra$.
The last example shows that ${\cal MP}_K^\Ra$ and ${\cal LC}_K$ are incomparable. 

\begin{corollary}
There is some knowledge base $K$ such that neither   ${\cal MP}_K^\Ra \subseteq {\cal LC}_K$ nor ${\cal LC}_K \subseteq {\cal MP}_K^\Ra$ hold.
\end{corollary}
We observe that, while the lexicographic closure of $K$ and the MP-closure have a  definition in terms of maxiconsistent sets the rational extension of the MP-closure ${\cal MP}_K^\Ra$ has only a semantic definition.  We have introduced it to show that there is at least another rational consequence relation, different from the lexicographic closure and incomparable with it, which is still a natural superset of the MP-closure (and of the rational closure).  

\bigskip

To conclude this section we show  that, among the ranked models whose modular preference relation extends $<'$,  $\enne^\Ra$ is the least one, with   
respect to $<_{FIMS}$. 
\begin{proposition}\label{proposition:FIMS}
Let $\enne= \langle \WW, <', v \rangle$ in $ {\cal F}(Min_{RC}(K))$. 
For all  ranked interpretations $\enne^*= \langle \WW, <^*, v \rangle$ such that $<' \subseteq  <^*$ and $\enne^* \neq \enne^\Ra$, it holds $\enne^\Ra <_{FIMS} \enne^*$. 
\end{proposition}
\begin{proof}
We prove that, for all $x \in \WW$, 
\begin{align}\label{cond:M^R-M^*}
k_{\enne^\Ra}(x) \leq k_{\enne^*}(x).
\end{align}
The proof is by induction on the rank $i$ of $x$ in $\enne^\Ra$. Let $k_{\enne^\Ra}(x)=i$. \\
For $i=0$, condition (\ref{cond:M^R-M^*}) holds trivially, as $k_{\enne^*}(x) \geq 0$. \\
For $i>0$, $i$ is the length of the longest paths from $x$ to a minimal world $w_0$.
Let $w_0 <' \ldots <'w_{i-1}<'x$ be 
such a path.
Then $w_{i-1} \in \WW$  and $k_{\enne^\Ra}(w_{i-1})=i-1$.
The rank of $w_{i-1}$ cannot be less than $i-1$ as, otherwise, this would not be a path of length $i$.
The rank of $w_{i-1}$ cannot be more than $i-1$ as, otherwise, there would be a path longer than $i$ from $x$  to some minimal world.

By inductive hypothesis, $k_{\enne^\Ra}(w_{i-1}) \leq k_{\enne^*}(w_{i-1}) $.
Then, $k_{\enne^*}(w_{i-1}) \geq i-1$.
As $w_{i-1}<'x$ and $<' \subseteq  <^*$,  then $w_{i-1}<^*x$,
Therefore, the rank of $x$ in $\enne^*$ must be higher than the rank of $w_{i-1}$,
i.e., $k_{\enne^*}(x) \geq i$. Thus,  $k_{\enne^\Ra}(x) \leq k_{\enne^*}(x)$ and
condition (\ref{cond:M^R-M^*}) follows.

Let us now consider any interpretation $\enne^*$ such that $<' \subseteq  <^*$  and $\enne^* \neq \enne^\Ra$.
As the two interpretations may only differ in the ranking, it must be that , for some world $w$, 
$ k_{\enne^\Ra}(w) \neq k_{\enne^*}(w)$.
From condition (\ref{cond:M^R-M^*}) it follows that $ k_{\enne^\Ra}(w) < k_{\enne^*}(w)$.
As for all the other worlds $w' \in \WW$, by  (\ref{cond:M^R-M^*}), $k_{\enne^\Ra}(w') \leq k_{\enne^*}(w')$,
then $\enne^\Ra <_{FIMS} \enne^*$.
\hfill $\Box$
\end{proof}

Proposition \ref{proposition:FIMS} proves that ${\cal N}^\Ra$, the ranked model resulting from the transformation of a model $\enne$ of the MP-closure, is $<_{FIMS}$ minimal among the ranked interpretations $\enne^*$
whose preference relation $<^*$ extends $<'$. 
Observe that the ranked interpretations $\enne^*$ whose preference relation$<^*$ extends $<'$ are (ranked) models of ${\cal MP}_K$, by Proposition \ref{prop:<' subseteq <''}, as  $<' \subseteq  <^*$ and $\enne$ is a model of ${\cal MP}_K$.
Furthermore, it is possible to prove that, if ${\enne}^*$ is a model of ${\cal MP}_K$, then $<' \subseteq  <^*$.

\begin{proposition} \label{prop:N_star}
Let $\enne= \langle \WW, <', v \rangle$ be in $ {\cal F}(Min_{RC}(K))$.
For all  ranked interpretations $\enne^*= \langle \WW, <^*, v \rangle$, 
if $<' \subseteq  <^*$ does not hold, then there is some conditional which is satisfied in $\enne$ and not in $\enne^*$.
\end{proposition}
\begin{proof}
Let $x$ and $y$ be two worlds in $\WW$ such that $x <'y$ but $ x\not <^* y$.
Let $\alpha_x$ be the conjunction of all the atomic formulas satisfied in $x$ and 
 $\alpha_y$ the conjunction of all the atomic formulas satisfied in $y$ (remember that we are considering a finite language).
As $x <'y$, $x$ and $y$ must violate different defaults (i.e., $V(x) \neq V(y)$). Then they must have different propositional valuations, 
so that $\alpha_x$ anf $\alpha_y$ are not propositionally equivalent.

The conditional $\alpha_x \vee \alpha_y \ent \alpha_x$ is satisfied in $\enne$. In fact, $x$ is a minimal world satisfying $\alpha_x \vee \alpha_y$,
while $y$ is not. The same would hold for all other worlds having the same valuations as $x$ and $y$ (if any), as the preference $<'$ only depends on the defaults violated by the worlds, which depend on their valuations.

As $x <^*y$ does not hold in $\enne^*$, either  $y <^*x$ holds  in  $\enne^*$, or $x$ and $y$ are not comparable in  $\enne^*$. In both cases, the conditional  $\alpha_x \vee \alpha_y \ent \alpha_x$ is not satisfied in $\enne^*$. In the first case, $y$ is a minimal world satisfying $\alpha_x \vee \alpha_y$, and it does not satisfy $\alpha_x$.
In the second case, both $x$ and $y$ are minimal worlds satisfying $\alpha_x \vee \alpha_y$, and $y$ does not satisfy $\alpha_x$. \hfill $\Box$
\end{proof}
From the previous propositions, we get the following result, showing that $\enne^\Ra$ is $<_{FIMS}$ minimal among all the ranked models of the MP-closure:
 \begin{proposition}
Let $\enne= \langle \WW, <', v \rangle$ be in $ {\cal F}(Min_{RC}(K))$. 
For all  ranked interpretations $\enne^*= \langle \WW, <^*, v \rangle$ such that $\enne^*$ is a model of ${\cal MP}_k$ and $\enne^* \neq \enne^\Ra$, it holds that $\enne^\Ra <_{FIMS} \enne^*$. 
 \end{proposition}
\begin{proof}
As $\enne^*$ is a model of ${\cal MP}_k$, then $\enne^*$ satisfies all the conditionals satisfied by $\enne$, as  ${\cal MP}_k$ is the consequence relation defined by $\enne$. By the contrapositive of Proposition  \ref{prop:N_star},  $<' \subseteq  <^*$ holds.
Then, by Proposition \ref{proposition:FIMS}, we can conclude that $\enne^\Ra <_{FIMS} \enne^*$. \hfill $\Box$
\end{proof}
By this result, which is in agreement with the more general characterization result by Lehmann and Magidor \cite{whatdoes},
we can say that, from a semantic point of view, ${\cal MP}^\Ra_k$ can be regarded as the rational closure of ${\cal MP}_k$.

\normalcolor
%
%
%
%

\normalcolor

\section{Comparison with related approaches} 
\label{sec:further_issues}

In this section we relate the MP-closure and its rational extension with other related proposals.
%
In Section \ref{sec:iterating}, we use the semantic construction developed in the previous section to establish some relationships with the multi-preference semantics proposed by Gliozzi \cite{GliozziNMR2016,GliozziAIIA2016} for description logics with typicality.
In Section  \ref{sec:relevant_closure} we compare the MP-closure with the Relevant Closure proposed by Casini et al. \cite{Casini2014} for description logics, 
in Section \ref{sec: confronti_Brewka_ARS} we compare it with Basic Preference Descriptions \cite{Brewka04} and with system ARS \cite{IsbernerRitterskamp2010};
in Section \ref{sec:Cond_Ent}  with Conditional Entailment  \cite{GeffnerAIJ1992} and other preferential approaches and in Section  \ref{sec:confronti_con_DLs} with further related proposals introduced for DLs.

\normalcolor

 \subsection{Relations with the multipreference semantics} \label{sec:iterating}
 

The MP-closure has been proposed in  \cite{Multipref_arXiv2018}  as a construction which is a sound approximation of the multipreference semantics \cite{GliozziAIIA2016} for the description logic $\alc$.  Here we recall this semantics and compare with it, abstracting away from the peculiarities of description logics and considering only the propositional part (with no ABox, no individual constants, no roles and no universal and existential restrictions).
The idea of the multipreference semantics is to define a refinement of the rational closure in which preference with respect to 
specific aspects is considered.
It is formulated in terms of enriched models, which also consider preference relations $<_{A_i}$
associated with different aspects, where each relation $<_{A_i}$ refers (only) to the defeasible inclusions of the form $\tip(D) \sqsubseteq A_i$
(corresponding here to conditionals of the form $D \ent A_i$). 
The idea of having different preference relations, associated to different typicality operators, was already proposed by Gil \cite{fernandez-gil} to define a multipreference formulation of the logic $\alctmin$ \cite{AIJ13} with multiple typicality operators, a logic which exploits a different minimal model semantics w.r.t. rational closure semantics.
In \cite{GliozziAIIA2016}, on the contrary, a refinement of rational closure was considered and with a single typicality operator. 

 Each preference relation $ <_{A_i}$ only regards the defaults of the form $C \ent A_i$ in the knowledge base,
and, essentially, it ranks the worlds (propositional valuations) based on the defaults of the form $C \ent A_i$ they satisfy.
The aim of the semantics in  \cite{GliozziAIIA2016} was to define a refinement of rational closure semantics 
(and, specifically, of the models in $Min_{RC}(K)$) 
in which the different preferences $<_{A_i}$ are combined in a single modular preference relation
 $<$.
 To this purpose,  the following condition is introduced for $<$:
 \begin{quote}
(a) {\bf If}  $x <_{A_i} y$, for some $A_i$, and there is no $A_j$ such that $y <_{A_j} x$, 
{\bf then} $x < y$.
\end{quote}
The intended meaning of $x <_{A_i} y$ is that $x$ satisfies some default for $A_i$ which is violated by $y$. More precisely, $<_{A_i} $ is the preference relation in a ranked model of a knowledge base $K_i$ containing all and only the defaults of the form $D \ent A_i$ in $K$. In the minimal ranked models $\emme = \langle \WW, <_{A_i}, v \rangle$ of $K_i$ (minimal according to Definition \ref{Preference between models in case of fixed valuation}), $x <_{A_i} y$  has precisely the meaning that $x$ satisfies some default for $A_i$ which is violated by $y$.

Condition (a) alone, however, is too weak to capture models of  rational closure and a {\em specificity condition} 
was added to define enriched models. Here, we refer to the definition of {\em S-enriched rational models} by Giordano and Gliozzi \cite{Multipref_arXiv2018,Ecsqaru19}, which is slightly stronger than the one in  \cite{GliozziAIIA2016} (although both of them lead to refinements of rational closure), and reformulate it in the propositional case replacing typicality inclusions $\tip(D) \sqsubseteq C$
with conditionals  $D \ent C$.

 \begin{definition}[S-Enriched rational models of K]\label{def-SenrichedmodelR} 
$\emme = \langle \WW, <_{A_1}, \ldots,<_{A_n}, <, v \rangle$ is a strongly enriched model of $K$ if the following conditions hold:
\begin{itemize}
 \item  $ \langle \WW, <, v \rangle$ is a ranked model of $K$ (as in Section 2, Definition \ref{semantica_rational}); 
 \item   for all $C \ent A_i \in K$,
for all $w \in \WW$, if $w \in Min_{<_{A_i}}^\emme(C)$ then $\emme,w \models A_i$ 
 and 
\item the preference relation $<$ satisfies both the condition $(a)$ above, and the following {\em specificity condition}:
\begin{align*}
x <y  \mbox{ if \ \ } 
(i) & \mbox{ $y$ violates some default satisfied by $x$ and} \\
(ii) & \mbox{ for each $C_j \ent D_j \in K$  violated by $x$ and not by $y$,}\\
 &  \mbox{ there is a $C_k \ent D_k \in K$  violated by $y$ and not by $x$, }\\
  &  \mbox{ and } k_{\emme}(C_j) < k_{\emme}(C_k).
\end{align*}
\end{itemize}

\end{definition}
In $(i)$ and $(ii)$ 
the ranking function $k_\emme$ is the ranking function 
of model $\emme$ itself
and the intended meaning of the specificity condition is that preference should be given to  worlds that falsify less specific defaults (defaults with lower rank in $\emme$).
Defaults violated by $x$ are less serious than defaults violated by $y$, as formula $C_k$ is more specific than $C_j$.

It is easy to see that conditions $(i)$ and $(ii)$  in the specificity condition above (together) are equivalent  to the condition 
$$V(x) \prec^{k_\emme} V(y),$$
where $ \prec^{k_\emme} $ is the $k_\emme$-seriousness ordering in Definition \ref{k-order}.
As a consequence, one can reformulate S-enriched models by reformulating the {\em specificity condition} as: 
$$x<y \mbox{ if }V(x) \prec^{k_\emme} V(y).$$

 A further simplification to the notion of S-enriched models comes from the fact that
the semantics in  \cite{GliozziAIIA2016,Multipref_arXiv2018} considers the {\em minimal S-enriched models},  among all the S-enriched models of $K$, which are
obtained by first minimizing the  $<_{A_i}$ and then minimizing $<$ (as done for the ranked models of the rational closure), in this order, thus giving preference to models with lower ranks. 
It was proved in  \cite{Multipref_arXiv2018}  (Proposition 1 therein) that, in minimal S-enriched models, the specificity condition is strong enough to enforce condition (a).
As a consequence, one can simplify the definition of S-enriched rational models from the beginning, by removing condition (a) as well as the preference relations $ <_{A_1}, \ldots, <_{A_n}$, 
thus starting from the following simplified notion of enriched model. 

 \begin{definition}[simplified-enriched models of K]\label{def-Simp-enrichedmodelR} 
A simplified-enriched model of $K$ is a ranked model $\emme = \langle \WW, <, I \rangle$ of $K$ (according to Definition \ref{semantica_rational} in Section \ref{sez:RC}) such that the preference relation $<$ satisfies the condition 
$$x<y \mbox{ if }V(x) \prec^{k_\emme} V(y)$$
\end{definition}
With this simplification, the minimal S-enriched models in \cite{Multipref_arXiv2018} would correspond to the $<_{FIMS}$-minimal simplified-enriched rational models of $K$, 
where $<_{FIMS}$ minimality is  precisely the same notion of minimality used in the semantic characterization of the rational closure  (see Definition \ref{Preference between models in case of fixed valuation} in Section \ref{sez:RC}).
Furthermore, minimal simplified enriched models also satisfy property $(a)$. 
 Thus the multipreference semantics in \cite{Multipref_arXiv2018} collapses into a semantics without multiple preferences (a result that actually was implied by Proposition 1).

%

We prove that $<_{FIMS}$-minimal simplified-enriched models of $K$  coincide with the $<_{FIMS}$-minimal ranked models $\emme$ 
such that ${\cal F}^\Ra(\emme)= \emme$. 
It is easy to prove that the ranked models of $K$ such that ${\cal F}^\Ra(\emme)= \emme$ are simplified-enriched models of $K$.
\begin{proposition}\label{s-enriched-fixedpoint}
Given a ranked model $\emme=\langle \WW, <, v \rangle$ of $K$, 
if  ${\cal F}^\Ra(\emme)= \emme$  then  $\emme$ is a simplified-enriched model of $K$.
\end{proposition}

\begin{proof}
Let $\emme$ be a ranked model of $K$ such that ${\cal F}^\Ra(\emme)= \emme$.
Then $\emme=\enne^\Ra$ for some $\enne = {\cal F}(\emme)$ where, based on the formulation of functor  ${\cal F}$ and Proposition \ref{prop:functor}, 
$\enne= \langle \WW, <', v \rangle$ and 
\begin{align*}
 x <' y \mbox{\ \ {\bf \em iff} } V(x)  \prec^{k_\emme} V(y).
 \end{align*}
 Let $\enne^\Ra=  \langle \WW, <^{\Ra}, v \rangle= {\cal F}^\Ra(\emme)$.
 By Proposition \ref{prop:<' subseteq <^R}, $<' \subseteq <^\Ra$, and then:
 \begin{align*}
 x <^\Ra y \mbox{\ \ {\bf \em if} } V(x)  \prec^{k_\emme} V(y).
 \end{align*}
 But, as $\emme= {\cal F}^\Ra(\emme)$, the relation  $<$ in $\emme$ must be equal to $<^\Ra$ and, therefore:
  \begin{align*}
 x < y \mbox{\ \ {\bf \em if} } V(x)  \prec^{k_\emme} V(y).
 \end{align*}
Thus $\emme$ is a simplified-enriched model of $K$.
 \hfill $\Box$
 \end{proof}
 From the proposition above, the next corollary follows.
\begin{corollary}\label{coroll:relation_con_multipref}
If $\emme$ is a $<_{FIMS}$-minimal canonical ranked model of $K$ 
such that ${\cal F}^\Ra(\emme)= \emme$ then $\emme$ is  a $<_{FIMS}$-minimal canonical simplified-enriched model of $K$.
\end{corollary}
Although the converse of Proposition  \ref{s-enriched-fixedpoint} does not hold, the converse of Corollary \ref{coroll:relation_con_multipref} 
can be proved, and we have the following result.

  \begin{proposition}\label{prop:relation_con_multipref} 
If $\emme$ is  a $<_{FIMS}$-minimal canonical simplified-enriched model of $K$ then 
$\emme$ is a $<_{FIMS}$-minimal canonical ranked model of $K$ 
such that ${\cal F}^\Ra(\emme)= \emme$.
\end{proposition}
\begin{proof}
 Let $\emme=\langle \WW, <, v \rangle$ be a $<_{FIMS}$-minimal canonical simplified-enriched model of $K$. 
 By Definition \ref{def-Simp-enrichedmodelR},
 $$x<y \mbox{ \bf \em if }V(x) \prec^{k_\emme} V(y).$$
 
Let $\enne = \langle \WW, <', v \rangle=  {\cal F}(\emme)$. 
By  Proposition \ref{prop:functor}, 
\begin{align*}
 x <' y \mbox{\ \ {\bf \em iff} } V(x)  \prec^{k_\emme} V(y) \ \ \ \ (*)
 \end{align*}
Therefore, $x <' y$ implies $x<y$ and, hence, $<' \subseteq <$.

Consider the ranked model $\enne^\Ra=  \langle \WW, <^{\Ra}, v \rangle$.
By Proposition \ref{proposition:FIMS},  $\enne^{\Ra} <_{FIMS} \enne^*$  
for all the ranked interpretations $\enne^*= \langle \WW, <^{*}, v \rangle$ such that $\enne^*\neq \enne^{\Ra}$ and 
$<' \subseteq <^*$.
As we have just proved that $<' \subseteq <$ (and from the hypothesis $\emme$ is a ranked interpretation), it must be the case that $\enne^{\Ra} <_{FIMS} \emme$,
unless $\emme=\enne^{\Ra}$.

Furthermore, we know by Proposition \ref{prop:<' subseteq <^R} that $<' \subseteq <^{\Ra}$.
Therefore,  from condition (*) above, we can conclude that
 $$x<^{\Ra}y \mbox{ \bf \em if }V(x) \prec^{k_\emme} V(y),$$
 that is, $\enne^{\Ra}$ satisfies the specificity condition.

We can then conclude that $\enne^\Ra$ is a canonical simplified-enriched model of $K$. In fact:
$\enne^\Ra$ is a model of $K$ by Proposition \ref{prop:N^R model of K}, as $\emme \in Min_{RC}(K)$ and $\enne$ is an MP-model of $K$;
$\enne^\Ra$ is canonical model as $\emme$ is a canonical model in $Min_{RC}(K)$;
$\enne^{\Ra}$ satisfies the specificity condition.

As $\enne^\Ra$ is a canonical simplified-enriched model of $K$,
we must conclude that $\emme=\enne^{\Ra}$, otherwise, we have that  $\enne^{\Ra} <_{FIMS} \emme$,
which would contradict the hypothesis that $\emme$ is a $<_{FIMS}$-minimal canonical simplified-enriched model of $K$.
\hfill $\Box$
\end{proof} \normalcolor

As a consequence of Corollary \ref{coroll:relation_con_multipref} and Proposition \ref{prop:relation_con_multipref}
we can conclude that the $<_{FIMS}$-minimal canonical simplified-enriched models of $K$ coincide with the $<_{FIMS}$-minimal canonical ranked models which are fixed-points of the operator ${\cal F}^\Ra$.
This provides an alternative characterization of 
the multipreference semantics.
However, 
whether entailment in the multipreference semantics defines a rational consequence relation or not, 
and whether it can be computed by iterating the operator ${\cal F}^\Ra$, 
is still to be understood and will be subject of future work. 

\normalcolor

In \cite{Multipref_arXiv2018, Ecsqaru19} it was shown 
that the MP-closure is a sound approximation of the multipreference semantics, 
that is, the typicality inclusions (the conditionals) that follow from the MP-closure hold in all the minimal canonical S-enriched models of the KB.  
This result extends to the 
$<_{FIMS}$-minimal canonical ranked models of $K$ such that ${\cal F}^\Ra(\emme)= \emme$.

 \subsection{Relations with Relevant Closure} \label{sec:relevant_closure}
The relevant closure was developed  by Casini et al. \cite{Casini2014}  as a proposal for defeasible reasoning in description logics
to overcome the inferential weakness of rational closure. It is based on the idea of relevance of subsumptions to a query,
where relevance is determined on the basis of justifications, 
which are minimal sets of sentences responsible for a conflict.
Any sentence occurring in some justification is potentially relevant for resolving the conflict.

In a defeasible description logic, the knowledge bases contain a set of defeasible subsumptions  $D \defincl C$, a DBox ${\cal D}$,
corresponding to a set of conditionals $D\ent C$, and a set of classical subsumptions $D \sqsubseteq C$, a TBox ${\cal T}$,
which have to be satisfied by all the elements of the domain in any model.
When evaluating a query $C \defincl D$, one has to compute the $C$-justifications w.r.t. ${\cal D}$, that is, the minimal sets of defaults ${\cal J} \subseteq {\cal D}$ making $C$ exceptional (or, supporting $\neg C$). 
The idea is that, for each $C$-justification ${\cal J}$, some defeasible subsumption occurring in ${\cal J}$ 
is to be  removed  from ${\cal D}$ for consistency with $C$, 
and it is convenient to remove first the defeasible subsumptions with lower ranks in the $C$-justifications.
This is done by the Relevant Closure algorithm.
 
For a given query $C \defincl D$, the algorithm receives as input the ranking in the rational closure  of the defeasible subsumptions in ${\cal D}$,
and a set $R$ of the defeasible subsumptions which are {\em relevant} to the query,
i.e., the set of  defeasible subsumptions which are eligible for removal during the execution of the Relevant Closure algorithm.
The algorithm determines from ${\cal D}$ a new set of defeasible subsumptions ${\cal D}'$, by removing from ${\cal D}$, 
rank by rank, starting from the lower rank $0$, all the subsumptions in $R$ with that rank, until the remaining set  of (non-removed) defeasible subsumptions ${\cal D}'$  is consistent with ${\cal T}$ and with $C$. For the pseudocode of this algorithm, we refer to Algorithm 2 in \cite{Casini2014}.

In the Basic Relevant closure, the set $R$ of relevant defeasible subsumptions is the union $\bigcup {\cal J}_j$ of all $C$-justifications ${\cal J}_j$ w.r.t. ${\cal D}$
where, by Corollary 1  in \cite{Casini2014},  ${\cal J}$ is a {\em $C$-justification} if it is an inclusion-minimal subset of ${\cal D}$ such that 
${\cal T} \models \overline{\cal J} \sqsubseteq \neg C$
($ \overline{\cal J} $ being the materialization of the defeasible subsumptions in ${\cal J}$ and $\models$ entailment in the underlying description logic). 
At the end of the iteration phase, a set ${\cal D}'\subseteq {\cal D}$ of defeasible subsumptions   consistent with  ${\cal T}$ and with $C$ is obtained, 
and ${\cal D}'$ is used, together with ${\cal T}$,  to check whether or not $C \sqsubseteq D$ follows from ${\cal D}'$ and ${\cal T}$ (i.e.,  whether ${\cal T} \models \tilde{\cal D}' \sqcap C \sqsubseteq D$).

The Minimal Relevant closure exploits exactly the same algorithm as the basic Relevant closure, but it takes $\bigcup {\cal J}_j^{min}$,
the union of all sets ${\cal J}_j^{min}$ containing the conditionals with lowest rank in each $C$-justification ${\cal J}_j$,
 as the set $R$ of relevant defaults which are eligible for removal  (instead of $\bigcup {\cal J}_j$).
The Basic Relevant Closure is weaker than the Minimal Relevant Closure, and  the Minimal Relevant Closure is weaker than the lexicographic closure  \cite{Casini2014}.

For a comparison with MP-closure, let us consider the case when the TBox ${\cal T}$ is empty and ${\cal D}$ is a set of conditionals, so that the knowledge base $K$
is just a set of conditionals, as before (i.e, $K={\cal C}$).
In the following, we transpose the definition of the basic and minimal Relevant Closure to the propositional case.

Let $K$ be a knowledge base and   $C \ent D$ a query.
A {\em $C$-justification} w.r.t. $K$ is an (inclusion) minimal subset ${\cal J}$ of $K$ such that $\models \bigwedge \tilde{\cal J} \ri \neg C$,
 that is, $\tilde{\cal J} \cup \{  C\}$ is inconsistent,
where $ \tilde{J} $ is the materialization of the conditionals in $J$, as in Section \ref{sez:Lex-closure}, and $\models$ is logical consequence in the propositional calculus.
Let $\bigcup {\cal J}_j$ be the union of all $C$-justifications w.r.t. $K$. 
The algorithm exploits the ranking of conditionals computed by rational closure of $K$.

Given a query  $C \ent D$, and  $R=\bigcup {\cal J}_j$,
the Basic Relevant closure algorithm, for each rank $i$ (in the rational closure of $K$) starting from $0$, removes from $K$ all the defaults with rank $i$ occurring in $R$, until the remaining set of conditionals ${\cal D}'$  is consistent with $C$, i.e., $\tilde{\cal D}' \cup \{C \}$ is consistent in 
the propositional calculus
(at least one conditional has been removed from any $C$-justification $ {\cal J}_j$). 
A conditional $C \ent D$ is in the Basic Relevant Closure of $K$ if $\models (\bigwedge \tilde{\cal D}' \wedge C) \ri D $ (in the propositional calculus).

As before, the Minimal Relevant Closure algorithm differs from the previous one only in that it takes $R$ as $\bigcup {\cal J}_j^{min}$  the union of all sets ${\cal J}_j^{min}$, where ${\cal J}_j^{min}$ is the set of   conditionals with lowest rank in the $C$-justification ${\cal J}_j$. 

Let us consider again Example \ref{example-differenza_MC_LC}, the one in which the lexicographic closure comes to the conclusion that typical employed students, like typical students, are young and do not pay taxes (which appears to be too bold). 
We can see that neither the Basic Relevant Closure nor the Minimal Relevant Closure conclude $\mathit{Employee \wedge Student  \ent}$ $\mathit{ Young \wedge \neg Pay\_Taxes}$.
\begin{example}
The knowledge base $K''$ contains the conditionals:
\begin{quote}
1. $\mathit{ Student \ent \neg Pay\_Taxes}$\\
2. $\mathit{Student \ent  Young}$\\
3. $\mathit{Employee \ent  \neg Young \wedge  Pay\_Taxes}$\\
4. $\mathit{Employee \wedge Student  \ent  Busy}$
\end{quote}
 There are two justifications of the exceptionality of $\mathit{Employee \wedge Student}$ w.r.t. $K''$,
namely ${\cal J}_1=\{1,3\}$ and ${\cal J}_2=\{2,3\}$, and their union $\bigcup {\cal J}_j=\{1,2,3\}$, used in the basic relevant closure algorithm, contains only conditionals with rank 0, which are all removed as responsible of the exceptionality of $\mathit{Employee \wedge Student}$ at the first iteration stage (for rank $0$).
The set of remaining conditionals is then  ${\cal D}'=\{4\}=\{\mathit{Employee \wedge Student  \ent  Busy}\}$, so that 
$\bigwedge \tilde{\cal D}'=\{\mathit{Employee \wedge Student  \ri  Busy}\}$
and 
$$\not \models (\bigwedge \tilde{\cal D}' \wedge \mathit{Employee \wedge Student} ) \ri (\mathit{ Young \wedge \neg Pay\_Taxes}).$$
Therefore, the conditional $\mathit{Employee \wedge Student  \ent}$ $\mathit{ Young \wedge \neg Pay\_Taxes}$
is not in the Basic Relevant closure of $K''$.
The result is the same for the Minimal Relevant closure, as ${\cal J}_1$ and ${\cal J}_2$ only contain conditionals with rank $0$ and, therefore, ${\cal J}_1={\cal J}_1^{min}$ and ${\cal J}_2={\cal J}_2^{min}$. 
\end{example}
The Basic and Minimal Relevant closure as well as the MP-closure are all more cautious than the lexicographic closure.
The next example shows that the MP-closure is neither equivalent to Basic nor to Minimal Relevant closure. 
 
 
\begin{example}  \label{example-differenza_MC_RC}
Let $K$  be the knowledge base containing  conditionals:
\begin{quote}
1. $\mathit{ Italian \ent Residence\_in\_Italy}$\\ 
2. $\mathit{German \ent Residence\_in\_Germany}$\\
3. $\mathit{Residence\_in\_Italy  \wedge \neg  Has\_Residence \ent \bot}$\\ 
4. $\mathit{Residence\_in\_Germany \wedge \neg  Has\_Residence \ent \bot}$\\
5. $\mathit{Residence\_in\_Italy  \wedge Residence\_in\_Germany \ent \bot}$
\end{quote}
 Italians normally have a residence in Italy and Germans normally have a residence in Germany.  Those who have a residence in Italy have residence. Those who have a residence in Germany have residence.
 It is not the case that somebody has residence both in Germany and in Italy.
 Observe that the last three conditionals have no rank in the rational closure, and they represent properties which no model of $K$ can violate.
 
There is a unique $\mathit{Italian \wedge German}$-justification w.r.t. $K$, namely ${\cal J}=\{1,2,5\}$. In fact,  ${\cal J}$ is a minimal set of conditionals such that $ \models \bigwedge \tilde{{\cal J}} \ri \neg (\mathit{Italian \wedge German})$. 
%
 As defaults 1 and 2  in ${\cal J}$ have rank $0$ in the rational closure, while default 5 has an infinite rank, the basic Relevant closure algorithm first removes conditionals $1$ and $2$ from ${\cal D}$.
The resulting set of defaults ${\cal D}'=\{3,4,5\}$ is consistent with  $\mathit{Italian \wedge German}$, and nothing else needs to be removed.
The conditional $\mathit{ Italian \wedge}$ $\mathit{ German \ent}$ $\mathit{  Has\_}$  $\mathit{Residence}$ is not in the basic Relevant closure of $K$
as $\not \models \mathit{ (\bigwedge \tilde{\cal D}' \wedge Italian \wedge German) }$ $\mathit{  \ri Has\_Residence}$.

Concerning the Minimal Relevant closure, as defaults 1 and 2 have rank $0$, the set of the defaults in ${\cal J}$ with lowest rank is ${\cal J}^{min}=\{1,2\}$. Therefore, also in the  Minimal Relevant Closure construction, justifications 1 and 2 are both removed from ${\cal D}$, and the conditional  $\mathit{ Italian \wedge German \ent  Has\_Residence}$ is not  in the Minimal Relevant Closure of $K$.

As a difference, the defeasible inclusion $\mathit{ Italian \wedge German \ent Has\_Residence}$ is  in the MP-closure of $K$ as well as in the Lexicographic closure of $K$, which both have two bases, $\{1\}$ and $\{2\}$.
\end{example}
%
Notice also that, for each justification $ {\cal J}_j$, the minimal Relevant closure always removes all the defaults in ${\cal J}_j^{min}$. Indeed, the defaults in ${\cal J}_j^{min}$ have all the same rank and they have to be removed all together at the same iteration of the algorithm (at iteration $i$ if they have rank $i$). It cannot be the case that some default in ${\cal J}_j^{min}$ is removed from ${\cal D}$ but not all of them. This observation will be useful in the following.

The example above shows that 
the MP-closure is different from both the Basic Relevant Closure and the Minimal Relevant Closure, and that the MP-closure cannot be weaker than such closures. We prove that the Minimal Relevant Closure is a subset of the MP-closure.


\begin{proposition}
Let $K$ be a set of conditionals.
If $C \ent A$ is in the Minimal Relevant Closure of $K$, then $C \ent A$ is in the MP-closure of $K$.
\end{proposition}
\begin{proof}
The proof is by contraposition.
If $C \ent A$ is not in the MP-closure of $K$, then there must be an MP-basis ${\cal B}$ for $C$ such that $\tilde{{\cal B}} \cup \{C\}  \not \models A$.
By definition of MP-basis, ${\cal B}$ is maximal w.r.t. the MP-seriousness ordering among the subsets of $K$ consistent with $C$.

We show that, when executing the Minimal Relevant Closure algorithm for the goal $C \ent A$, each conditional $d= E \ent F\not \in {\cal B}$  is removed from $K$ by the Minimal Relevant Closure algorithm, so that the resulting set of defaults ${\cal D}'$ must be a subset of the MP-basis ${\cal B}$.
As a consequence, $\tilde{{\cal D}}' \cup \{C\}  \not \models A$ and, by the deduction theorem,  $\not \models \tilde{\cal D}' \wedge C \ri A $. We can then conclude that $C \ent A$ is not in the Minimal Relevant Closure of $K$.

First observe that, as ${\cal B}$ is an MP basis for $C$, it is a maximal set of defaults such that $\tilde{{\cal B}} \cup \{C\}$ is consistent.
Then, if $d \not \in {\cal B}$,  $\tilde{{\cal B}} \cup \{ \tilde{d}\}  \cup \{ C\}$ is inconsistent and there must be some $C$-justification ${\cal J}$ such that $d \in{\cal J}$.
We show that, in particular, there must be a $C$-justification ${\cal J}_r$ such that  $d \in {\cal J}_r^{min}$, i.e., $d$ must have the lowest rank among the conditionals in ${\cal J}_r$.

Suppose, by absurd, that a $C$-justification ${\cal J}_r$ with $d \in {\cal J}_r^{min}$ does not exist.
Then for all $C$-justifications ${\cal J}_{s_h}$  such that $d \in {\cal J}_{s_h}$ (with $h=1,\ldots,t$), the rank of $d$ is not the lowest among the ranks of the conditionals in $ {\cal J}_{s_h}$, and (for each $h$) there must be another conditional 
$d_{s_h} \in  {\cal J}_{s_h}$ such that $\rf(d_{s_h}) <\rf(d)$.
Observe that the set of conditionals ${\cal E}={\cal B} \cup \{d\} \backslash \{d_{s_1},\ldots, d_{s_t}\}$ is then more serious than ${\cal B}$ in the MP-ordering (${\cal B} \prec^{MP} {\cal E}$),
and ${\cal E}$  must be consistent with $C$ as, for all the $C$-justifications ${\cal J}_{s_h}$, a conditional ($d_{s_h}$) is not in ${\cal E}$.
This contradicts the assumption that ${\cal B}$ (being a basis for $C$) is a maximally MP-serious set of conditionals in $K$ consistent with $C$.

Hence, there must be a $C$-justification ${\cal J}_{s_h}$ (for some $h$) such that  $d \in {\cal J}_{s_h}^{min}$. 
As a consequence, the Minimal Relevant Closure algorithm must remove $d$ from the set of conditionals in $K$ at the iteration  stage $i=\rf(d)$.
All the conditionals with rank $\rf(d)$ are removed from $K$ at that stage as no-other conditional in ${\cal J}_{s_h}$ has been removed in advance
(there is none in ${\cal J}_{s_h}$ with rank lower than $i$). 

As conditional $d$ is removed at some iteration, $d$ is not in ${\cal D}'$  the  set of defaults resulting from the execution of the Minimal Relevant Closure algorithm. But, as our choice of $d$ is arbitrary, this holds for all $d\not \in {\cal B}$. Therefore, ${\cal D}' \subseteq  {\cal B}$, that is, the set of conditionals ${\cal D}'$ computed by the Minimal Relevant Closure algorithm,  for the query $C\ent A$, is a subset of the MP-basis $ {\cal B}$ for $C$.

As mentioned above, we can now conclude that $C\ent A$ is not in the Minimal Relevant Closure of $K$.
As we know that 
 $\tilde{ {\cal B}} \cup \{C\}  \not \models A$, then $\not \models \bigwedge \tilde{ {\cal B}} \wedge \{C\}  \ri A$.
As ${\cal D}' \subseteq  {\cal B}$, then
$\not \models \bigwedge \tilde{ {\cal D}'} \wedge \{C\}  \ri A$. Thus, $C\ent A$ is not in the Minimal Relevant Closure of $K$.
\hfill $\Box$
\end{proof}
Basic Relevant closure is known \cite{Casini2014} to be weaker than Minimal Relevant Closure and, therefore, it is also weaker than the MP-closure. 

As for the MP-closure and the lexicographic closure, which may have an exponential number of bases (in the size of the knowledge base), computing the Relevant closure may require as well an exponential number of classical entailment checks \cite{Casini2014}.
Concerning the properties of the relevant closure it was shown by Casini et al. \cite{Casini2014}  that both Basic  Relevant Closure and Minimal Relevant Closure, among the properties of a rational inference relation, do not satisfy Or, Cautious Monotonicity and Rational Monotonicity.


\subsection{Comparisons with Basic Preference Descriptions and with System ARS} \label{sec: confronti_Brewka_ARS}

In this section, we compare the MP-closure with Brewka's basic preference descriptions \cite{Brewka04} and with system ARS  \cite{IsbernerRitterskamp2010} a refinement of System Z introduced by Kern-Isberner and Ritterskamp, using techniques for handling preference fusion.

We first show that the MP-closure semantics can be defined within the framework of basic preference descriptions, which allows the combination of different preference strategies for defining preferences among models, for a ranked knowledge base.
A ranked knowledge base (RKB), also called a stratified knowledge base, is a set of formulas $F$ together with a total preorder $\leq$ on $F$  \cite{Brewka04}.
It can be represented as a set of pairs $K=\{(f_i,r_i)\}$, where the rank $r_i$ of formula $f_i$ is a non-negative integer such that $f_i\geq f_j$ iff $r_i \geq r_j$.
A basic preference description associates with an RKB $K$ one preference strategy out of four, where each preference strategy defines an ordering on the set of possible worlds.

For a propositional interpretation $m$, let $K^i(m) = \{f \mid (f,i) \in K \mbox{ and } m \models f\}$, i.e., $K^i(m)$ is the set of formulas with rank $i$ which are satisfied by $m$. The subset strategy $\leq_\subseteq$ is defined as follows:
\begin{align*}
m_1 & \leq_\subseteq  m_2  \mbox{ ($m_1$ is preferred to $m_2$)  \ \ iff \ \ }   \\
&  K^i(m_1) = K^i(m_2) \mbox{ for all $i$, or} \\
& \mbox{there is an $i$ such that $K^i(m_1) \supset K^i(m_2)$ and,  for all $j>i$, $K^j(m_1) = K^j(m_2)$} 
\end{align*}
A strict preference relation $<_\subseteq$ can be defined as usual: $m_1 <_\subseteq m_2$ iff $m_1 \leq_\subseteq m_2$ and  not $m_2 \leq_\subseteq m_1$. 

Given a knowledge base $K$ (a set of conditionals), it is possible to associate with $K$ a basic preference description $K^{\subseteq}$ so that 
the associated preference relation $<_\subseteq$ is equivalent to preference relation among worlds used in the definition of MP-models (Definitions \ref{def:functor} and \ref{defi:MP-models}).
More precisely, we let $K^{\subseteq}=\{(f_d,r_d)_{d \in K}\}^{\subseteq}$, where  each formula $f_d=\alpha \ri \beta$ is the materialization of a default $d=\alpha \ent \beta$ in $ K$ and the rank $r_d= \rf(d)$ is the rank of the default in the rational closure.
   
To show that 
the strict preference relation $<_\subseteq$ induced by $K^{\subseteq}$ is equivalent to the preference relation among worlds used in the definition of MP-models (Definitions \ref{def:functor} and \ref{defi:MP-models}), we prove the following proposition.
\begin{proposition}
$m_1 <_\subseteq m_2$ iff $V(m_1) \prec^{MP} V(m_2)$.
\end{proposition}
\begin{proof}
For the ``only if" direction, suppose $m_1 <_\subseteq m_2$ holds. Then $m_1 \leq_\subseteq m_2$ holds and $m_2 \leq_\subseteq m_1$ doesn't.
As $m_2 \leq_\subseteq m_1$ does not hold, in particular, $K^i(m_2) \neq K^i(m_1)$, for some $i$.
This also means that the defaults violated by $m_2$ and $m_1$ cannot be the same, i.e., $V(m_1) \neq V(m_2)$
(remember that $V(x)$ is the set of defaults violated by $x$).
From the fact that $K^i(m_2) \neq K^i(m_1)$ for some $i$, and that $m_1 \leq_\subseteq m_2$, we know that 
there must be a rank $i$ such that $K^i(m_1) \supset K^i(m_2)$ and, for all $j>i$, $K^j(m_1) = K^j(m_2)$. 

Letting $V_i(x)$ be the defaults with rank $i$ violated by $x$,
we get $V_i(m_1) \subset V_i(m_2)$ and, for all $j>i$, $V_j(m_1) = V_j(m_2)$.
As there is some $i$ such that $V_i(m_1) \subset V_i(m_2)$ and, for all $j>i$, $V_j(m_1) = V_j(m_2)$,
the tuple associated with $V(m_1)$ is less serious (in the MP-seriousness ordering, Definition \ref{MP-order}) than 
the  tuple associated with $V(m_2)$, that is 
$V(m_1) \prec^{MP} V(m_2)$.

For the ``if" direction, suppose $V(m_1) \prec^{MP} V(m_2)$.
Then for some rank $i \leq k$, $V_i(m_1) \subset V_i(m_2)$, and for all ranks $j>i$, $V_j(m_1) = V_j(m_2)$.
Then it must be that $K^i(m_1) \supset K^i(m_2)$ and, for all $j>i$, $K^j(m_1) = K^j(m_2)$.

From this, it follows that $m_1 \leq_\subseteq  m_2$. 
To show that $m_2 \leq_\subseteq  m_1$ does not hold, considering that, for rank $i$,  $K^i(m_1) \neq K^i(m_2)$,
it is enough to show that it is not the case that, for some rank $h$,
$K^h(m_2) \supset K^h(m_1)$ and,  for all $j>h$, $K^j(m_1) = K^j(m_2)$.
This is not possible as we have already proved that $i$ is the highest rank such that $K^i(m_1) \neq K^i(m_2)$
and that $K^i(m_1) \supset K^i(m_2)$. Hence,  $m_2 \leq_\subseteq  m_1$ does not hold and $m_1 \leq_\subseteq  m_2$ holds,
and the conclusion $m_1 <_\subseteq  m_2$ follows. 
\hfill $\Box$
\end{proof}

In \cite{IsbernerRitterskamp2010},
Kern-Isberner and Ritterskamp apply techniques from preference fusion to the total preorders induced by conditional knowledge bases to define an inference relation, called {\em system ARS},  which refines and improves system Z.  They view ``inference as a decision making problem establishing which worlds are more plausible than others". Starting from the priority relation between defaults given by system Z (and by the rational closure) and using the preference fusion approach, they obtain a system which is capable of dealing with irrelevant information. A farther refinement of system ARS, called {\em system $\mathit{ARS^\#}$}, reproduces lexicographic entailment.

Starting from an ordered partition $\Delta = \Delta_0 \cup \Delta_1 \ldots \cup \Delta_k $ of the set of defaults $\Delta$, which is determined as in system Z, where $\Delta_k$ is the most specific layer, a total preorder $\leq_i$ on worlds is associated with each layer $\Delta_i$, where $w \leq_i w'$ informally means that if $w$ violates some default in $\Delta_i$, $w'$ does as well. Fusion operations are iteratively applied, 
starting from layer $\Delta_0$, to fuse $\leq_0$ with $\leq_1$ (with priority to $\leq_1$), then to fuse the resulting relation with $\leq_2$ (with priority to $\leq_2$), and so on up to $\leq_k$, using the fusion ARS operator {\em but} to give priority to the more specific defaults. 
The resulting total preorder on the set of worlds is used to define system ARS.

Kern-Isberner and Ritterskamp have shown that system ARS solves the problem of irrelevance, but not the drowning problem, while the stronger system $\mathit{ARS^\#}$, which takes into consideration the number of defaults falsified in each $\Delta_i$,  deals with the drowning problem and reproduces the lexicographic entailment.
\cite{IsbernerRitterskamp2010} also provides an equivalent formulation of the preference among worlds in system  ARS using Brewka's basic preference descriptions. 
We report it in the following, to clarify the relationships between the MP-closure and system ARS.

For system ARS, 
given a partition of the set of defaults $\Delta = \Delta_0 \cup \Delta_1 \ldots \cup \Delta_k $, an associated  RKB  $K^{\subseteq}$ can be defined containing, for each default layer $\Delta_i$, 
a single formula $f_i= \bigwedge_{A \ent B \in \Delta_i} (A  \ri B)$, the conjunction of the materializations of all the defaults in $\Delta_i$, with the associated rank $r_i=i$.
It has been proved in \cite{IsbernerRitterskamp2010} that the ranking of worlds in system ARS is equivalent to the ranking of worlds induced by the basic preference relation $K^{\subseteq}$.

The main difference between system ARS and MP-closure is that system ARS does not consider the single defaults in each layer $\Delta_i$ separately, while the MP-closure semantics does.
For instance, in Example \ref{example-Student}, with $K$ containing the conditionals:
\begin{quote}
 1. $\mathit{ Student \ent \neg Pay\_Taxes}$\\
 2. $\mathit{Student \ent  Young}$\\
 3. $\mathit{ Employee \wedge Student \ent Pay\_Taxes}$
\end{quote}
the world  $\mathit{x=\{ Employee, Student,  Pay\_Taxes, Young \}}$ in which employed students pay taxes and are young is preferred in the MP-closure semantics to the world $\mathit{y=}$ $\mathit{\{ Employee, Student,  Pay\_Taxes, \neg Young \}}$ in which employed students pay taxes and are not young,
as the first one satisfies default  $ \mathit{ Student  \ent Young}$ (having rank 0), while the second doesn't (and both of them falsify the default $ \mathit{ Student  \ent \neg Pay\_Taxes}$ with rank $0$). 
In system ARS the two worlds $x$ and $y$ are not comparable, as they both falsify some default in $\Delta_0$, and the conditional $\mathit{ Employee \wedge Student \ent Young}$ does not belong to ARS inference relation (while it belongs to the MP-closure of $K$). 

Hence, system ARS is not stronger than the MP-closure. One can see that it is also not weaker.
In fact, in system ARS, worlds satisfying all defaults in some layer $\Delta_i$  may be preferred 
to worlds satisfying some, but not all, defaults in a layer $\Delta_j$, with $j>i$. 
This may lead to inferences which are not allowed in the MP-closure.

\normalcolor

\subsection{Comparisons with other preferential approaches} \label{sec:Cond_Ent}


In this section, we compare with  Geffner and Pearl's {\em Conditional Entailment} \cite{GeffnerAIJ1992} and with $ \alctm$ \cite{LPAR2007,AIJ13} a  typicality based extension of the description logic $\alc$ introduced by Giordano et al.
These approaches exploit preferential interpretations, without restricting to ranked models, and have both some relationships with circumscription 
based approaches.
They define consequence relations that do not satisfy rational monotonicity but, in contrast to rational closure and its stronger refinements, do not have the problem that ``conflicts among defaults that should remain unresolved, are resolved anomalously" \cite{GeffnerAIJ1992}. 
At the end of the section we also compare with Weydert's system JLZ \cite{Weydert03}, with system LCD \cite{Benferhat2000} by Benferhat at al.
and other preferential semantics.

Let us consider 
 the following example by Geffner and Pearl:

\begin{example}  \label{example-Pearl_Geffner}
Let $K$  be the knowledge base containing the conditionals:
\begin{quote}
1. $\mathit{ p \wedge s \ent q}$\\ 
2. $\mathit{r \ent \neg q}$
\end{quote}
in the rational closure (and its refinements), as well as in conditional entailment and in $\alctmin$, 
one can neither conclude  $\mathit{ p \wedge s \wedge r \ent q}$ nor  $\mathit{ p \wedge s \wedge r \ent \neg q}$.
In particular, in the rational closure both defaults 1 and 2 have rank 0, while the formula $\mathit{ p \wedge s \wedge r}$ has rank 1,
and $\rf(\mathit{ p \wedge s \wedge r \wedge q})$  $=\rf(\mathit{ p \wedge s \wedge r \wedge \neg q})=1$.
Similarly, in the MP-closure and in the lexicographic closure, there are two bases for $\mathit{ p \wedge s \wedge r}$, $B=\{1\} $
and $D=\{2\}$, and  neither $q$  belongs to all the bases, nor  $\neg q$ does. That is, the conflict among defaults 1 and 2 is unresolved.

However,  if we add to $K$ an additional default 3. $p \ent \neg q$,  defaults 2 and 3 get rank 0 in the rational closure, while default 1 gets rank 1, and the same for the formula $\mathit{ p \wedge s \wedge q \wedge r}$. 
As $\rf(\mathit{ p \wedge s \wedge r \wedge q})=1$  $<\rf(\mathit{ p \wedge s \wedge r \wedge \neg q})=2$, 
$\mathit{ p \wedge s \wedge r \ent q}$ follows.
The conflict is resolved in favor of $q$, but this is not really justified (and we could add further defaults to rise the rank of $r$ in the rational closure to resolve the conflict in favor of $\neg p$).
The same conclusion follows from the MP-closure and from the lexicographic closure as they define consequence relations stronger than the rational closure.
Conditional entailment 
on the contrary does not build on the rational closure ranking and leaves the conflict unresolved after the introduction of default 3. 
\end{example}
The problem above, which is related to the representation of preferences as levels of reliability, has also been recognized by Brewka in his logical framework for default reasoning \cite{Brewka89}, 
leading to a generalization of the approach to allow a partial ordering between premises. 



Geffner and Pearl introduce the notion of conditional entailment  based on a class of preferential structures called {\em prioritized structures}.
They consider, for each default $p(x) \ri q(x)$ in the knowledge base, a sentence $p(x) \wedge \delta_i(x) \Ri q(x)$\footnote{$p(x) \wedge \delta_i(x) \Ri q(x)$ stands for  the closed formula $\forall x(p(x) \wedge \delta_i(x) \Ri q(x))$}, as well as a default $p(x) \ri \delta_i(x)$, where $ \delta_i$ denotes a new and unique assumption predicate which summarizes the normality conditions required for concluding $q(x)$ from $p(x)$.
They refer to the collection of defaults violated by a model $M$ in a structure as a {\em gap} of the model, denoted by $D[M]$, and define a prioritized preferential structure as follows:

\begin{definition}[\cite{GeffnerAIJ1992}] \label{def:prioritized_structure}
A {\em prioritized preferential structure} is a quadruple $\la { \cal G_{L}}, <,{\cal D_L}, \linebreak \prec \ra$, where
${\cal G_{L}}$ stands for the set of interpretations over the underlying language ${\cal L}$,
${\cal D_L}$ stands for the set of assumptions in ${\cal L}$, $\prec$ for an irreflexive and transitive priority relation over ${\cal D_L}$,
and $<$ is a binary relation over ${\cal G_{L}}$, such that for two interpretations $M$ and $M'$, $M<M'$ holds iff 
$D[M] \neq D[M']$, and for every assumption $\delta$ in $D[M] - D[M']$ there exists an assumption $\delta'$ in $D[M'] - D[M]$
such that $d \prec d'$.
\end{definition}
This means that when $M<M'$, if $M$ violates some default which is satisfied by $M'$, then there is a default $d'$ with higher priority
which is violated by $M'$ and not by $M$.

To guarantee that the priority ordering $\prec$ (which must not contain infinite ascending chains) reflects the structure of the knowledge base, Geffner and Pearl develop a notion of admissible priority ordering and define admissible  prioritized structures, i.e. structures where $\prec$ is admissible. While we will not recall here the definition of an admissible priority ordering $\prec$,
let us observe that
the notion of specificity  used  in Section \ref{sec:iterating} in the definitions of S-Enriched rational models  (Definition \ref{def-SenrichedmodelR}) and of simplified-enriched models (Definition \ref{def-Simp-enrichedmodelR}), and formulated  through conditions $(i)$ and $(ii)$, is strongly related with the conditions above defining the preference $M <M'$ among two interpretations in prioritized preferential structures. 
There are few major differences, as the preference relation $<$ of an  S-enriched model $\emme$ in Definition \ref{def-SenrichedmodelR} is modular,
and its definition exploits the ranking function $k_\emme$ to determine the priority among defaults, while in conditional entailment all possible priority orderings $\prec$ which are admissible with $K$ are considered.
 The MP-closure semantics could be reformulated in a similar style (and the bi-preference interpretation developed for the MP-closure in DL setting \cite{arXiv_Skeptical_closure} is an example). 
As it exploits the rational 
closure ranking for determining the priority of defaults, and such a ranking is an admissible ranking, one may conjecture that 
the MP-closure defines a stronger notion of entailment with respect to conditional entailment. 

Conditional entailment has a sophisticated proof theory based on argumentation and
deals with irrelevance and the drowning problem. It has been shown to be weak in some respects  (e.g., it violates some form of defeasible specificity \cite{Benferhat2000}). 

As for $\alctmin$,
its semantics is based on preferential KLM semantics (but not on ranked interpretations)  and minimal entailment exploits a different notion of minimization with respect to the rational closure. 
The idea is that the conditionals $A \ent B$ are represented by strict inclusions $\tip(A) \sqsubseteq B$ in DL setting (meaning that the typical $A$ elements are $B$ elements), corresponding to material implications $\tip(A) \ri B$, where
 $\tip(A)$ is true at the minimal worlds $w$ satisfying $A$. 
Based on a modal interpretation of typicality, $\tip(A)$ can be interpreted as $A \wedge \neg \Box \neg A$, where $\Box$ is a G\"odel-L\"ob modality, whose accessibility relation is taken to be the inverse of $<$. 
Minimal interpretations are then defined as those minimizing the formulas $\neg \Box \neg A$ true at the worlds.

As conditional entailment, $\alctmin$ does not resolve conflicts anomalously. 
In Example  \ref{example-differenza_MC_LC}, it is less bold than lexicographic closure and (as the MP-closure) does not allow to conclude the conditional $\mathit{Employee \wedge Student  \ent  Young  \wedge}$ $\mathit{ \neg  Pay\_Taxes }$.
However, $\alctmin$ suffers from the blocking of property inheritance problem 
which can be easily explained as the minimization criterium does not consider each single default individually, as in conditional entailment (where an assumption $\delta_i$ is associated with each default), in the Lexicographic closure and in the MP-closure.

An extension of the typicality logic $\alctmin$ 
which deals with the  blocking of property inheritance problem 
has been developed by Fernandez Gil \cite{fernandez-gil}, based on the idea of considering subsets of the axioms in the knowledge base, by developing  a multi-typicality version of $\alctmin$ , which allows for different typicality operators $\tip_i$. 
As $\alctmin$ is based on a preferential extension of $\alc$  with typicality, but not on the rational closure, the resulting extension is neither weaker nor stronger than  MP-closure. 
Indeed, $\alctmin$ and the rational closure of $\alc$ with typicality are already incomparable \cite{AIJ15}.
A further difference between the MP-closure 
and the proposal in \cite{fernandez-gil} is that in the MP-closure we do not consider multiple typicality operators, but a single notion of ``normality" (or typicality) expressed by conditionals. 




System JLZ is a quasi-probabilistic default formalism for graded defaults introduced by Weydert \cite{Weydert03}, based on a canonical ranking construction.
We have already mentioned that, in Example \ref{exa:evidence_comparison}, system JLZ (as the lexicographic closure) allows for stronger conclusions than the MP-closure,
as it takes  into account the ``weight of independent reasons" for supporting some conclusion.
In other examples, however, such as in the presence of redundant conditionals, system JLZ is more cautious than the  MP-closure (and than the lexicographic closure).
Consider the following instance of the ``redundant shortcut" example  by Weydert  \cite{Weydert03}. From the knowledge base containing the conditionals $\{ \mathit{ Student \ent Adult}, \; \mathit{ Adult \ent Married},\;  \mathit{ Student \ent}$ $\mathit{ \neg Married}, \mathit{ Adult \ent }$ $\mathit{Responsible}$,  $\mathit{ Married \ent Responsible} \}$, system JLZ neither concludes that student are normally responsible nor its negation.
Vice-versa, $\mathit{ Student \ent Responsible}$ belongs to the MP-closure and to the lexicographic closure of the knowledge base.
Therefore system JLZ is neither weaker nor stronger than the MP-closure.

Benferhat et al. have developed a very general approach to deal with default information based on the theory of belief functions \cite{Benferhat2000}.
They prove that this approach allows a uniform characterization of several popular non-monotonic systems, and use $\epsilon$-belief assignments to build a new system, called LCD, which addresses correctly the well known problems of defeasible reasoning including specificity, irrelevance and blocking of inheritance.
In particular, LCD deals correctly with the problem of ambiguity preservation. LCD does not satisfy rationality monotonicity but, rather than a drawback, considers this ``as an indication that rational monotonicity does not necessarily apply to all situations" \cite{Benferhat2000}. While LCD gives rise to a stratification of rules in the knowledge base, this stratification may be different from the one produced by system Z and by the rational closure and for this reason is incomparable to other systems, such as Brewka's preferred subtheories \cite{Brewka89} and the lexicographic closure, as well as to the MP-closure.

 In \cite{CaisniJelia19} a systematic approach for extending the KLM framework for defeasible entailment has been developed by Casini et al., showing that the rational  closure and lexicographic closure fall within the refined framework, but that there are forms of defeasible entailment within the framework that are more ``adventurous" than (and bolder than) the lexicographic closure. A semantic characterization in terms of a class of ranked interpretations is developed. More precisely, a notion of {\em basic defeasible entailment relation} is defined which satisfies all KLM properties plus some additional ones (namely, Inclusion and Classic Preservation) and 
basic defeasible entailment relations are shown to be characterizable using ${\cal K}$-faithful rank functions. 
Furthermore, an algorithm is defined which computes the defeasible entailment relation generated by a  ${\cal K}$-faithful rank function. 
Starting from the observation that basic defeasible entailment is too permissive, 
rank preserving ${\cal K}$-faithful rank functions are introduced, which are required to be a refinement of the rational closure ranking, and the notion of 
{\em rational defeasible entailment relation} is defined as a basic defeasible entailment relations which is closed with respect to the rational closure consequences.
It is shown in \cite{CaisniJelia19} that rational defeasible entailment relations can be characterized as defeasible entailment relations generated by rank preserving  ${\cal K}$-faithful rank functions, and that the lexicographic closure is a rational defeasible entailment relation.
It has to be investigated whether the rational extension of the MP-closure introduced in Section  \ref{sec:rational_relation} also fits in this framework. 
We leave this investigation for future work.

To conclude this section, let us mention some recent work on the semantics of conditionals.

A general semantics for conditional knowledge bases has been developed by Kern-Isberner  based on {\em c-representations} \cite{Kern-Isberner01,Kern-IsbernerAMAI2004},
which are particular ranking functions using the principle of conditional preservation as their core construction mechanism. 
Several nonmonotonic inference relations have been proposed considering subsets of all c-representations based on various notions of minimality \cite{BeierleAMAI2018}.
In \cite{BeierleFLAIRS19}  Beierle et al. study the properties of skeptical, weakly skeptical and credulous c-inference and, in particular, they prove that skeptical inference over any notion of minimal c-representations does not satisfy Rational Monotonicity, but satisfies Weak Rational Monotonicity (and similarly for weakly skeptical inference over all c-representations and over any notion of minimal c-representations).

 In \cite{KoniecznyIJCAI19} Konieczny et al. have studied a new family of inference relations based on the selection of some maximal consistent subsets, leading to inference relations with a stronger inferential power than the basic relation based on all maximal consistent subsets of the knowledge base.
They define a general class of monotonic selection relations for comparing maximal consistent subsets and show that it corresponds precisely to the class of rational inference relations.

\normalcolor

\subsection{Comparisons with further approaches from DLs} \label{sec:confronti_con_DLs}

  There are other related approaches from DL literature that build on the rational closure, or deal with its limitations. 
 In the following we consider some of them, as well as some other related proposals.
 
${\cal DL}^N$ captures a form of  ``inheritance with overriding": a defeasible inclusion is inherited by a more specific class if it is not overridden by more specific (conflicting) properties. 
 ${\cal DL}^N$ is not necessarily defined starting from the ranking given by the rational closure but, when it does, it provides a possible approach  to address the problem of inheritance blocking in the rational closure.
Inference is based on a polynomial algorithm which allows a default property to be inherited. When a defeasible property of a concept is conflicting with another defeasible property, and none of them is more specific so to override the other, the concept may have an inconsistent prototype.
For instance, in Example \ref{example-Student-new}
 the concept $\mathit{Employee \wedge Student}$ has an inconsistent prototype, as employed students inherit  the property of  students of not paying taxes and the property of employee of paying taxes, none is more specific than the other. 
In such an example, as we have seen, the MP-closure and  the lexicographic closure only conclude  that employed students are busy, and
silently ignore the conflicting defaults.
In  ${\cal DL}^N$ unresolved conflicts have to be detected and then fixed by modifying the knowledge base. 
The logical properties of  ${\cal DL}^N$ are studied in \cite{BonattiSauro17}. It is shown that, when considering the internalized KLM postulates, where each inclusion $NC \sqsubseteq D$ corresponds to a conditional $C \ent D$, few of the postulates are satisfied (namely, Reflexivity, Left Logical Equivalence and Right Weakening) but, when only {\em N-free} knowledge bases are allowed (i.e., knowledge bases which do not allow normality concepts $NC$ on the r.h.s. of conditionals), all the postulates are satisfied, with the partial exception of Cautious Monotonicity.
That is, satisfying KLM properties in  ${\cal DL}^N$ comes at the price of  renouncing to the full expressiveness of the non monotonic DL (such as, supporting role restrictions to normal instances).


The inheritance-based rational closure by Casini and Straccia \cite{Casinistraccia2011,CasiniJAIR2013}
is a closure construction defined by combining the rational closure with defeasible inheritance networks. 
For answering a query "if A, normally B",
it relies on the idea that only the information related to the connection of  $A$ and $B$ (and, in particular, only the defeasible inclusions occurring on the routes connecting $A$ and $B$ in the corresponding net) are relevant and have to be considered in the rational closure construction for answering the query. 

 Recent work in Description Logic literature has considered extensions of the rational closure semantics under several respects. 

A notion of {\em stable rational closure} has been proposed by Bonatti \cite{Bonatti2019} to extend the rational closure to a wider class of expressive description logics, which do not satisfy the disjoint union model property and for which the consistency of rational closure is not guaranteed. Observing that the definition of exceptionality ranking in the rational closure is not adequate to deal with this case, a notion of stable ranking is introduced.
Stable rankings are not intended to deal with the blocking of property inheritance problem and they collapse to the classical ranking of the rational closure when the underlying logic satisfies the disjoint union model property.


A further problem of rational closure for DLs is that it disregards defeasible information for existential concepts, a problem which has been addressed
by Pensel and Turhan \cite{Pensel18}, who developed stronger versions of rational and relevant entailment in ${\cal EL}$, which considers defeasible information for quantified concepts.

An extension of DLs with defeasible role quantifiers and defeasible role inclusions has been developed by Britz and Varzinczak \cite{Britz2018,Britz2019}, by associating multiple preference relations with roles.

Casini et al.  \cite{CasiniDL19} propose a procedure that, given as input any possible ranking of the defeasible concept inclusions contained in the knowledge base,
allows to define a notion of entailment 
from the knowledge base. They show that, for any arbitrary ranking, there is a defeasible knowledge base whose Rational Closure is equivalent to that ranking.
This work relates to the KLM framework for defeasible entailment introduced by Casini et al. for the propositional case \cite{CaisniJelia19} that we have referred to in the previous section.

The rational closure construction for DLs has been exploited as well in the context of fuzzy logic \cite{CasiniStracciaLPAR2013} and probabilistic logics \cite{Lukasiewicz08}. 
A description logic with probabilistic conditionals $\alc^{ME}$ has been considered by Wilhelm and Kern-Isberner \cite{WilhelmIsberner19}, exploiting a reasoning methodology based on the principle of maximum entropy. The complexity of the consistency problem in $\alc^{ME}$ has also been studied  \cite{BaaderEKW19}. 

\normalcolor


Among recent approaches dealing with the problem of inheritance with exceptions in description logics,  
the one by Bozzato et al. 
presents a defeasible extension of the Context Knowledge Repositories (CKR) framework \cite{Bozzato2018} in which defeasible axioms are allowed in the global context 
and can be overridden by knowledge in a local context. Exceptions have to be justified in terms of semantic consequence.
A translation of extended CKRs (with  knowledge bases in ${\cal SROIQ}$-RL) into Datalog programs under the answer set semantics has also been developed.

   \section{Conclusions}  \label{sec:conclu} 

In this paper we have studied the notion of MP-closure in the propositional case. The MP-closure was originally introduced for description logics in  \cite{Multipref_arXiv2018,Ecsqaru19} as an approximation of multipreference semantics. 
As  lexicographic closure, MP-closure builds on  rational closure but it exploits a different seriousness ordering to compare sets of defaults:
a different lexicographic order is used which compares tuples of sets of defaults rather than tuples of numbers (the number of defaults in the sets).

Lehmann in his seminal work on lexicographic closure  \cite{Lehmann95} proposes that ``in case of contradictory defaults of the same rank, we try to satisfy as many as possible". 
Here, we have explored an alternative option of considering alternative sets of defaults with the same rank as equally serious, in spite of their cardinality.

In the paper, we have presented a characterization of  MP-closure both in terms of maxiconsistent sets and of a model-theoretic construction.
In particular, we have developed a simple preferential semantics for  MP-closure, 
introducing a functor ${\cal F}$ that maps each minimal canonical model of the knowledge base (characterizing the rational closure) into a preferential model,
thus defining, for a knowledge base $K$, a consequence relation ${\cal MP}_K$ which is a superset of the rational closure of $K$, ${\cal RC}_K$, but
a subset of the lexicographic closure, ${\cal LC}_K$.

${\cal MP}_K$ is not rational; however, we have seen that, starting from the preferential semantics of the MP-closure, a ranked semantics can be easily obtained 
to define a rational consequence relation (that we called ${\cal MP}^\Ra_K$),  which is a superset 
of the MP-closure. We have shown that ${\cal MP}^\Ra_K$  is incomparable with lexicographic closure, that is ${\cal MP}^\Ra_K$ neither includes ${\cal LC}_K$, nor is included in ${\cal LC}_K$.

We have also compared  MP-closure  with multipreference semantics  \cite{GliozziAIIA2016,Ecsqaru19} and with  Relevant Closure \cite{Casini2014}. 
They are both refinements of  rational closure as they build on rational closure ranking to define stronger consequence relations. 
These formalisms 
have been defined for description logics, but their definition can be transposed to  propositional logic.
Concerning the multipreference semantics, 
which is a semantic strengthening of rational closure, we have shown that, in the propositional setting,
there is an equivalence  relation among the $<_{FIMS}$-minimal canonical fixed-points of the operator ${\cal F}^\Ra$ (the operator used for defining a rational consequence relation extending the MP-closure)
and the minimal canonical simplified-enriched models in the multipreference semantics. 
Concerning Relevant Closure,
we have proven that MP-closure is stronger than both Basic Relevant Closure and  Minimal Basic Relevant Closure.


 MP-closure has also strong relationships with other non-monotonic formalisms based on the preferential approach, and specifically with those building on rational closure construction.
We have shown that MP-closure semantics 
can be equivalently formulated using Brewka's  Basic Preference Descriptions \cite{Brewka04} for ranked knowledge bases, starting from rational closure ranking and using the subset strategy.
It has also strong relationships with Kern-Isberner and Ritterskamp's {\em system ARS}  \cite{IsbernerRitterskamp2010}, which refines system Z, and with Geffner and Pearl's conditional entailment \cite{GeffnerAIJ1992}, which on the contrary does not build on  rational closure ranking
but exploits partial order priority relations among defaults.
%
Hence, 
we believe the MP-closure is a legitimate option 
among the different stronger and weaker semantics and constructions that have been proposed in the literature as  refinements or alternatives to  rational closure. Most of the paper is devoted to establishing such relationships. 
\normalcolor

A first semantic characterization of the MP-closure for the description logic $\alc$ was developed using bi-preferential (BP) interpretations \cite{arXiv_Skeptical_closure}, preferential interpretations developed along the lines of the preferential semantics introduced by Kraus, Lehmann and Magidor \cite{KrausLehmannMagidor:90,whatdoes}, but containing two preference relations, the first one $<_1$ playing the role of the ranked preference relations in the models of the RC, and the second one $<_2$ representing a preferential refinement of $<_1$. 
Another construction, developed for DLs, the skeptical closure \cite{Pruv2018}, was shown to be  a weaker variant of the MP-closure 
in  \cite{arXiv_Skeptical_closure}, which requires to build a single base, and we refer thereto for detailed comparisons.

The idea of considering subsets of the axioms in the knowledge base has also been considered by Fernandez Gil in \cite{fernandez-gil},
who developed a multi-typicality version of the typicality logic $\alctmin$ \cite{AIJ13}, 
which, as recalled in Section \ref{sec:iterating}, is another defeasible description logic based on a preferential extension of $\alc$ with typicality,  which, differently from the rational closure, is not based on a ranked semantics but,
nevertheless, suffers from the blocking of property inheritance problem.

There are several issues that are worth being investigated for future work. 
One is to verify whether these refinements of the rational closure can be characterized within the KLM framework for defeasible entailment developed by Casini et al. \cite{CaisniJelia19}.
Another issue is the extension of these constructions to the first-order case. 
In this regard, Description Logics may provide an interesting special case, which encompasses a limited treatment of non-monotonic reasoning in first-order logic, namely, the treatment of the decidable fragment including only unary and binary predicates. Some work in this direction has been already done, for instance,  
 by Bonatti in his extension of the rational closure to a wide class of expressive description logics \cite{Bonatti2019}, by Britz and Varzinczak by allowing in the language defeasible quantifiers and role inclusions besides defeasible subsumptions  \cite{Britz2018,Britz2019}, and  by Pensel and Turhan by developing versions of rational and relevant entailment in the low complexity description logic ${\cal EL}$, which considers defeasible information for quantified concepts \cite{Pensel18}.
In another direction, new alternative interpretations of conditionals also deserve being investigated. For instance, Koutras et al. \cite{KoutrasKR18} have proposed a  majority interpretation of conditionals (via a "most" generalized quantifier) which does not satisfy some of the KLM postulates, such as CM and AND. 
We have already mentioned in Section  \ref{sec:confronti_con_DLs} the minimal model semantics based on c-representations whose properties have been investigated by Beierle et al. \cite{BeierleFLAIRS19}, and the family of inference relations based on use of a scoring function studied by Konieczny et al. \cite{KoniecznyIJCAI19}.
%
As a final point, as in the MP-closure and in the lexicographic closure the number of possible bases for a given formula may be exponential in the number of defaults
(and, in the Relevant Closure, an exponential number of justifications is to be computed in the worst case),
from a practical point of view,  
to avoid an exponential blowup in the complexity, it may be essential to consider sound approximations of these constructions.
\normalcolor

\normalcolor

\medskip
{\bf Acknowledgement:} 
This research has been partially supported by INDAM-GNCS 
Project 2019  
and Project 2020.

\bibliographystyle{elsarticle-harv}

\end{document}